\documentclass[11pt]{article}
\usepackage[utf8]{inputenc}
\usepackage[T1]{fontenc}
\usepackage{tgpagella}

\usepackage[letterpaper, portrait, margin=1in]{geometry}
\usepackage[linktocpage=true,colorlinks=true,linkcolor=blue,citecolor=ForestGreen]{hyperref}
\usepackage{algorithm}
\usepackage[noend]{algpseudocode}
\usepackage{url}
\usepackage{amsmath,amssymb,amsthm}
\usepackage[noabbrev,capitalise,nameinlink]{cleveref}
\usepackage{mathtools}
\usepackage{parskip}
\usepackage{xspace}
\usepackage{verbatim}
\usepackage{mathrsfs}
\usepackage[usenames,dvipsnames,svgnames,table]{xcolor}
\usepackage{pgf}
\usepackage[dvipsnames]{xcolor}
\usepackage{todo}
\usepackage{tabularx}

\usepackage{dsfont}

\crefname{equation}{equation}{equations}
\Crefname{equation}{Equation}{Equations}

\theoremstyle{plain}
\newtheorem{theorem}{Theorem}[section]
\newtheorem{lemma}[theorem]{Lemma}
\newtheorem{fact}[theorem]{Fact}
\newtheorem{proposition}[theorem]{Proposition}
\newtheorem{claim}[theorem]{Claim}
\newtheorem{corollary}[theorem]{Corollary}
\newtheorem{question}[theorem]{Question}

\theoremstyle{definition}
\newtheorem{definition}[theorem]{Definition}

\newtheorem{example}[theorem]{Example}

\theoremstyle{plain}

\newcommand{\ignore}[1]{}

\DeclareMathOperator{\poly}{poly}

\newcommand{\dist}{\mathsf{dist}}

\newcommand{\Ex}[1]{\bE \left[ #1 \right]}
\newcommand{\Exu}[2]{\underset{#1} \bE \left[ #2 \right] }
\newcommand{\Exuc}[3]{\underset{#1} \bE \left[\left. #2 \; \right| \; #3
\right] }

\renewcommand{\Pr}[1]{\bP \left[ #1 \right]} 
\newcommand{\Pru}[2]{\underset{ #1 }\bP \left[ #2 \right]}
\newcommand{\Pruc}[3]{\underset{ #1 }\bP \left[\left. #2 \; \right| \; #3 \right]}

\newcommand{\define}{\vcentcolon=}

\newcommand{\inn}[1]{\langle #1 \rangle}

\newcommand{\ind}[1]{\mathds{1} \left[ #1 \right] }

\newcommand{\zo}{\{0,1\}}
\newcommand{\pmset}{\{\pm 1\}}

\newcommand{\cD}{\ensuremath{\mathcal{D}}}
\newcommand{\cF}{\ensuremath{\mathcal{F}}}

\newcommand{\cI}{\ensuremath{\mathcal{I}}}

\newcommand{\cX}{\ensuremath{\mathcal{X}}}

\newcommand{\bE}{\ensuremath{\mathbb{E}}}
\newcommand{\bP}{\ensuremath{\mathbb{P}}}

\usepackage{bm}

\hypersetup{colorlinks = true,linkcolor = blue, anchorcolor = blue, citecolor = blue, filecolor = red,urlcolor = red}
\usepackage [english]{babel}
\usepackage [autostyle, english = american]{csquotes}
\usepackage{blindtext}
\usepackage{framed}
\usepackage[shortlabels]{enumitem}
\usepackage{graphicx}
\usepackage{subcaption}
\usepackage{subfiles}
\usepackage{thmtools}
\usepackage{thm-restate}

\newtoggle{anonymous}
\togglefalse{anonymous}
\toggletrue{anonymous}

\MakeOuterQuote{"}

\makeatletter
\newcommand\footnoteref[1]{\protected@xdef\@thefnmark{\ref{#1}}\@footnotemark}
\makeatother

\newcommand{\norm}[1]{\left\lVert#1\right\rVert}
\newcommand{\Expect}{{\rm I\kern-.3em E}}

\newcommand{\zpm}{\{0,\pm 1\}}
\newcommand{\upset}{\mathsf{Up}}

\newcommand{\R}{\mathbb R}
\newcommand{\N}{\mathbb N}

\newcommand{\eps}{\varepsilon}

\newcommand{\sgn}{\mathrm{sgn}}

\newcommand{\bt}{t}
\newcommand{\bv}{v}
\newcommand{\bx}{x}
\newcommand{\by}{y}
\newcommand{\bz}{z}

\newcommand{\bQ}{Q}
\newcommand{\bT}{T}

\newcommand{\bV}{V}

\newcommand{\bS}{S}
\newcommand{\bX}{X}

\newcommand{\NN}{\mathbb{N}}
\newcommand{\RR}{\mathbb{R}}

\renewcommand{\norm}[1]{\left\lVert #1 \right\rVert}
\newcommand{\Exp}{\EX}

\newcommand{\EX}{\mathbb{E}}

\newcommand{\inlayers}{\mathsf{Inn}} 
\newcommand{\midlayers}{\mathsf{Mid}} 
\newcommand{\outlayers}{\mathsf{Out}}

\newcommand{\conv}{\mathsf{Conv}}

\newcommand{\pr}{\mathbb{P}}
\newcommand{\slab}{\mathsf{Slab}}
\newcommand{\tas}{\mathsf{TAS}}

\newcommand{\Inf}{\mathbb{I}}

\newcommand{\Dyes}{\mathcal{D}_{\mathsf{yes}}}
\newcommand{\Dno}{\mathcal{D}_{\mathsf{no}}}

\newcommand{\Variance}{\mathbf{Var}}

\usepackage[font={small,it}]{caption}

\begin{document}

\title{Testing and Learning Convex Sets in the Ternary Hypercube\footnote{This paper
is an updated and significantly expanded version of an
earlier preprint \cite{BBH23}.}}

\author{Hadley Black\thanks{Supported by NSF award AF:Small 2007682, NSF Award: Collaborative Research Encore 2217033}\\
University of California, Los Angeles\\
{\tt hablack@cs.ucla.edu}
\and
Eric Blais\thanks{Supported by an Ontario Early Researcher Award and an NSERC Discovery grant.} \\
University of Waterloo\\
{\tt eric.blais@uwaterloo.ca}
\and
Nathaniel Harms\thanks{Some of this work was done while the author was a student at the University of Waterloo. Partly supported by an NSERC Graduate Scholarship, an NSERC Postdoctoral Fellowship, and the Swiss State Secretariat for Education, Research and Innovation (SERI) under contract number MB22.00026.} \\
EPFL\\
{\tt nathaniel.harms@epfl.ch}
}

\maketitle
\date{}

\begin{abstract}
    We study the problems of testing and learning high-dimensional discrete convex sets. The simplest high-dimensional discrete domain where convexity is a non-trivial property is the \emph{ternary hypercube}, $\{-1,0,1\}^n$. The goal of this work is to understand structural combinatorial properties of convex sets in this domain and to determine the complexity of the testing and learning problems. We obtain the following results.
    
    \emph{Structural:} We prove nearly tight bounds on the edge boundary of convex sets in $\{0,\pm 1\}^n$, showing that the maximum edge boundary of a convex set is $\widetilde \Theta(n^{3/4}) \cdot 3^n$, or equivalently that every convex set has influence $\widetilde{O}(n^{3/4})$ and a convex set exists with influence $\Omega(n^{3/4})$.
    
    \emph{Learning and sample-based testing:} We prove upper and lower bounds of $3^{\widetilde{O}(n^{3/4})}$ and $3^{\Omega(\sqrt{n})}$ for the task of learning convex sets under the uniform distribution from random examples. The analysis of the learning algorithm relies on our upper bound on the influence. Both the upper and lower bound also hold for the problem of \emph{sample-based testing} with two-sided error. For sample-based testing with one-sided error we show that the sample-complexity is $3^{\Theta(n)}$.
    
    \emph{Testing with queries:} We prove nearly matching upper and lower bounds of $3^{\widetilde{\Theta}(\sqrt{n})}$ for one-sided error testing of convex sets with \emph{non-adaptive queries}. 
\end{abstract}

\thispagestyle{empty}
\newpage

\tableofcontents
\thispagestyle{empty}
\setcounter{page}{0}
\newpage

\section{Introduction}

A subset $S \subseteq [m]^n$ of the
hypergrid is \emph{discrete convex} if it is the intersection of a convex set $C \subseteq \R^n$ with the grid, $S = C \cap [m]^n$, or equivalently if $S = [m]^n \cap \conv(S)$ where $\conv(S)$ is the convex hull of $S$. Discrete convex sets may not even be \emph{connected} (see \cref{fig:convex_example}), which, along with some of their other unpleasant features, makes them difficult to handle algorithmically and analytically, the most famous example
being the difference between linear programming and integer linear programming.

We are interested in testing and learning discrete convex sets. A learning algorithm should output an approximation
of an unknown convex set $S$ by using membership queries to $S$, while a testing algorithm should decide
whether an unknown set $S$ is either convex or $\epsilon$-\emph{far} from convex, meaning that
$\dist(S,T) > \epsilon$ for all convex sets $T$, where $\dist(S, T)$ is the measure of the symmetric
difference.

Convexity is particularly interesting for property testing because it can be defined by a local condition: a set $S \subseteq \R^n$ is convex if and only if for every 3 colinear points $x,y,z$, if $x,z \in S$ then $y \in S$. This means that, to certify the non-convexity of a (continuous) set, it suffices to provide 3 colinear points that violate this condition. Speaking informally, property testing results, especially testing with \emph{one-sided error}, are statements about the difficulty of finding such a certificate of non-membership to the property, when the object $S$ is $\epsilon$-far from satisfying the property. But, the fact that convexity is \emph{defined} by a local condition does not make it easy to find violations of the condition when a set is far from convex. This is particularly evident for \emph{discrete} convex sets where, unlike continuous sets, there may not be \emph{any} lines which witness non-convexity, and one must instead look for up to $n+1$ points that violate Carath{\'e}odory's theorem.

We are aware of no non-trivial algorithms for testing or learning discrete convex sets in high dimensional grids $[m]^d$ when $m$ is small. Prior works on testing and learning convex sets include:
\begin{enumerate}
    \item The analysis of convexity testers, such as the \emph{line tester} and more general \emph{convex hull testers}, which are designed to simply ``spot-check'' for violations of the local conditions that define convexity in $\R^n$~\cite{RV04,BlaisB20}. These works show that these spot-checkers are not very efficient, requiring $2^{\Omega(n)}$ queries to detect sets that are $\Omega(1)$-far from convex.
    \item Testing or learning convex sets in two dimensions, including the continuous square $[0,1]^2$ \cite{Schmeltz92,BermanMR19b} or the discrete grid $[m]^2$ \cite{Raskhodnikova03,BermanMR19a,BermanMR22}.
    \item Testing convexity in high dimensions with \emph{samples}, either in the continuous setting \cite{CFSS17,HY21} or discrete setting \cite{HY21}, and learning convex sets from random examples of the set \cite{RG09} or from Gaussian samples \cite{KlivansOS08}.
\end{enumerate}
When $m \gg \poly(d)$, a ``downsampling'' or ``gridding'' approach  can
reduce to the case $m = \poly(d)$ \cite{CFSS17,HY21}, but once $m$ is small the only known algorithm for testing or learning is brute-force.
So let us see what happens when we make $m$ as small as possible. When $m=2$, testing and learning convex sets in $[m]^n \equiv \zo^n$ is trivial, because every subset of $\zo^d$ is convex and therefore testing is as easy as possible (the tester may simply accept on every input) and learning is as hard as possible (requiring $\Omega(2^n)$ queries).

The story changes significantly when $m=3$, so that $[m]^n$ is equivalent to the ternary hypercube $\zpm^n$, where the difficulties of handling high-dimensional discrete convex sets suddenly become evident. Although this
is the simplest domain where where high-dimensional discrete convex sets are non-trivial, little is known about the structure of discrete convex sets on the ternary hypercube that would help in designing testing and learning algorithms.
In this paper we will give the first results towards understanding testing and learning discrete sets in high dimensions by focusing on the ternary hypercube.

\begin{figure}[t]
    \centering
    \includegraphics{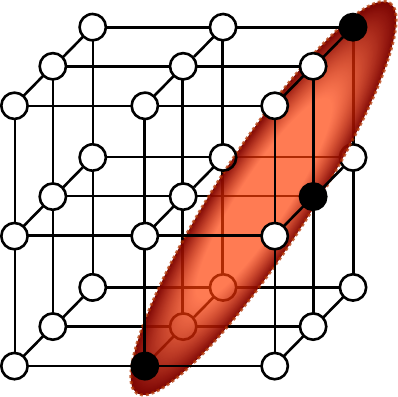}
    \caption{Example of a convex set in $\zpm^3$. The black dots are the set and the convex red ellipsoid contains them. Note that the set may not be ``connected'' on the hypergrid.}
    \label{fig:convex_example}
\end{figure}

\subsection{Results}

For two sets $S, T \subseteq \zpm^n$, we define $\dist(S,T) \define \frac{|S \Delta T|}{3^n}$, where $S \Delta T$ denotes the symmetric difference.
A set $S \subseteq \zpm^n$ is $\varepsilon$-far from convex if for every (discrete) convex set $T \subseteq \zpm^n$, $\dist(S,T) \geq \varepsilon$. Given $\varepsilon > 0$, a \emph{convexity tester} is a randomized algorithm which is given membership oracle access to an input $S \subseteq \zpm^n$ and must satisfy
\begin{enumerate}
    \item If $S$ is convex then the algorithm accepts with probability at least $2/3$.
    \item If $S$ is $\varepsilon$-far from convex then the algorithm rejects with probability at least $2/3$.
\end{enumerate}
The tester is \emph{one-sided} if it must accept convex sets $S$ with probability
1 instead of $2/3$. A tester is \emph{non-adaptive} if it chooses its set of queries before receiving the answers to any of the queries and it is \emph{sample-based} if its queries are independently and uniformly random.

A learning algorithm is given membership oracle access to a convex set $S \subseteq \zpm^n$ and must output (with probability at least $2/3$) a set $T \subseteq \zpm^n$ with $\dist(S,T) < \varepsilon$; it is \emph{proper} if its output $T$ must be convex.

\subsubsection{The Edge Boundary and Influence of Convex Sets}

One of the most important things to know about a set is its \emph{edge boundary}.
The edge set of the ternary hypercube is defined as
\begin{align} \label{eq:edges}
    E = \left\{(x,y) \in (\zpm^n)^2 \colon \sum_{i=1}^n \left|x_i-y_i\right| = 1\right\} \text{.}
\end{align}
Observe that $|E| = 2n \cdot 3^{n-1}$. We will identify a set $S \subseteq \zpm^n$ with its
characteristic function and write $S(x) = 1$ if $x \in S$ and $S(x) = 0$ otherwise. An edge $(x,y)$
is on the boundary of $S$ if $S(x) \neq S(y)$. The influence of a set $S \subseteq \{0,\pm 1\}^n$ is
its normalized boundary size:
\begin{align} \label{eq:influence}
    \Inf(S) := \frac{1}{3^n} \cdot |\{(u,v) \in E \colon S(u) \neq S(v)\}| = \frac{2n}{3} \cdot \pr_{(u,v) \sim E}[S(u) \neq S(v)] \text{.}
\end{align}

Before we state our results, consider some examples. Two important classes of convex sets in
$\zpm^n$ are halfspaces and balls, which often have minimal ``boundary size'' in various
settings.

\begin{example}[Halfspaces]
A halfspace is a set $H = \{ x \in \zpm^n \;:\; \inn{v,x} < \tau \}$ where $v \in \R^n$ and $\tau
\in \R$. To maximize the influence, we want $\tau$ to be small, say $\tau = 0$, and we want $v
\approx \vec 1$. The probability that a random edge $(x,y)$ is on the boundary is at most the
probability that a uniformly random $x \sim \zpm^n$ satisfies $|\inn{\vec 1, x}| \leq 1$, and it is
not difficult to show that this is at most $O\left( \frac{1}{\sqrt n} \right)$, giving an estimate
of $O(\sqrt n)$ for the maximum influence of a halfspace.
\end{example}

\begin{example}[Balls]\label{ex:balls}
A ball is a set $B_r = \{ x \in \zpm^n \;:\; \|x\|_2^2 < r \}$ where $r \in \R$ is the radius. The
average (squared) norm $\Ex{ \|x\|_2^2 }$ for $x \sim \zpm^n$ is the same as the expected number of
nonzero coordinates of $x$, which is $\frac{2}{3} n$, so to maximize the edge boundary we think of
$r \approx \frac{2}{3} n$. Similar to above, the probability that $x \sim \zpm^n$ is close enough to
this threshold to find a boundary edge is $O\left( \frac{1}{\sqrt n} \right)$, again giving an
estimate of $O(\sqrt n)$ for the maximum influence.
\end{example}

Our first result shows that there are convex sets with significantly larger influence, which can be
obtained by taking $S$ to be the intersection of roughly $3^{\Theta(\sqrt n)}$ random halfspaces
with thresholds $\tau = \Theta(n^{3/4})$; we think of these sets as interpolating between the
halfspaces and the ball. Our construction is inspired by \cite{Kane14}, who showed bounds on the
influence of intersections random halfspaces on the hypercube $\zo^n$, and we note that
similar constructions also achieve maximal surface area under the Gaussian distribution on $\R^n$
\cite{Nazarov03}.

\begin{restatable}{theorem}{thminfLB} \label{thm:infLB}
There exists a convex set $S \subseteq \{0,\pm 1\}^n$ with influence $\Inf(S) = \Omega(n^{3/4})$.
\end{restatable}

Our main result on the influence of convex sets is that this construction is essentially optimal:
we show a matching upper bound (up to log factors) for any convex set in $\zpm^n$. Due to the
discrete nature of the domain, our proof of this theorem is significantly different from the
previous techniques that have been used to bound the surface area of convex sets in continuous
domains.

\begin{restatable}{theorem}{thminfbound}\label{thm:inf-bound}
If $S \subseteq \zpm^n$ is convex, then $\Inf(S) = O(n^{3/4} \log^{1/4} n)$.
\end{restatable}

\subsubsection{Sample-Based Learning and Testing}

As an application of our bounds on the influence, we show using standard Fourier analysis that any set $S \subseteq \zpm^n$ can be approximated with error $\eps$ by a polynomial of degree $\Inf(S)/\eps$. Using \cref{thm:inf-bound} and the "Low-Degree Algorithm" of Linial, Mansour, and Nisan \cite{LinialMN93} then gives us the following upper bound for learning.

\begin{restatable}{theorem}{thmlearner}
\label{thm:learner}
There is a uniform-distribution learning algorithm for convex sets in $\zpm^n$ which achieves error
at most $\eps$ with time and sample complexity $3^{\widetilde{O}(n^{3/4}/\eps)}$. The
$\widetilde{O}(\cdot)$ hides a factor of $\log^{1/4} n$.
\end{restatable}

A corollary of \cref{thm:learner} is that the same upper bound on the sample complexity holds for
sample-based testing, due to the testing-by-learning reduction (which is slightly non-standard because the learner is not \emph{proper}, see \Cref{sec:learning->testing}).

\begin{restatable}{corollary}{corUBtwosidedsamples}
\label{cor:UB-2sidedsamples}
There is a sample-based convexity tester for sets in $\zpm^n$ with sample complexity
$3^{\widetilde{O}(n^{3/4}/\eps)}$ where the $\widetilde{O}(\cdot)$ hides a factor of $\log^{1/4} n$.
\end{restatable}

To complement our upper bounds, we prove also a lower bound for sample-based testing. Here we remark
that one of our motivations for studying convex sets in $\zpm^n$ is their similarity (in an informal
sense) to \emph{monotone} functions on $\zo^n$; an analogy between monotone functions on $\zo^n$ and
convex sets in Gaussian space was proposed in \cite{DeNS22} and we are interested in this analogy
for \emph{discrete} convex sets. Our lower bound for sample-based testing discrete convex sets uses
a version of Talagrand's random DNFs, which were used previously to prove lower bounds for testing
monotonicity on $\zo^n$ \cite{BeBl16,Chen17}.

\begin{restatable}{theorem}{thmtwosidedsamplestestingLB} \label{thm:2sidedsamplestestingLB}
For sufficiently small constant $\eps > 0$, every sample-based convexity tester for sets in $\zpm^n$
has sample complexity $3^{\Omega(\sqrt{n})}$.
\end{restatable}

Again, the testing-by-learning reduction of \Cref{lem:learning->testing} implies that this lower bound also holds for learning. 

\begin{restatable}{corollary}{corlearningLB} \label{cor:learningLB}
For sufficiently small constant $\eps > 0$, sample-based learning convex sets in $\zpm^n$ requires
at least $3^{\Omega(\sqrt n)}$ samples.
\end{restatable}

\subsubsection{Non-Adaptive One-Sided Testing}

A convexity tester with one-sided error is one that finds a \emph{witness} of non-convexity with
probability at least $2/3$ when the tested set is $\eps$-far from convex. A convexity tester is
\emph{non-adaptive} if it must choose its  set of membership queries before receiving any of the
query results. Bounds on non-adaptive one-sided error testing therefore have a natural combinatorial
interpretation as bounds on the likelihood of blindly finding a witness of non-convexity in a random
substructure of the domain.

Our first result shows that there is a non-adaptive one-sided error tester with sub-exponential
query complexity $3^{o(n)}$. In contrast, a similar bound for the Gaussian setting is not yet known
to exist.

\begin{theorem} \label{thm:UB_nonadaptive_1sided}
For every $\eps > 0$, there is a non-adaptive convexity tester with one-sided error for sets in
$\zpm^n$ that has query complexity $3^{\widetilde{O}\big(\sqrt{n \ln 1/\eps}\big)}$ where the
$\widetilde{O}(\cdot)$ notation is hiding an extra $\ln n$ term.
\end{theorem}

Next, we show that \cref{thm:UB_nonadaptive_1sided} is essentially tight, in that the exponential
dependence on $\sqrt{n}$ in its bound is unavoidable.

\begin{theorem} \label{thm:LB_nonadaptive_1sided} For sufficiently small constant $\eps > 0$, every
non-adaptive convexity tester with one-sided error for sets in $\zpm^n$ has query complexity at
least $3^{\Omega(\sqrt{n})}$.
\end{theorem}

Our \cref{thm:2sidedsamplestestingLB} above showed that $3^{\Omega(\sqrt n)}$ is required for
\emph{sample-based} testing. For one-sided error testers, we can improve this lower bound to show
that non-adaptive testers are significantly more powerful than sample-based testers for one-sided
testing.

\begin{restatable}{theorem}{thmsamplesonesided}
\label{thm:samples_1sided}
For sufficiently small constant $\eps > 0$, sample-based convexity testing in $\zpm^n$ with
one-sided error requires $3^{\Theta(n)}$ samples.
\end{restatable}

This theorem also includes a matching upper bound.  The upper bound in \cref{thm:samples_1sided} is
trivial because a coupon-collector argument shows that one can learn any set $S \subseteq \zpm^n$
exactly using $O(n 3^n)$ samples. A slightly improved bound of $O\left(3^n \cdot
\tfrac{1}{\eps}\log(1/\eps)\right)$ also holds by a general upper bound on one-sided error testing
via the VC dimension \cite{BFH21}.

\subsection{Techniques}

The discrete nature of the ternary hypercube, in contrast to the continuity of the domains $\mathbb{R}^n$ or $[0,1]^n$, provides a new angle in the study of convexity which leads to the development of a new set of \emph{combinatorial} techniques and tools. 
In this section we give a brief overview of the techniques we use to prove each of our theorems. 

\subsubsection{The Edge Boundary and Influence of Convex Sets}

\paragraph{Influence Upper Bound:} Our proof of \Cref{thm:inf-bound}, which gives an upper bound on the edge boundary of a convex set, is accomplished by relating the number of boundary edges to the expected number of sign-changes of one-dimensional random processes. This is done by constructing a distribution $\cD$ over the edge-set $E$ of the ternary hypercube, such that (a) $\cD$ is ``close'' to the uniform distribution over $E$ and (b) the probability that a random edge drawn from $\cD$ is influential for our convex set $S \subseteq \zpm^n$ is equal to the expected number of sign-changes of a certain random process. This process is defined by considering a random walk $\bm{X}^{(0)},\ldots,\bm{X}^{(m)}$ of length $m \approx n^{1/2}$ where $\bm{X}^{(0)}$ is a random point from the middle layers of $\zpm^n$ and each $\bm{X}^{(s)}$ is obtained by flipping a random $0$-valued bit of $\bm{X}^{(s-1)}$ to a uniform random $\{\pm 1\}$-value; the process finally draws $\bm{s} \sim [m]$ uniformly at random and outputs the edge $(\bm{X}^{(\bm{s}-1)},\bm{X}^{(\bm{s})})$.

The crux of the argument is to bound the expected number of times this random walk enters and leaves the set $S$. Since $S$ is convex, it can be written as an intersection of halfspaces $S = H_1 \cap H_2 \cap \cdots \cap H_k$ of the form $H_i = \{x \in \zpm^n \colon \langle x,v^{(i)} \rangle < \tau_i \}$ where $v^{(i)} \in \mathbb{R}^n$ and $\tau_i \in \mathbb{R}$. For each halfspace $H_i$, we define a corresponding one-dimensional random walk $\bm{W}_i(s) = \langle \bm{X}^{(s)}, v^{(i)} \rangle - \tau_i$ and observe that the original random walk crosses the boundary of $H_i$ at step $s$ if and only if $\bm{W}_i$ \emph{changes sign} at step $s$. Then the number of times the walk $\bm{X}^{(0)}, \bm{X}^{(1)}, \dotsc$ crosses the boundary of $S = \bigcap_i H_i$ is the number of times the \emph{maximum} of the processes $\bm{M} = \max_i\bm{W}_i$ changes sign. Therefore, our goal is to bound the expected number of sign-changes for $\bm{M}$, which we accomplish by using Sparre Andersen's fluctuation theorem \cite{SA54} (as stated in \cite{BB23}) to relate this quantity to the number of sign-changes of a uniform random walk.

\paragraph{High-Influence Set Construction:} Our proof of \Cref{thm:infLB} is inspired by the proof of ~\cite[Theorem 2]{Kane14} which constructs a set in the Boolean hypercube $\{\pm 1\}^n$ with influence $\Omega(\sqrt{n \log k})$ by considering an intersection of $k$ random halfspaces each of which is at distance $\approx \sqrt{n \log k}$ from the origin. In particular, when $k \approx 2^{\sqrt{n}}$ the construction has influence $\approx n^{3/4}$ and when $k \approx 2^n$ the set has influence $\approx n$. On the ternary hypercube $\zpm^n$, the behaviour is different: here, halfspaces exhibit a ``density increment'' behaviour as their threshold moves away from the origin, which prevents the influence from increasing as $k$ grows past $2^{\sqrt n}$, when $\Omega(\sqrt{n \log k})$ matches our upper bound of $\widetilde O(n^{3/4})$.

We can summarize this ``density increment'' phenomenon as follows. Most of the edges of $\zpm^n$ occur in the middle layer $\{ x \in \zpm^n : \|x\|_1 = \tfrac{2}{3} n \pm O(\sqrt n) \} = \bigcup_{\ell = -O(\sqrt n)}^{O(\sqrt n)} \{ x \;:\; \|x\|_1 = \tfrac{2}{3}n + \ell \}$.
A convex set is an intersection of halfspaces, but for convenience we consider its complement which is a union of halfspaces, and has the same influence. Consider the ``density'' or measure of the halfspace with normal vector $\vec{1}$ at distance $\tau$ from the origin on the points
$\{ x \;:\; \|x\|_1 = \tfrac{2n}{3} + \ell \}$:
\[
\rho(\ell,\tau) := \pr_{x \in \zpm^n \colon \norm{x}_1 = \frac{2n}{3}+\ell}\left[\sum_{i} x_i > \tau\right] \,.
\]
Suppose that there is a fixed value $\rho$ such that $\rho(\ell,\tau) \approx \rho$ up to constant factors for \emph{all} $\ell = \pm O(\sqrt n)$ simultaneously. Then we can take $k \approx \tfrac{1}{\rho}$ random halfspaces
with threshold $\tau$ and combine their boundary edges, since they will be essentially disjoint on the whole middle layer, and it is not hard to show that the influence of the resulting union is roughly $\tau$. It happens that the condition of $\rho(\ell,\tau)$ being approximately equal for all values $\ell = \pm O(\sqrt n)$ holds for $\tau$ up to $\tau \approx n^{3/4}$ but for $\tau \gg n^{3/4}$ the intersection of the halfspace with the set $\{ x \;:\; \|x\|_1 = \frac{2n}{3} + \ell\}$ grows extremely fast with $\ell$ making $\rho(-\sqrt{n}, \tau) \ll \rho(\sqrt n, \tau)$, and the intersection of halfspaces with threshold $\tau$ quickly approaches the ball with influence $O(\sqrt n)$ (see \cref{ex:balls}).

\subsubsection{Sample-Based Learning and Testing}

\paragraph{Learning Upper Bound:} Our proof of \Cref{thm:learner} follows by combining our upper bound on the influence from \Cref{thm:inf-bound} with the Low-Degree Algorithm of Linial, Mansour, and Nisan~\cite{LinialMN93}. In particular, using Fourier analysis over $\zpm^n$ in combination with \Cref{thm:inf-bound} we can show that for convex sets, a $(1-\eps)$-fraction of the Fourier mass is on the coefficients with degree at most $\widetilde{O}(n^{3/4})/\eps$. 
Then we may use the Low-Degree Algorithm for learning the convex sets; see \Cref{section:low-degree-algorithm}.
Since the ternary hypercube is a non-standard domain, we state the necessary Fourier analysis for functions over $\zpm^n$ in \Cref{section:fourier-setup}, which follows \cite[Chapter 8]{O14}. One technical difference between Fourier analysis over the Boolean and ternary hypercubes is that the standard Fourier basis over $\{\pm 1\}^n$ is given by the parity functions which are bounded in $[0,1]$, whereas any Fourier basis over $\zpm^n$ will have functions taking value $2^{O(n)}$ on some elements $x \in \zpm^n$. Nevertheless, with some care, we show that the Low-Degree Algorithm still works.

\paragraph{Sample-Based Testing Lower Bound:} Our proof of \Cref{thm:2sidedsamplestestingLB} uses a family of functions known as \emph{Talagrand's random DNFs} adapted to the ternary hypercube. As we mentioned, this family of functions has been used to prove lower bounds for monotonicity testing \cite{BeBl16,Chen17}. Our adapted version is described as follows. Each "term" of the DNF is chosen to be a random point $t \in \zpm^n$ with $\norm{t}_1 = \sqrt{n}$. We then say that a point $x \in \zpm^n$ "satisfies" $t$ if $x_i = t_i$ for all $i \in [n]$ where $t_i \in \{\pm 1\}$. After choosing $N$ random terms $t^{(1)},\ldots,t^{(N)}$ we define the disjoint regions of $\zpm^n$ given by $U_1,\ldots,U_N$ where $U_i$ is the set of points $x \in \zpm^n$ with $\norm{x}_1 \in [2n/3 \pm \sqrt{n}]$ which satisfy a unique term. Choosing $N = 3^{\sqrt{n}}$ results in $\bigcup_{i=1}^N U_i$ covering a constant fraction of the domain. We then define two distributions $\Dyes$ and $\Dno$ as follows. Recall that $B_r$ is the radius-$r$ ball in the ternary cube (\Cref{ex:balls}) and let $D$ denote the set of points $x \in \zpm^n$ with $\norm{x}_1 \in [2n/3 \pm \sqrt{n}]$ that don't satisfy \emph{any} term.

\begin{itemize}
    \item $S \sim \Dyes$ is drawn by setting $S = B_{\frac{2n}{3}-\sqrt{n}} \cup D \cup \left(\bigcup_{i \in T} U_i\right)$ where $T$ includes each $i \in [N]$ independently with probability $1/2$. Such a set is always convex.
    \item $S \sim \Dno$ is drawn by setting $S = B_{\frac{2n}{3}-\sqrt{n}} \cup D \cup C$ where $C$ includes each $x \in \bigcup_{i=1}^N U_i$ independently with probability $1/2$. Informally, this set will be $\Omega(1)$-far from convex with constant probability since its intersection with the middle layers is  random. 
\end{itemize}

For both distributions, each point $x \in \bigcup_{i=1}^N U_i$ satisfies $\pr_S[x \in S] = 1/2$ and if $x \in U_i$ and $y \in U_j$ where $i\neq j$, then the events $x \in S$ and $y \in S$ are independent. Thus, to distinguish $\Dyes$ and $\Dno$ one has to see at least two points from the same $U_i$ and this gives our sample complexity lower bound. 

\subsubsection{Non-Adaptive One-Sided Testing} 

The proofs of \cref{thm:UB_nonadaptive_1sided,thm:LB_nonadaptive_1sided,thm:samples_1sided} all rely on a partial order $\preceq$ defined on $\zpm^n$, which we call the \emph{outward-oriented poset}, that has the origin $0^n$ as the minimum element and the corners of the cube $\{\pm 1\}^n$ as the maximum elements. (See \cref{sec:outward} for the formal definition of this poset and a discussion of its properties and history.) For any $\by \in \zpm^n$, we define $\upset(\by) \define \{ \bx \in \zpm^n : \by \preceq \bx \}$ to represent the set of points above $\by$ in this poset. 

\paragraph{Non-Adaptive One-Sided Upper Bound:} An important property of the outward-oriented poset in the context of testing convexity is that any point $\by$ in the convex hull of a set of points $\bX \subseteq \zpm^n$ is also in the convex hull of $\bX \cap \upset(\by)$. Conversely, if a set $\bS \subseteq \zpm^n$ is \emph{not} convex, then there is a certificate of non-convexity of the form $(\bX, \by)$ where $\by \notin S$ is in the convex hull of $\bX \subseteq \bS$, and $\bX \subseteq \upset(\by)$.
This property implies that a convexity tester can  search for certificates of non-convexity by repeatedly choosing a random point $\by$ and querying all points in $\upset(\by)$. A na\"ive implementation of this idea leads to a query complexity that is significantly larger than the bound in the theorem. However, the ternary hypercube satisfies a strong concentration of measure property: almost all of the points in the ternary hypercube have $\tfrac23 n \pm O(\sqrt n)$ non-zero coordinates. As a result, we can refine the convexity tester to only query the points in $\upset(\by)$ whose number of non-zero coordinates is at most $\tfrac{2}{3}n + O(\sqrt n)$ to obtain the desired query complexity. The details of the proof of \cref{thm:UB_nonadaptive_1sided} are presented in \cref{sec:UB_nonadaptive_1sided}.

\paragraph{Non-Adaptive One-Sided Lower Bound:} The lower bound in \cref{thm:LB_nonadaptive_1sided} is obtained by considering the class of \emph{anti-slabs}, which are defined by choosing a vector $\bv \in \zpm^n$ with $n/2$ non-zero coordinates and taking the set of points $\{ \bx \in \zpm^n : |\inn{\bv, \bx}| > \tau \}$. It is quite easy to find certificates of non-convexity for anti-slabs---the three points $-\bx$, $0^n$, and $\bx$ obtained by choosing $\bx$  uniformly at random in the ternary hypercube forms such a certificate with reasonably large probability whenever $\tau$ is small enough. However, we can eliminate these certificates of non-convexity if we ``truncate'' the anti-slabs by including the set of points whose number of non-zero coordinates is below $\tfrac{2}{3}n - O(\sqrt n)$, and excluding the points whose number of non-zero coordinates is above $\tfrac{2}{3}n + O(\sqrt n)$. We show that any certificate of non-convexity for these truncated anti-slabs must have two points $\bx, \bz$ with a large difference between $\inn{\bv, \bx}$ and $\inn{\bv, \bz}$, but on the other hand, any small set of queries has a low probability of including such a pair when $\bv$ is chosen at random.

\paragraph{Sample-Based One-Sided Lower Bound:} Finally, the proof of the lower bound \cref{thm:samples_1sided} again uses the outward-oriented poset and the connection between convex hulls and the upwards sets $\upset(\by)$ to show that any set of $3^{o(n)}$ samples is unlikely to draw any point $\by$ that is contained in the convex hull of the other sampled points and thus to have any possibility of identifying a certificate of non-convexity of any set.

\subsection{Discussion and Open Problems}

As far as we know, we are the first to study convex sets and their associated algorithmic problems on the ternary hypercube. Thus there are many possible questions one could ask. In this section we discuss a few such questions which we find most interesting.

\paragraph*{Learning and sample-based testing.}
The most obvious question which our work leaves open is that of determining the true sample complexity of learning and sample-based testing of convex sets in the ternary hypercube, where our results leave a gap of $3^{\Omega(\sqrt{n})}$ vs. $3^{\widetilde{O}(n^{3/4})}$. By \cref{thm:infLB}, our upper bound of $\widetilde{O}(n^{3/4})$ on the influence of convex sets is tight up to a factor of $\log^{1/4} n$, and therefore to improve our learning upper bound would require another method. 

\begin{question} Can we close the gap of $3^{\Omega(\sqrt{n})}$ vs. $3^{\widetilde{O}(n^{3/4})}$ for learning convex sets and for sample-based convexity testing in $\zpm^n$? \end{question}

\paragraph{Testing with two-sided error.}
Our results for testing with queries apply only to case of one-sided error. Earlier work on testing convex sets under the Gaussian distribution on $\R^n$ with samples showed that, in that setting, two-sided error was more efficient than one-sided \cite{CFSS17}.

\begin{question}
    Is there a two-sided error non-adaptive tester for domain $\zpm^n$ 
    with better query complexity than our one-sided error tester?
\end{question}

Our lower bound technique does not suffice for two-sided error. This is because the class of anti-slabs, which we proved are hard to distinguish from convex sets using a one-sided tester, can be distinguished from convex sets with \emph{two-sided} error using only $O(n)$ samples. To do so, one may use the standard testing-by-learning reduction of \cite{GGR98}, together with an $O(n)$ bound on the VC dimension of the anti-slabs (which are essentially the union of two halfspaces).

\paragraph*{Testing convexity in other domains.} 
Our results show that queries can be more effective than samples for testing discrete convex sets in some high-dimensional domains.
Is this true for all discrete high-dimensional domains?

\begin{question}
What are the sample and query complexities for testing discrete convexity over general hypergrids $[m]^n$?
\end{question}

Note that our techniques do not immediately generalize to larger hypergrids, so answering the last question even for the hypergrid $\{0,\pm 1,\pm 2\}^n$ requires some new ideas.

It would also be interesting to see if the gap between sample and query complexity also holds for continuous sets.

\begin{question}
    Can queries improve upon the bounds of \cite{CFSS17,HY21} for testing convex sets with samples in $\R^n$ under the Gaussian distribution? 
\end{question}

It is not clear if there is a formal connection between testing convex sets on the domain $\zpm^n$ and on the domain $\R^n$ under the standard Gaussian distribution. One might expect a connection here because the uniform distribution on $\zpm^n$ acts similarly to the Gaussian in certain ways when $n \to \infty$. But we do not see how to construct direct reductions between these two settings for the problem of convexity testing.
Also, there is an intriguing analogy between monotone subsets of $\{\pm 1\}^n$ and convex subsets of $\R^n$ in the Gaussian space~\cite{DeNS22}. How do convex subsets of $\{0,\pm 1\}^n$ fit into this analogy?

\section{Convexity on the Ternary Hypercube}

The main object of study in this paper is the \emph{ternary hypercube}, an analogue of the Boolean hypercube over the ternary set $\zpm^n$. This set can be viewed as a discrete subset of $\R^n$, as a (hyper)grid graph in which two points $\bx, \by \in \zpm^n$ are connected by an edge if and only if $\sum_{i=1}^n |x_i - y_i| = 1$, and as a poset that we will describe in more detail in the subsection below. 

The study of the ternary hypercube and more general grid graphs goes back at least to Bollob{\'a}s and Leader~\cite{BollobasL91}. As a poset, its study goes back at least to Metropolis and Rota~\cite{MetropolisR78}. The ternary hypercube appears to have some particularly elegant structure that is not necessarily shared by larger hypergrids. We describe some of these fundamental properties in the following subsections.

\subsection{The Outward-Oriented Poset}
\label{sec:outward}

We define a partial order over $\zpm^n$, which puts the origin $0^n$ as the minimum element and the corners $\{\pm 1\}^n$ as the maximum elements. 

\begin{definition} [Outward-Oriented Poset] \label{def:oop} We denote by $(\zpm^n,\preceq)$ the $n$-wise product of the partial order defined by $0 \prec 1$ and $0 \prec -1$. Equivalently, we write $\by \preceq \bx$ when $\forall i \in [n] \colon (\by_i \neq 0 ~\Longrightarrow~ \bx_i = \by_i)$. \end{definition}

The outward-oriented poset can easily be extended to a lattice (by adding a global maximum point), though since we do not need this extension we do not pursue it here.
The outward-oriented poset appears naturally in many different contexts and, as a result, has received different names. For instance, it arises in the study of the faces of the Boolean hypercube~\cite{MetropolisR78}, where it is sometimes called the ``cubic lattice'', and in the study of partial Boolean functions (see, e.g.,~\cite{Engel97}). We use the name ``outward-oriented poset'' to emphasize the fact that this poset is distinct from the partial order inherited from $\R^n$.

\begin{definition} [Upper Shadow] \label{def:upset} For any point $\by \in \zpm^n$, the \emph{upper shadow} of $\by$ is the set
\[
  \upset(\by) \define \left\{ \bx \in \zpm^n : \by \preceq \bx \right\} \,.
\]
\end{definition}

\subsection{Convexity and Witnesses of Non-Convexity}

Given a set of points $\bX \subseteq \{0,\pm 1\}^n$, we denote the convex hull of $\bX$ by
\[
\conv(\bX) := \left\{\sum_{\bx \in \bX} \lambda_{\bx} \bx \colon \sum_{\bx \in \bX} \lambda_{\bx} = 1 \text{ and } \lambda_{\bx} \geq 0 \text{, } \forall \bx \in \bX\right\} \text{.}
\]
\begin{definition} [Discrete Convexity] A set $\bS \subseteq \{0,\pm 1\}$ is \emph{convex} if $\bS = \conv(\bS) \cap \{0,\pm 1\}^n$. \end{definition}

Let $\Delta(\bS,\bT)$ denote the cardinality of the symmetric difference between $\bS$ and $\bT$. Given $\bS \subseteq \zpm^n$, we define $\dist(\bS,\mathbf{convex})$ as the minimum, over all convex sets $\bT \subseteq \zpm^n$, of $\Delta(\bS,\bT) \cdot 3^{-n}$. For brevity, we also sometimes use the notation $\eps(\bS) := \dist(\bS,\mathbf{convex})$. If $\eps(\bS) \geq \eps$ for some $\eps \in (0,1)$, then we say that $\bS$ is $\eps$-far from convex.

\begin{definition} [Violating Pairs] Consider $S \subseteq \zpm^n$. If $X \subseteq S$ and $y \in \conv(X) \cap \zpm^n$, but $y \notin S$, then we call $(X,y)$ a \emph{violating pair} for $S$. The pair is called \emph{minimal} if $y \notin \conv(X')$ for any strict subset $X' \subset X$. \end{definition}

All of our results exploit the following key property of the outward-oriented poset. This fact captures the structure of $\zpm^n$ which we use throughout the paper.

\begin{fact} \label{fact:simplex1} If a violating pair $(\bX,\by)$ is minimal, then $\bX \subseteq \upset(\by)$. \end{fact}

\begin{proof} We have $\by = \sum_{\bx \in \bX} \lambda_{\bx} \bx$ where $\sum_{\bx \in \bX} \lambda_{\bx} = 1$. Moreover, the minimality of $(\bX,\by)$ implies that $\lambda_{\bx} > 0$ for all $\bx \in \bX$. Now, let $i \in [n]$ be some coordinate where $\by_i \neq 0$. We need to show that $\bx_i = \by_i$ for all $\bx \in \bX$. Without loss of generality, suppose $\by_i = 1$. Thus, we have $1 = \sum_{\bx \in \bX} \lambda_{\bx} \bx_i$. If $\bx_i < 1$ for some $\bx \in \bX$, then we would have $\sum_{\bx \in \bX} \lambda_{\bx} \bx_i < 1$, which is a contradiction. \end{proof}

\begin{figure}[t] \label{fig:2d_example}
	\begin{center}
		\includegraphics[trim = 0 0 0 0, clip, scale=0.35]{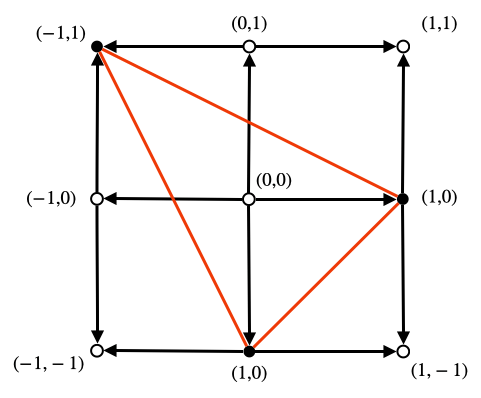}
		\caption{An illustration of $\zpm^2$. Arrows indicate the direction of the partial order. The red triangle shows the convex hull of $\bX := \{(-1,1),(1,0),(0,1)\}$, which contains the origin. I.e. $(\bX,(0,0))$ is a minimal violating pair for $\bX$.}
	\end{center}
\end{figure}

\begin{fact} \label{fact:convex-characterization} Let $\bS \subseteq \zpm^n$. The following two statements are equivalent.
\begin{itemize}
    \item $\bS$ is not convex.
    \item There exists a minimal violating pair $(X,y)$ for $S$.
\end{itemize}  
\end{fact}

\begin{proof} Suppose there exists a minimal violating pair $(\bX,\by)$ for $S$. Since $\bX \subseteq \bS$, we have $\conv(\bX) \subseteq \conv(\bS)$ and so $\by \in \conv(\bS)$. Thus, $\by \notin \bS$ implies $\bS$ is not convex. Now suppose $S$ is not convex. Then there exists $y \in (\conv(S) \cap \zpm^n) \setminus S$. Let $X \subseteq S$ be a minimal set of points such that $y \in \conv(X)$. The pair $(X,y)$ is a minimal violating pair for $S$. \end{proof}

\begin{fact} Consider  $\bS,\bQ \subseteq \zpm^n$. If $\bQ$ does not contain any $\bX \cup \{y\}$ such that $(\bX, \by)$ is a violating pair for $\bS$, then there exists a convex set $\bS'$ such that $\bS' \cap \bQ = \bS \cap \bQ$.
\end{fact}

\begin{proof} Let $\bS' = \conv(\bS \cap \bQ)$ and consider an arbitrary $\by \in \bQ$. We need to show that $\by \in \bS$ if and only if $\by \in \bS'$. Clearly, $\by \in \bS$ implies $\by \in \bS'$. Now suppose $\by \in \bS'$ and note this implies $\by \in \conv(\bS \cap \bQ) \subseteq \conv(\bS)$. Thus, if $\by \notin \bS$, then $(\bS \cap \bQ, \by)$ is a violating pair for $\bS$ and this contradicts our assumption about $\bQ$. \end{proof}

The following corollary is crucial for proving our lower bounds in \cref{sec:LB_nonadaptive_1sided} and \cref{sec:LB_samples_1sided}.

\begin{corollary} \label{cor:witness} Let $T$ be a convexity tester for sets $\bS \subseteq \zpm^n$ with $1$-sided error. Suppose $T$ rejects a set $\bS$ after querying a set $\bQ$.
Then $\bQ$ contains some $X \cup \{y\}$ such that $(X,y)$ is a minimal violating pair for $\bS$. \end{corollary}

\subsection{Concentration of Mass in the Ternary Hypercube} \label{sec:prelims_concentration}

For $\bx \in \zpm^n$, observe that $\norm{\bx}_1 = \norm{\bx}_2^2$ is precisely the number of non-zero coordinates of $\bx$. Moreover, each coordinate of a uniformly random $\bx$ is non-zero with probability $2/3$, and so $\Exu{\bx \in \zpm^n}{\norm{\bx}_1} = \frac{2n}{3}$. Standard concentration inequalities yield the following bound on the number of points $\bx \in \zpm^n$ where $\norm{\bx}_1$ is far from this expectation.

\begin{fact} \label{fact:chernoff} For every $\tau \geq 0$, 
\[
\Pru{\bx \in \zpm^n}{\left|\norm{\bx}_1 - \frac{2n}{3}\right| > \tau} \leq 2\exp(-\tau^2/2n) \text{.}
\]
\end{fact}

\begin{proof} We have $\norm{\bx}_1 = \sum_{i=1}^n X_i$ where $X_i = 1$ with probability $2/3$ and $X_i = 0$ with probability $1/3$. Thus, the bound follows immediately from Hoeffding's inequality. \end{proof}

Given $\tau \geq 0$, we use the following notation to denote the inner, middle, and outer layers of $\zpm^n$ with respect to distance $\tau$:
\begin{align} \label{eq:inn_mid_out}
   &\inlayers(\tau) := \left\{\bx \colon \norm{\bx}_1 - \frac{2n}{3} < - \tau \right\}\text{, } \nonumber \\
   &\midlayers(\tau) := \left\{\bx \colon \left|\norm{\bx}_1 - \frac{2n}{3}\right| \leq \tau \right\}\text{, } \nonumber \\
   &\outlayers(\tau) := \left\{\bx \colon \norm{\bx}_1 - \frac{2n}{3} > \tau \right\} \text{.}
\end{align}

\begin{figure}[t] \label{fig:zpm_layered}
	\begin{center}
		\includegraphics[trim = 0 0 0 3, clip, scale=0.14]{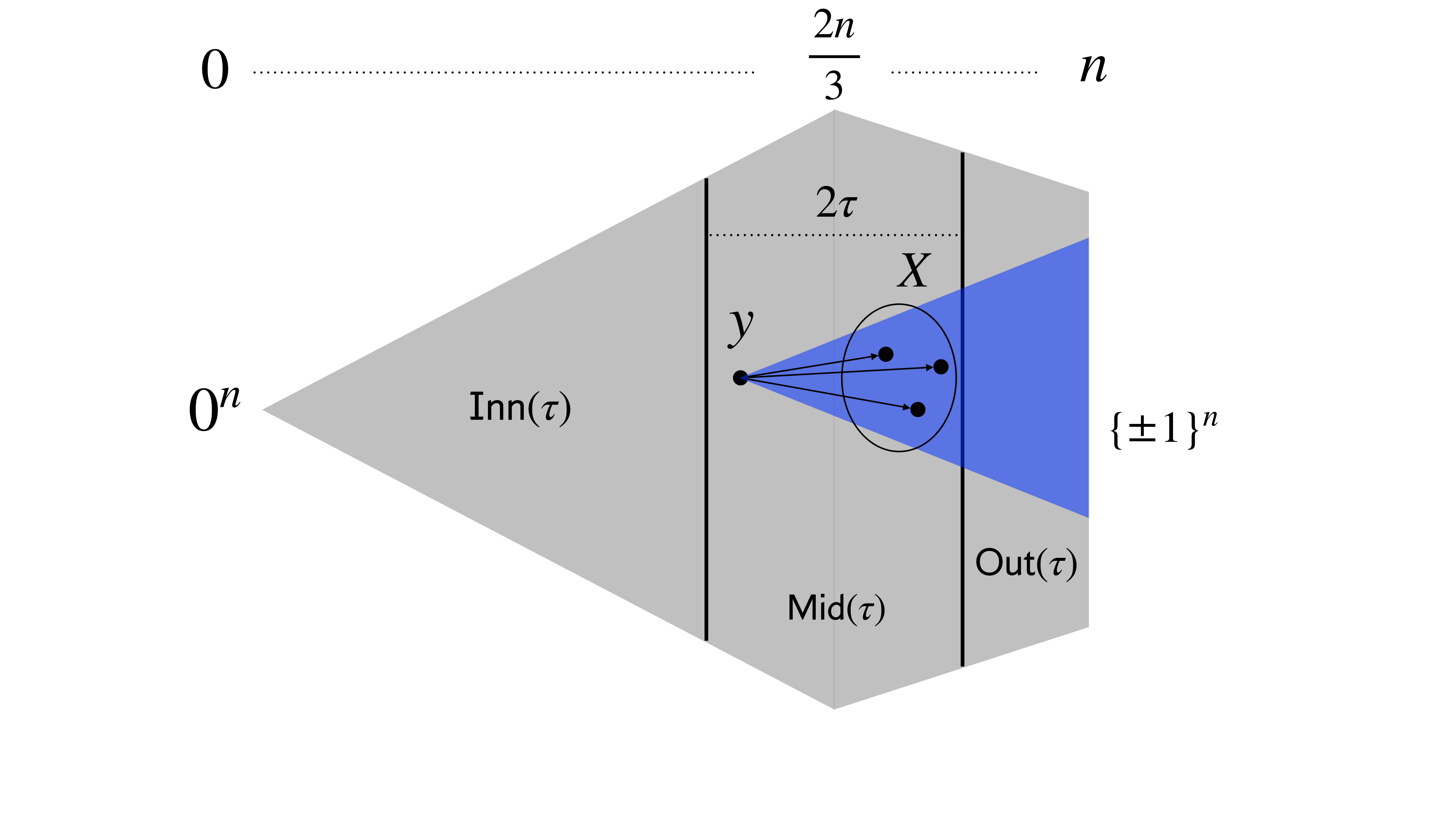}
		\caption{This figure shows a pictorial representation of $\zpm^n$ as a poset. Any vertical slice represents the set of all points with some fixed number of non-zero coordinates, and this number is increasing from left to right. The left-most point is the origin and the right-most points are the vertices of the hypercube $\{\pm 1\}^n$. The outward-oriented poset goes from left to right. The shaded blue region emanating from $\by$ is the set $\upset(\by)$ of points above $\by$ in the partial order. The set $X$ represents some minimal set of points for which $y \in \conv(X)$ and thus $\by \prec \bx$ for all $\bx \in \bX$, by \cref{fact:simplex1}.}
	\end{center}
\end{figure}

\section{The Influence of Convex Sets}
\label{section:influence}

In this section we prove that the maximum edge boundary of convex sets in $\zpm^n$ is $\widetilde \Theta(n^{3/4}) \cdot 3^n$, or equivalently that the influence is $\widetilde \Theta(n^{3/4})$.

\subsection{Upper Bound}

We prove that convex sets in the ternary hypercube have influence $\widetilde
O(n^{3/4})$. The main idea in the proof is to relate the influence of a convex set $S$ to the number
of sign-changes in the maximum of a set of one-dimensional random walks. The proof will consider a
random walk $\bm{X}^{(0)}, \bm{X}^{(1)}, \dotsc, \bm{X}^{(m)}$ starting from a random position in
the middle layer of the ternary hypercube and moving randomly ``outward'' for $m = O\left(\sqrt
\frac{n}{\log n} \right)$ steps, and count the number of influential edges crossed near the ``middle
layers'' by relating them to one-dimensional random walks. We begin in \Cref{section:random-walks} with definitions regarding the
one-dimensional random walks that we require and then in \Cref{section:inf-upperbound-proof} show how they relate to the number of
influential edges of $S$; finally, in \Cref{section:max-walk-crossing-bound} we prove the necessary bound on the number of sign-changes of the
one-dimensional random walks.

\emph{Notation.} In this section it will be convenient to use bold letters like $\bm{X}$ for random
variables, with the non-bold letter $X$ being reserved for a fixed instantiation of $\bm{X}$.

\subsubsection{One-Dimensional Random Walks and the Max-Walk} \label{section:random-walks}

Let us define the types of one-dimensional random walks that will be necessary for our proof.

\begin{definition} [Random Walks] \label{def:random-walk} Let $x = (x_1, \dotsc, x_m) \in \R^m$. Fix any permutation $\sigma : [m] \to [m]$ and sign vector
$\varepsilon
= (\varepsilon_1, \dotsc, \varepsilon_m) \in \pmset^m$. For any $a \in \R$, we define the function $W_{x}^{+a}(t ;
\sigma, \varepsilon)$ for $t \in \{0\} \cup [m]$ as
\[
  W_{x}^{+a}(t; \sigma, \varepsilon) \define \begin{cases}
    a &\text{ if } t = 0 \\
    a + \sum_{i=1}^t \varepsilon_i x_{\sigma(i)} &\text{ if } t > 0 \,.
  \end{cases}
\]
The random walk $\bm{W}_x^{+a}$ is defined by choosing a uniformly random permutation $\bm \sigma$ and vector $\bm \varepsilon \sim \pmset^m$ and
setting $\bm{W}_x^{+a}(t) = W_x^{+a}(t; \bm \sigma, \bm \varepsilon)$ for every $t$. If $a=0$ we drop the superscript.
\end{definition}

The main quantity of interest to us is the number of sign-changes of a random walk, defined as
follows.

\begin{definition}[Crossing Number]
Let $W : \{0\} \cup [m] \to \R$ be any sequence. We define the \emph{crossing number} $C(W)$ as the
number of sign-changes of $W$, defined as the number of times $t \in [m]$ such that either
$W(t) \geq 0 > W(t-1)$ or $W(t) < 0 \leq W(t-1)$.
\end{definition}

An important feature of our random walks will be that they have the Distinct Subset-Sum (DSS).

\begin{definition} [DSS Random Walk] \label{def:SA-walk}
We say a sequence $x \in \R^m$ has the Distinct Subset-Sum (DSS) property if for every two disjoint
subsets $A, B \subseteq [m]$, it holds that $\sum_{a \in A} x_a \neq \sum_{b \in B} x_b$. In
particular, the random walk $\bm{W}_x$ satisfies
\[
  \forall t \in [m]\,,\; \Pru{\bm{\sigma}, \bm{\varepsilon}}{ W_x(t; \bm{\sigma},\bm{\varepsilon}) = 0} = 0 \,.
\]
Note that, if $x$ has the DSS property, then so does any subsequence of $x$.
\end{definition}

We will require an upper bound on the crossing number of \emph{max-walks}, which are random walks
defined as the maximum of a set of \emph{constituent} walks of the type defined above.

\begin{definition}[Max-Walk] \label{def:max-walk}
Let $X$ be a set of sequences $x \in \R^m$, and let $a : X \to \R$. For a fixed permutation $\sigma$
and vector $\eps \in \pmset^m$, define
\[
  M_X^{+a}(t ; \sigma, \varepsilon) \define \max_{x \in X} W_x^{+a(x)}(t ; \sigma, \varepsilon) \,,
\]
and let the random walk $\bm{M}_X^{+a}$ be defined as
\[
  \bm{M}_X^{+a}(t) \define M_X^{+a}(t ; \bm{\sigma}, \bm \varepsilon)
\]
where $\bm\sigma, \bm \varepsilon$ are chosen uniformly at random.
\end{definition}

The main fact about max-walks that we require is the following, which we prove in
\cref{section:max-walk-crossing-bound}.

\begin{lemma}[Max-Walk Crossing Number]
\label{lemma:max-walk-crossing-number}
Let $X$ be a set of sequences $x \in \R^m$, each having the DSS property, and let $a : X \to \R$.
Then
\[
  \Ex{ C(\bm{M}_X^{+a}) } = O(\sqrt m) \,.
\]
\end{lemma}

\subsubsection{Upper Bound on the Number of Influential Edges of a Convex Set} \label{section:inf-upperbound-proof}

We now prove the following upper bound on the influence of any convex set in the ternary hypercube,
restated below for convenience.

\thminfbound*

We require the following basic property of discrete convex sets.

\begin{proposition}
\label{prop:convex-set-to-nice-halfspaces}
Let $S \subseteq \zpm^n$ be any discrete convex set. Then there is a finite set of vectors $V
\subseteq \R^n$ and thresholds $\tau : V \to \R$, where each $v \in V$ defines a halfspace $H_v \define
\{ x \in \zpm^n \;:\; \inn{v,x} < \tau(v) \}$, such that $S = \bigcap_{v \in V} H_v$. One may also
assume that $V$ satisfies the property that, for every $v \in V$ and every two disjoint subsets $A,
B \subseteq [n]$, $\sum_{i \in A} v_i \neq \sum_{j \in B} v_j$.
\end{proposition}
\begin{proof}
Since $S$ is the intersection of its convex hull $\conv(S)$ with $\zpm^n$, it may be written as the
intersection of $\zpm^n$ with a finite set of halfspaces with normal vectors $V$ and thresholds $\tau : V \to \R$, and one may assume that none of the points in $\zpm^n$ lie on the hyperplane boundary of
any of the halfspaces. Then there is some $\delta > 0$ such that the minimum distance between a
hyperplane and a point of $\zpm^n$ is at least $\delta \cdot n$. For each $v \in V$, apply independent
random perturbations to each coordinate to obtain $v'_i = v_i + r_i$ where $r_i$ is drawn from
$[-\delta,\delta]$ uniformly at random. With probability 1, the resulting set $V' = \{ v' \;:\; v
\in V \}$ satisfies the required conditions.
\end{proof}

\begin{proof}[Proof of \cref{thm:inf-bound}] Recall the definition of the edge-set $E$ of the ternary cube from \cref{eq:edges} and the set $\midlayers(\ell)$ from \cref{eq:inn_mid_out}. Given $\ell > 0$, let 
\[
E_{\ell} = \left\{(u,v) \in E \colon u,v \in \midlayers(\ell)\right\}
\]
denote the set of edges lying in the middle $\ell$ layers of $\zpm^n$.
We consider the following process which samples a random edge in $\zpm^n$. Define $\ell := \sqrt{2n\log n}$ and $m := \sqrt{\frac{n}{\log n}}$. Let $\cD$ denote the distribution over edges defined by the following procedure.
\begin{enumerate}
    \item Sample $\bm{X}^{(0)} \sim \midlayers(\ell)$.
    \item Choose a random subset $\bm{T} \subseteq \{i \colon \bm{X}^{(0)}_i = 0\}$ with $|\bm{T}| = m$ of coordinates where $\bm{X}^{(0)}$ has a $0$. 
    \item Let $\bm{\eps} = (\bm{\eps}_1,\ldots,\bm{\eps}_m) \in \{\pm 1\}^m$ be independent Rademacher random variables and let $\bm{\sigma} \colon [m] \to \bm{T}$ be a random bijection.
    \item For each $s \in [m]$, let $\bm{X}^{(s)} = \bm{X}^{(s-1)} + \bm{\eps}_s e_{\bm{\sigma}(s)} =
\bm{X}^{(0)} + \sum_{i=1}^s \bm{\eps}_i e_{\bm{\sigma}(i)}$ where $e_j$ is the unit vector with a $1$ in coordinate $j$.
    \item Choose $\bm{s} \sim [m]$ and return the edge $(\bm{X},\bm{Y})
      = (\bm{X}^{(s-1)},\bm{X}^{(s)})$.
\end{enumerate}
Note that the above process can be equivalently defined as obtaining $\bm{X}^{(s)}$ by selecting a
uniform random coordinate $\bm{i}$ where $\bm{X}^{(s-1)}_{\bm{i}} = 0$ and flipping that bit to a
random value in $\{\pm 1\}$, with equal probability. This results in a random walk
$\bm{X}^{(0)},\bm{X}^{(1)}, \ldots, \bm{X}^{(m)}$ of length $m$ where each $(\bm{X}^{(s-1)},\bm{X}^{(s)})$ is a
random out-going edge from $\bm{X}^{(s-1)}$. We use two main claims regarding this random walk to complete the proof of the theorem. The first is that choosing an edge $(\bm{X},\bm{Y})
\sim \cD$ is approximately the same as choosing a uniformly random edge from the middle layers. 

\begin{claim} \label{clm:unif} Fix any $z \in \midlayers(\ell)$ and $s \in [m]$. Then $\pr[\bm{X}^{(s)} = z] = \Theta(3^{-n})$. I.e., each step of the random walk is approximately uniformly distributed over $\midlayers(\ell)$. As a corollary, for any fixed edge $(u,v) \in E_{\ell}$, we have 
\[
\pr_{(\bm X,\bm Y) \sim \cD}[(\bm X,\bm Y) = (u,v)]
    = \Theta\left(\frac{1}{n \cdot 3^n}\right) \text{.}
\]
\end{claim}

The second claim is that the probability of $(\bm X, \bm Y) \sim \cD$ being an influential edge is small.

\begin{claim}
\label{lem:prob_D}
$\pr_{(\bm X,\bm Y) \sim \cD}[S(\bm X) \neq S(\bm Y)] \leq O\left(\frac{1}{\sqrt m}\right)$.
\end{claim}

We defer the proof of both claims to the end of the section.
We now prove \cref{thm:inf-bound} using \cref{clm:unif} and \cref{lem:prob_D} as follows. Let $E$ denote the edges of the ternary hypercube and let $E_{\ell} = \{(u,v) \in E \colon u,v \in \midlayers(\ell)\}$. By definition,
\begin{align*}
    \Inf(S) = \frac{1}{3^n} \cdot \Big(|\{(u,v) \in E \setminus E_{\ell} \colon S(u) \neq S(v)\}| + |\{(u,v) \in E_{\ell} \colon S(u) \neq S(v)\}|\Big) \text{.}
\end{align*}
The first term is bounded using \cref{fact:chernoff} as 
\begin{align*}
    \frac{|\{(u,v) \in E \setminus E_{\ell} \colon S(u) \neq S(v)\}|}{3^n} \leq \frac{|E \setminus E_{\ell}|}{3^n} \leq 2n \cdot \frac{|\overline{\midlayers(\ell)}|}{3^n} \leq 2n \cdot 2\exp(-\ell^2/2n) = O(1) 
\end{align*}
since every vertex has degree at most $2n$. The second term is bounded as
\begin{align}
    \frac{|\{(u,v) \in E_{\ell} \colon S(u) \neq S(v)\}|}{3^n} &= \frac{|E_{\ell}|}{3^n} \cdot
\pr_{(u,v) \sim E_{\ell}} [S(u) \neq S(v)] \leq \frac{2n}{3} \cdot  \pr_{(u,v) \sim E_{\ell}} [S(u) \neq S(v)] \nonumber \\
    &\leq L n \cdot \pr_{(\bm X,\bm Y) \sim \cD} [S(\bm X) \neq S(\bm Y)] \leq L' n \cdot m^{-1/2} = L' \cdot n^{3/4} \log^{1/4} n \,, \nonumber
\end{align}
where $L, L'$ are absolute constants.  The first inequality follows simply from $E_{\ell} \subset E$
and $|E| = 2n \cdot 3^{n-1}$. The second inequality follows from \cref{clm:unif} and the third
inequality follows from \cref{lem:prob_D}. This completes the proof of \cref{thm:inf-bound}.
\end{proof}

Let us now complete the deferred proofs of \cref{clm:unif} and \cref{lem:prob_D}.

\begin{proof}[Proof of \cref{clm:unif}]
Let $\norm{z}_1 = \frac{2n}{3}+r$ where $|r| = O(\sqrt{n\log n})$. In order for $X^{(s)} = z$ to occur we
must have $\norm{X^{(0)}}_1 = \frac{2n}{3} + r - s$. Thus, the probability is
\begin{align*}
    \pr[X^{(s)} = z] &= \frac{1}{3^n} \left({n \choose \frac{2n}{3}+r-s} \cdot 2^{\frac{2n}{3}+r-s}\right) \cdot \left({n \choose \frac{2n}{3}+r} \cdot 2^{\frac{2n}{3}+r}\right)^{-1} \\
    &= \frac{1}{3^n} \cdot \frac{1}{2^s} \cdot  {n \choose \frac{2n}{3}+r-s} {n \choose \frac{2n}{3}+r}^{-1} = \Theta(3^{-n})
\end{align*}
where the last step is due to the following fact: 


\emph{If $|r| \leq O(\sqrt{n\log n})$ and $s = O(\sqrt{\frac{n}{\log n}})$, then ${n \choose
\frac{2n}{3} + r - s}{n \choose \frac{2n}{3} + r}^{-1} = \Theta(2^s)$. As a corollary, the number of
points in the ternary cube with hamming weight $\frac{2n}{3}+r-s$ and $\frac{2n}{3}+r$ differ by at
most a constant factor.}

This is proved as follows.
\begin{align*}
    \frac{{n \choose \frac{2n}{3} + r - s}}{{n \choose \frac{2n}{3} + r}}
  &= \frac{\left(\frac{2n}{3}+r\right)! \left(\frac{n}{3}-r\right)!}
      {\left(\frac{2n}{3}+r-s\right)!\left(\frac{n}{3}-r+s\right)!} 
  = \prod_{p=0}^{s-1} \frac{\frac{2n}{3}+r-p}{\frac{n}{3}-r+s-p} 
  = 2^s \cdot \prod_{p=0}^{s-1} \frac{\frac{n}{3}+\frac{r}{2}-\frac{p}{2}}{\frac{n}{3}-r+s-p} \\
  &= 2^s \cdot \prod_{p=0}^{s-1} \frac{\frac{n}{3}-r+s-p + (\frac{3r}{2}+\frac{p}{2}-s)}
                                      {\frac{n}{3}-r+s-p} 
  = 2^s \cdot \prod_{p=0}^{s-1} \left(1
      + \frac{\frac{3r}{2}+\frac{p}{2}-s}{\frac{n}{3}-r+s-p}\right)
\end{align*}
Observe that the numerator inside the product is $\pm O(\sqrt{n\log n})$ since $r$ is the dominating term and the denominator is $\Omega(n)$ since $n/3$ is the dominating term. Therefore, we have
\begin{align*}
    \frac{{n \choose \frac{2n}{3} + r - s}}{{n \choose \frac{2n}{3} + r}} = 2^s \cdot \left(1 \pm O\left(\sqrt{\frac{\log n}{n}}\right)\right)^s = \Theta(1) \cdot 2^s
\end{align*}
since $s = O(\sqrt{\frac{n}{\log n}})$. \end{proof}

\begin{proof}[Proof of \cref{lem:prob_D}.]
Let $S$ be an intersection of halfspaces $S = \bigcap_{v \in V} H_v$,
with thresholds $\tau : V \to \R$, in the form promised by \cref{prop:convex-set-to-nice-halfspaces}. In particular, each vector $v \in V$ has the DSS property (\cref{def:SA-walk}).
Fix any value of $\bm{X}^{(0)} = X^{(0)}$ and fix any permutation $\sigma$ and sign-vector
$\varepsilon = (\varepsilon_1, \dotsc, \varepsilon_m) \in \pmset^m$ in the definition of $\cD$, and
consider the resulting fixed values of $X^{(0)}, X^{(1)}, \dotsc, X^{(m)}$. Define $a : V \to \R$ as
$a(v) = \inn{v, X^{(0)}} - \tau(v)$.

For each $v \in V$, consider the sequences $W_v^{+a(v)} \define W_v^{+a(v)}( \cdot \;; \sigma, \varepsilon)$. For
each $X^{(s)}$, observe that $X^{(s)} \in H_v$ if and only if $\inn{v,X^{(s)}} < \tau(v)$, which is
equivalent to the condition $W_v(s) < 0$, since
\begin{align*}
W_v^{+a(v)}(s)
&= a(v) + \sum_{j = 1}^s \varepsilon_j \cdot v_{\sigma(j)} 
= \left(\sum_{i : X^{(0)}_i \neq 0} v_i X^{(0)}_i \right) - t(v) + \sum_{j=1}^s
X^{(s)}_{\sigma(j)} v_{\sigma(j)} \\
&= \left(\sum_{j : X^{(s)}_j \neq 0} X^{(s)}_j v_j \right) - t(v) = \inn{X^{(s)}, v} - t(v) \,.
\end{align*}
Therefore $X^{(s)} \in S$ if and only if $W_v^{+a(v)}(s) < 0$ for all $v \in V$, which is equivalent to
$M_V^{+a}(s) < 0$ where $M_V^{+a}$ is the max-walk (recall \cref{def:max-walk}). Then for fixed sequence $X^{(0)}, \dotsc,
X^{(m)}$ and uniformly random $\bm{s} \sim [m]$, the probability that $(X^{(\bm{s}-1)},
X^{(\bm{s})})$ is an influential edge is equal to $\frac{C(M_V^{+a})}{m}$. 
Therefore, taking
$\bm{\sigma}$ and $\bm{\varepsilon}$ to be random, we have
\[
  \Pru{(\bm X, \bm Y) \sim \cD}{ S(\bm X) \neq S(\bm Y) } = \frac{1}{m} \cdot \Ex{ C(\bm{M}_V^{+{\bm
a}}) } = O(\sqrt m) \,,
\]
where the final bound is due to \cref{lemma:max-walk-crossing-number}, since each vector in $V$ was
assumed to have the DSS property. This concludes the proof of the claim.
\end{proof}

\subsubsection{Crossing Bound for the Max-Walk\texorpdfstring{: Proof of \cref{lemma:max-walk-crossing-number}}{}}
\label{section:max-walk-crossing-bound}

We prove an upper bound on the number of times the maximum of a set of one-dimensional random walks
can change sign. Let us define certain special events in a random walk.

Fix any walk time $m$ and let $W : \{0\} \cup [m] \to \R$. We define:
\begin{itemize}
\item A \emph{downcrossing} of $W$ is a time $t \in [m]$ such that $W(t) < 0 \leq W(t-1)$.
      $C_\downarrow(W)$ is the number of downcrossings of $W$.
\item An \emph{upcrossing} of $W$ is a time $t \in [m]$ such that $W(t) \geq 0 > W(t)$.
      $C_\uparrow(W)$ is the number of upcrossings of $W$.
\item A \emph{downwards level return} of $W$ is any time $t$ such that either:
  \begin{itemize}
    \item If $W(0) \geq 0$ then the smallest time $t \in [m]$ such that $W(t) < W(0)$ is a downwards
      level return.
    \item For any upcrossing $s$ of $W$, the first time $t > s$ such that $W(t) < W(s)$ is a
      downwards level return.
  \end{itemize}
  We write $L_\downarrow(W)$ for the number of downwards level returns of $W$.
\item The \emph{downwards level decrease} times of $W$ is the unique sequence $s_1 < s_2 < \dotsm$
    defined inductively as follows.
  \begin{itemize}
    \item If $W(0) \geq 0$ then $s_1$ is the first time such that $W(s_1) < W(0)$. Otherwise let $t$
be the first upcrossing of $W$. Then $s_1$ is the first time such that $W(s_1) < W(t)$.
    \item For $i > 1$, if $W(s_{i-1}) \geq 0$ then $s_i \in [m]$ is the smallest time such that
        $W(s_i) < W(s_{i-1})$. Otherwise, if $W(s_{i-1}) < 0$, then let $t$ be the first upcrossing
$t > s_{i-1}$ and define $s_i$ as the first time $s_i > t$ such that $W(s_i) < W(t)$.
  \end{itemize}
  We write $S_\downarrow(W)$ for the number of downwards level decreases of $W$.
\item The \emph{upwards level increase} times of $W$ is the unique sequence $t_1 < t_2 < \dotsm$
    defined inductively as follows.
  \begin{itemize}
    \item If $W(0) < 0$ then $t_1$ is the first time such that $W(t_1) > W(0)$. Otherwise let $s$ be
the first downcrossing of $W$. Then $t_1$ is the first time such that $W(t_1) > W(t)$.
    \item For $i > 1$, if $W(t_{i-1}) < 0$ then $t_i$ is the first time such that $W(t_i) >
W(t_{i-1})$. Otherwise if $W(t_{i-1}) \geq 0$, then let $s$ be the first downcrossing $s > t_{i-1}$
and define $t_i$ as the first time $t_i > s$ such that $W(t_i) > W(s)$.
  \end{itemize}
  We write $S_\uparrow(W)$ for the number of upwards level increases of $W$.
\end{itemize}

The main technical tool in our analysis is the following version of Sparre Andersen's fluctuation
theorem~\cite{SA54}, as found in~\cite[Prop.\ 4.1]{BB23}. Recall the definition of $\bm{W}_x$ from
\cref{def:random-walk} and the DSS property from \cref{def:SA-walk}.

\begin{theorem}[Sparre Anderson; see \cite{BB23}, Proposition 4.1]
\label{thm:sa}
For every $m \in \N$,  if $x \in \R^m$ has the DSS property, then the random walk $\bm{W}_{x}$
satisfies
\[
  \Pr{ \forall t \in [m] : \bm{W}_{x}(t) > 0 } = g(m) \define \frac{1}{4^m}{2m \choose m} \,.
\]
\end{theorem}
We define a random variable $\bm{R}$ on the positive integers with
\[
  \forall t \in \N \,,\; \Pr{ \bm{R} = t } \define g(t-1) - g(t)
    = \frac{1}{4^{t-1}}{2(t-1) \choose t-1} - \frac{1}{4^t}{2t \choose t} \,,
\]
where we define $g(0) \define 1$. For each $m \in \N$, we also define a random variable
$\bm{Q}^{(m)}$ by the following process. Set $q = 0$ and $X = 0$; while $X < m$, increment $q$ and
set $X \gets X + \bm{R}$ where $\bm{R}$ is a new independent copy of the random variable defined
above. Then set $\bm{Q}^{(m)} = q$ once this process terminates; note that $\bm{Q}^{(0)} = 0$.
Observe that for every $k \in \N$,
\[
  \Pr{ \bm{Q}^{(m)} \geq k } = \Pr{ \bm{R}_1 + \bm{R}_2 + \dotsm + \bm{R}_k \leq m }
\]
where each $\bm{R}_i$ is an independent copy of $\bm{R}$, and
\begin{align*}
  \Ex{\bm{Q}^{(m)} }
    = \sum_{t=1}^m \Pr{ \bm{R} = t } \cdot \left( 1 + \Ex{ \bm{Q}^{(m-t)} } \right)\,.
\end{align*}

The following holds due to \cref{thm:sa}.
\begin{proposition}
\label{prop:sa-first-decrease}
Let $x \in \R^m$ have the DSS property, let $a \in \R$, and let $\bm{s}_1$, $\bm{t}_1$ denote the
first downwards level decrease time and upwards level increase times of $\bm{W}_x^{+a}$,
respectively. Then for all $z \in [m]$,
\begin{enumerate}
\item If $a \geq 0$ then $\Pr{ \bm{s}_1 = z } = \Pr{ \bm{R} = z }$; and,
\item If $a < 0$ then $\Pr{ \bm{t}_1 = z } = \Pr{ \bm{R} = z }$.
\end{enumerate}
\end{proposition}

\begin{proposition}
\label{prop:sa-level-decrease}
Let $x \in \R^m$ have the DSS property and let $a \in \R$. Then
\[
  \Ex{ \bm{Q}^{(m)} } =
  \Ex{ S_\downarrow(\bm{W}_x^{+a}) + S_\uparrow(\bm{W}_x^{+a}) } \,.
\]
\end{proposition}
\begin{proof}
By induction on $m$. For $m = 1$ we have $\Ex{ \bm{Q}^{(1)} } = \Pr{ \bm{R} = 1 } = 1/2$
and $\Ex{ S_\downarrow(\bm{W}_x^{+a}) + S_\uparrow(\bm{W}_x^{+a}) } = 1/2$ since the random walk
has probability $1/2$ of increasing or decreasing in the first step; if $a \geq 0$ then the walk
must decrease to create a downwards level decrease, while if $a > 0$ then the walk must increase to
create an upwards level increase.

Now let $m > 1$. Suppose $a \geq 0$ without loss of generality. Then the first level increase or decrease is a downwards level
decrease. Let $\bm{s}_1$ be the first downwards level decrease and let $\bm{y}$ denote the random
subsequence of $x$ that remains after removing the first $\bm{s}_1$ elements according to the random permutation
$\bm{\sigma}$. Then by induction and \cref{prop:sa-first-decrease},
\begin{align*}
  \Ex{ S_\downarrow(\bm{W}_x^{+a}) + S_\uparrow(\bm{W}_x^{+a}) }
  &= \sum_{t=1}^m \Pr{ \bm{s}_1 = t } \cdot \left( 1 + 
    \Exuc{}{ S_\downarrow\left(\bm{W}_{\bm{y}}^{+ \bm{W}_x(\bm{s}_1)}\right) + 
         S_\uparrow\left(\bm{W}_{\bm{y}}^{+ \bm{W}_x(\bm{s}_1)}\right) }{ \bm{s}_1 = t } \right) \\
  &= \sum_{t=1}^m \Pr{ \bm{R} = t } \cdot \left(1 + \Ex{ \bm{Q}^{(m-t)} } \right) = \Ex{ \bm{Q}^{(m)} }
        \,. \qedhere
\end{align*}
\end{proof}

For a sequence $W : \{0\} \cup [m] \to \R$, write $Z(W) = \sum_{t=1}^m \ind{W(t) \in \{0,\pm 1\}}$.
\begin{lemma}
\label{lemma:extended-crossings-uniform}
For any $m$, $\Ex{ Z(\bm{W}_{\vec 1}) } = O(\sqrt m)$.
\end{lemma}
\begin{proof}
We first bound the number of times $t$ such that $\bm{W}_{\vec 1}(t) = 0$. If $t$ is odd then $\Pr{
\bm{W}_{\vec 1}(t) = 0} = 0$. If $t$ is even then there is a universal constant $C$ such that
\[
  \Pr{ \bm{W}_{\vec 1}(t) = 0 } = \frac{1}{2^t} {t \choose t/2} \leq C \cdot \frac{1}{ \sqrt t } \,.
\]
Therefore the expected number of times $t$ with $\bm{W}_{\vec 1}(t) = 0$ is at most
\[
  \sum_{t \text{ even}} \Pr{ \bm{W}_{\vec 1}(t) = 0 }
  \leq C \cdot \sum_{t=1}^m \frac{1}{\sqrt t} = O(\sqrt m) \,.
\]
Now observe that the expected number of times $t$ where $\bm{W}_{\vec 1}(t) = 1$ is the average of
the expected number of times where the shifted walks $\bm{W}_{\vec 1}^{+a}$ is 0 on domain $[m-1]$,
where $a = \pm 1$, and the same holds for the number of times $t$ where $\bm{W}_{\vec 1}(t) = -1$.
\end{proof}

\begin{proposition}
\label{prop:sa-q-bound}
There exists
$x \in \R^m$ with the DSS property such that $\Ex{ S_\downarrow(\bm{W}_x) + S_\uparrow(\bm{W}_x) } \leq \Ex{Z(\bm{W}_{\vec 1})}$.
As a consequence,
\[
  \Ex{ \bm{Q}^{(m)} } = O(\sqrt m) \,.
\]
\end{proposition}
\begin{proof}
Let $\delta \define \frac{1}{3m}$. Let $\bm{x} \define \vec 1 + \bm{z}$ where $\bm{z} \sim
[-\delta,\delta]^m$ uniformly at random. Note that $\bm{x}$ has the DSS property with probability 1.
For any fixed $z \in [-\delta,\delta]^m$, any permutation $\sigma$, and any $r \in \pmset^n$, write
$W_x(t) \define W_x(t ; \sigma, \varepsilon)$ for $x = \vec 1 + z$. Then we have $W_x(t ; \sigma,
\varepsilon) \in [W_{\vec 1}(t ; \sigma, \varepsilon) - 1/3, W_{\vec 1}(t ; \sigma, \varepsilon) + 1/3]$. Now
fix any $z \in [-\delta,\delta]^m$ such that $x = \vec 1 + z$ has the DSS property; we show that it
satisfies the required condition.

Let $s_1 < s_2 < \dotsm < s_k$ be the downwards level decreasing or upwards level increasing points
for $W_x$, let $s_0 = 0$, and observe that a point cannot be both downwards level decreasing and
upwards level increasing.  We show by induction on $i$ that $|W_x(s_i)| \leq 1 + 1/3$ and therefore
that $W_{\vec 1}(s_i) \in \{0, \pm 1\}$. Since $W_x(0) = 0$ it must be that $s_1$ is downwards level
decreasing and $W_x(s_1) < 0 = W_x(0) \leq W_x(s_1-1)$ and therefore $W_x(s_1) \geq W_x(s_1-1) - 1 -
1/3 \geq -4/3$ so it must be that $W_{\vec 1}(s_1) \in \{-1,0\}$. For $i > 1$, suppose that
$s_i$ is a downwards level decreasing point. If there exists an upcrossing point $a > s_{i-1}$ such
that $W_x(s_i) < W_x(a)$, then we observe that $W_x(a-1) < 0 \leq W_x(a)$ and therefore $W_x(a) < 1
+ 1/3$ so $W_{\vec 1}(a) \in \{0,1\}$. Now $W_x(s_i) < W_x(a) \leq W_x(s_i-1)$ so it must be that $-
1 - 1/3 \leq W_x(s_i) < 1 + 1/3$ so $W_{\vec 1}(s_i) \in \{0, \pm 1\}$. On the other hand, if
$W_x(s_{i-1}) \geq 0$ and $W_x(s_i) < W_x(s_{i-1})$ then by induction we have $W_{\vec 1}(s_{i-1})
\in \{0,1\}$, and also $W_x(s_i-1) > W_x(s_{i-1})$, so again we have $-1-1/3 \leq W_x(s_i) < 1+1/3$
and therefore $W_{\vec 1}(s_i) \in \{0, \pm 1\}$. A similar argument holds for the upwards level
increasing points.

The conclusion now follows from \cref{prop:sa-level-decrease} and
\cref{lemma:extended-crossings-uniform}, since for the $x \in \R^m$ defined in the current proof,
\[
  \Ex{ \bm{Q}^{(m)} } = \Ex{ S_\downarrow(\bm{W}_x) + S_\uparrow(\bm{W}_x) }
  \leq \Ex{Z(\bm{W}_{\vec 1})} = O(\sqrt m) \,. \qedhere
\]
\end{proof}

\begin{proposition}
\label{prop:max-walk-bound-q}
Let $X$ be a set of sequences $x \in \R^m$ each having the DSS property, and let $a : X \to \R$ be
arbitrary. Then
\[
  \Ex{ L_\downarrow(\bm{M}_X^{+a}) } \leq \Ex{ \bm{Q}^{(m)} } \,.
\]
\end{proposition}
\begin{proof}
By induction on $m$. For $m=1$, the probability that $\bm{M}_X^{+a}$ has a downwards level return is
at most $1/2$, because if $M_X^{+a}(0) \geq 0$, all of the maximizing constituent walks $x \in X$
satisfying $W_x(0)^{+a(x)} = M_X^{+a}(0)$ must decrease. If $M_X^{+a}(0) < 0$ then there is no
downwards level return for $m=1$.

Let $m > 1$ and consider two cases. First assume that $M_X^{+a}(0) \geq 0$ and let $x \in X$ be an
arbitrary constituent walk satisfying $W_x^{+a(x)}(0) = M_X^{+a}(0)$. For fixed permutation $\sigma$
and sign vector $\varepsilon$, let $s_1$ be the first
downwards level return point of $M_X^{+a}(\cdot\; ; \sigma, \varepsilon)$ and let $s'_1$ be the first
downwards level return point of $W_x^{+a(x)}(\cdot\; ; \sigma, \varepsilon)$. Note that $s'_1 \leq s_1$ since
$M_X^{+a}$ is the maximum of its constituents.

Now we may write
\[
  \Ex{ L_\downarrow(\bm{M}_X^{+a}) }
  = \sum_{s=1}^m \Pr{ \bm{s}_1 = s } \left(1 + \Exuc{}{ L_\downarrow\left(\bm{M}_{\bm{Y}}^{+ \bm{b}}
\right) }{ \bm{s}_1 = s } \right)\,,
\]
where $\bm{Y}$ denotes the set of vectors $X$ after removing the first $t$ coordinates according to
the random permutation $\bm{\sigma}$ and $\bm{b} : \bm{Y} \to \R$ is the starting point
$\bm{b}(\bm{y}) = \bm{W}_x^{+a(x)}(t)$ of each walk $\bm{y} \in \bm{Y}$ obtained from the
original vector $x \in X$ by removing the first $t$ coordinates according to $\bm{\sigma}$. By
induction, this is
\[
  \Exuc{}{ L_\downarrow\left(\bm{M}_{\bm{Y}}^{+\bm{b}} \right) }{ \bm{s}_1 = s }
  \leq \Ex{ \bm{Q}^{(m-s)} } \,.
\]
Now we write $\bm{s}_1 = \bm{s}'_1 + (\bm{s}_1 - \bm{s}'_1)$ where the second term is non-negative.
Then
\begin{align*}
  \Ex{ L_\downarrow(\bm{M}_X^{+a}) }
    &\leq \sum_{s=1}^m \Pr{ \bm{s}'_1 + (\bm{s}_1 - \bm{s}'_1) = s } \cdot
      \left( 1 + \Ex{ \bm{Q}^{(m-s)} } \right) \\
    &= \sum_{t=1}^m \Pr{ \bm{s}'_1 = t } \sum_{s=0}^{m-t} \Pruc{}{ \bm{s}_1 - \bm{s}'_1 = s
}{\bm{s}'_1 = t}
      \left( 1 + \Ex{ \bm{Q}^{(m - (s+t))} } \right) \,.
\end{align*}
The inner sum is a convex sum of terms $1 + \Ex{ \bm{Q}^{(m-(s+t))} }$ which are each bounded by $1
+ \Ex{ \bm{Q}^{(m-t)} }$ because $\Ex{ \bm{Q}^{(k)} } \geq \Ex{ \bm{Q}^{(k')} }$ when $k \geq k'$.
We also have $\Pr{ \bm{s}'_1 = t } = \Pr{ \bm{R} = t }$ due to \cref{prop:sa-first-decrease}. Therefore,
\[
\Ex{ L_\downarrow(\bm{M}_X^{+a}) }
    \leq \sum_{t=1}^m \Pr{ \bm{R} = t } \left( 1 + \Ex{ \bm{Q}^{(m-t)} } \right)
    = \Ex{ \bm{Q}^{(m)} } \,.
\]
We must now handle the case where $M_X^{+a}(0) < 0$. Let $\bm{t}$ be the smallest time where
$\bm{M}_X^{+a}(\bm{t}) \geq 0$. Then
\[
  \Ex{ L_\downarrow(\bm{M}_X^{+a}) }
  = \sum_{t=1}^m \Pr{ \bm{t} = t } \Exuc{}{ L_\downarrow\left(\bm{M}_{\bm{Y}}^{+\bm{b}}\right) }{
\bm{t} = t }  \,,
\]
where $\bm{Y}$ and $\bm{b}$ are defined similarly as before, as the sequences $X$ after removing the
first $t$ coordinates according to the random permutation $\bm{\sigma}$ and $\bm{b}(\bm{y}) =
\bm{W}_x^{+a(x)}(t)$ is where the $x$ walk ended up at time $t$. This new walk starts above 0 so the
above argument applies and the conclusion holds.
\end{proof}

We can now prove \cref{lemma:max-walk-crossing-number}.
\begin{proof}[Proof of \cref{lemma:max-walk-crossing-number}]
First observe that $C(\bm{M}_X^{+a}) \leq 2 \, C_\downarrow(\bm{M}_X^{+a}) + 1$ so
it suffices to bound $C_\downarrow(\bm{M}_X^{+a})$. By definition it holds that
$C_\downarrow(\bm{M}_X^{+a}) \leq L_\downarrow(\bm{M}_X^{+a})$, so by 
\cref{prop:max-walk-bound-q} and \cref{prop:sa-q-bound}, we have
\[
  \Ex{ C_\downarrow(\bm{M}_X^{+a}) } \leq \Ex{\bm{Q}^{(m)}} = O(\sqrt m) \,.\qedhere
\]
\end{proof}

\subsection{Lower Bound}

Recall the definition of the \emph{influence} of a set in the ternary hypercube given in
\cref{eq:influence}. In this section, we show that there exists a convex set whose influence is $\Omega(n^{3/4})$, nearly matching the upper bound given in
\cref{thm:inf-bound}. 

The construction of the high-influence set is obtained by considering the
intersection of $2^{\sqrt{n}}$ random halfspaces whose distance from the origin is
$\Theta(n^{3/4})$. This approach is inspired by \cite[Theorem 2]{Kane14}, who showed that in the Boolean hypercube, an intersection of $k$ random halfspaces with an appropriately chosen distance from the origin will have expected influence $\Omega(\sqrt{n \log k})$. In the ternary hypercube, this type of argument still works as long as $k \leq 2^{O(\sqrt{n})}$. This type of construction was also used by Nazarov \cite{Nazarov03} to show the existence of convex sets in $\mathbb{R}^n$ whose Gaussian surface area is $\Omega(n^{1/4})$, which matches the $O(n^{1/4})$ upper bound proven by Ball \cite{Ball93}.  

\subsubsection{A Convex Set with Large Influence}

We prove the following theorem.

\thminfLB*

\begin{proof} Recall that the edges of $\zpm^n$ are the directed pairs of points $(x,y)$ such that there exists $i \in [n]$ for which $x_i = 0$, $y_i \in \{\pm 1\}$ and $x_j = y_j$ for all $j \neq i$. We will use the following claim which is also used by Kane (see \cite[Lemma 7]{Kane14} and its proof). 

\begin{claim} \label{clm:t} For any $n$ and $\eps \in [2^{-n},1/2]$, there exists $\tau = \Theta(\sqrt{n \log 1/\eps})$ such that 
\[
\pr_{x \sim \{\pm 1\}^n}\left[ \sum_{i=1}^n x_i > \tau \right] \geq \eps \text{.}
\]
\end{claim}

Let $\eps = 2^{-\sqrt{n}}$ and choose $\tau = \Theta(\sqrt{n \log 1/\eps}) = \Theta(n^{3/4})$ so that
\begin{align}
    \rho := \pr_{z \sim \{\pm 1\}^{2n/3}} \left[\sum_i z_i > \tau\right] \geq \eps
\end{align}
as guaranteed by \cref{clm:t}. The main technical lemma that allows our construction to work is the following, which we prove in \cref{sec:proof_of_halfspace_density_increment_bound}. Note that this lemma crucially uses the assumption that $\tau = O(n^{3/4})$ and this is where the structure of the ternary hypercube prevents this construction from obtaining sets with influence $\gg n^{3/4}$.

\begin{lemma} 
\label{lem:halfspace_density_increment_bound} For all $n$, all $\Omega(\sqrt{n}) \leq \tau \leq O(n^{3/4})$, and all $\ell=O(\sqrt{n})$,
\[
\Pru{x \sim \{\pm 1\}^{n + \ell}}{\sum_{i=1}^{n+\ell} x_i > \tau} \leq O\left( \Pru{x \sim \{\pm 1\}^{n}}{ \sum_{i=1}^{n} x_i > \tau} \right)\text{.}
\] 
\end{lemma}

By \cref{lem:halfspace_density_increment_bound} there are constants $C_0 < 1 < C_1$ such that for all $\ell \in [-\sqrt{n},\sqrt{n}]$, we have
\begin{align} \label{eq:density}
   C_0 \rho \leq \Pru{z \sim \{\pm 1\}^{2n/3 + \ell}}{\sum_i z_i > \tau} \leq C_1 \rho\text{.}
\end{align}
Define $H := \{x \in \zpm^n \colon \sum_{i=1}^n x_i > \tau\}$. Let $L_{m} = \{x \in \zpm^n \colon \norm{x}_1 = m\}$ and note that
\[
\Pru{z \sim \{\pm 1\}^{2n/3 + \ell}}{\sum_i z_i > \tau} = \Pru{z \sim L_{2n/3+\ell}}{\sum_i z_i > \tau}
\]
and so \cref{eq:density} tells us that the density of $H$ in $L_{2n/3+\ell}$ only differs by a constant multiplicative factor for any $\ell \in [-\sqrt{n},\sqrt{n}]$. We abuse notation and write $H(x) = \mathbf{1}(x \in H)$. Now, let $E_{\sqrt{n}}$ denote the set of edges in $\zpm^n$ which have both endpoints in $\midlayers(\sqrt{n})$ and let 
\begin{align}
    \Inf_{\mathsf{mid}}(H) &:= \frac{1}{3^{n}} \cdot |\{(x,y) \in E_{\sqrt{n}} \colon H(x) \neq H(y)\}|
\end{align}
denote the influence of $H$ restricted to $\midlayers(\sqrt{n})$. We prove the following lower bound on this quantity.


\begin{claim} \label{clm:halfspace_middle_influence} $\Inf_{\mathsf{mid}}(H) = \Omega(\rho \cdot \tau)$. \end{claim}

\begin{proof} Let $\cI = [2n/3 - \sqrt{n}, 2n/3 + \sqrt{n} - 1]$ and observe that we can write
\begin{align}
    \Inf_{\mathsf{mid}}(H) &= \frac{1}{3} \sum_{i=1}^n \Exu{x \sim \{0,\pm 1\}^{n}}{\mathbf{1}(\norm{x}_1 \in \cI) \cdot \left(|H(x^{i \gets 1})-H(x^{i \gets 0})| + |H(x^{i \gets 0})-H(x^{i \gets -1})|\right)} \nonumber \\
    &= \frac{1}{3}\sum_{i=1}^n \Exu{x \sim \{0,\pm 1\}^{n}}{\mathbf{1}(\norm{x}_1 \in \cI)  \cdot (H(x^{i \gets 1})-H(x^{i \gets -1}) )} \label{eq:monotone} \\
    &= \frac{1}{3} \Exu{x \sim \zpm^n}{\sum_{i=1}^n \mathbf{1}(\norm{x}_1 \in \cI) \left(H(x^{i \gets 1}) - H(x^{i \gets -1}) \right)} \nonumber \\
    &= \frac{1}{3} \Exu{z \sim \zpm^n}{\sum_{i : z_i = 1} \mathbf{1}(\norm{z^{i \gets 0}}_1) H(z)
                                       - \sum_{i : z_i = -1} \mathbf{1}(\norm{z^{i \gets 0}}_1) H(z) }
    \label{eq:influence-lb-rearrange}
\end{align}
where \cref{eq:monotone} holds because $H$ is a monotone function with respect to the standard partial order, i.e., if $x_i \leq y_i$ for all $i \in [n]$, then $x \in H$ implies $y \in H$, and \cref{eq:influence-lb-rearrange} follows by the observation that for each $i$, $z \in \zpm^n$ appears in the sum as $H(z)$ whenever $x = z^{i \gets 0}$ and $z_i = 1$, and appears in the sum as $-H(z)$ whenever $x = z^{i \gets 0}$ and $z_i = -1$.
Let $\cI' = [2n/3-\sqrt{n}+1,2n/3+\sqrt{n}-1]$. Thus, the above expression can be rewritten as 
\begin{align} \label{eq:Imid}
    \Inf_{\mathsf{mid}}(H) 
    &= \frac{1}{3} \Exu{z \sim \{0,\pm 1\}^n}{H(z)\sum_{i=1}^n z_i \cdot \mathbf{1}(\norm{z^{i \gets 0}}_1 \in \cI)} \nonumber \\
    &\geq \frac{1}{3} \Exu{z \sim \{0,\pm 1\}^n}{H(z)\mathbf{1}(\norm{z}_1 \in \cI')\sum_{i=1}^n z_i} > \frac{\tau}{3} \Exu{z \sim \{0,\pm 1\}^n}{H(z)\mathbf{1}(\norm{z}_1 \in \cI')}
\end{align}
where the first inequality holds since $\cI' \subset \cI$ and $\norm{z}_1 \in \cI'$ implies $\norm{z^{i \gets 0}}_1 \in \cI$ for all $i \in [n]$. The second inequality holds since $H(z) = 1$ if and only if $\sum_i z_i > \tau$. Finally,
\begin{align}
     \Exu{z \sim \{0,\pm 1\}^n}{H(z)\mathbf{1}(\norm{z}_1 \in \cI')} &= \frac{1}{3^n} \cdot \sum_{\ell = -\sqrt{n}+1}^{\sqrt{n}-1} \sum_{z \in L_{2n/3 + \ell}} H(z) \nonumber \\
     &= \sum_{\ell = -\sqrt{n}+1}^{\sqrt{n}-1} \frac{{n \choose 2n/3 + \ell} \cdot 2^{2n/3+\ell}}{3^n}\cdot \Pru{z \sim \{\pm 1\}^{2n/3 + \ell}}{\sum_i z_i > \tau}
\end{align}
and the quantity in the RHS is $\Omega(\rho)$ by the lower bound in \cref{eq:density} and since ${n \choose 2n/3 + \ell} \cdot 2^{2n/3+\ell} = \Omega(3^n/\sqrt{n})$ for all $\ell \in [-\sqrt{n},\sqrt{n}]$ by an application of Stirling's approximation. Combining this with \cref{eq:Imid}, we conclude that $\Inf_{\mathsf{mid}}(H) = \Omega(\rho \cdot \tau)$ as claimed. \end{proof}

Let $k := \max\{\lfloor (4C_1 \rho)^{-1} \rfloor,1\} \leq \eps^{-1}$. Choose $v^{(1)},\ldots,v^{(k)} \in \{\pm 1\}^n$ i.i.d.\ uniformly at random and for each $i \in [k]$ define $H_i = \{x \in \zpm^n \colon \langle x, v^{(i)} \rangle > \tau\}$. Let $S = \cap_{i=1}^k \overline{H_i}$ be the convex set formed by the intersection of the complements of the $H_i$'s. Note that $\overline{S} = \cup_{i=1}^k H_i$ and $\Inf(S) = \Inf(\overline{S})$ and thus it suffices to give a lower bound on $\Inf(\overline{S})$. 
Observe that every edge $(x,y)$ that is influential for $H_i$ is guaranteed to be influential for $\overline{S}$ if $x,y \notin H_j$ for all $j \neq i \in [k]$. Moreover, if $\norm{x}_1 = \frac{2n}{3} + \ell$ where $\ell \in [-\sqrt{n},\sqrt{n}]$, then 
\[
\Pru{v^{(j)} \sim \{\pm 1\}^n}{x \in H_j} = \Pru{v^{(j)} \sim \{\pm 1\}^n}{\langle v^{(j)}, x \rangle > \tau} = \Pru{z \sim \{\pm 1\}^{2n/3 + \ell}}{\sum_i z_i > \tau} \leq C_1 \cdot \rho \text{.}
\]
Thus, by a union bound, the probability that a $H_i$-influential edge $(x,y)$ with $x,y \in \midlayers(\sqrt{n})$ remains influential for $\overline{S}$ is at least $1 - 2(k-1) \cdot C_1 \rho \geq 1 - 2 \cdot (4C_1\rho)^{-1} \cdot C_1\rho \geq 1/2$. Therefore,
\begin{align}
    \Exu{v^{(1)},\ldots,v^{(m)}}{\Inf(\overline{S})} \geq \frac{1}{2} \sum_{i=1}^k \Inf_{\mathsf{mid}}(H_i) = \Omega(k \cdot \rho \cdot \tau) = \Omega(\tau) = \Omega\left(n^{3/4}\right)
\end{align}
and this completes the proof. \end{proof}

\subsubsection{Bounding the Density Increment of Halfspaces\texorpdfstring{: Proof of \cref{lem:halfspace_density_increment_bound}}{}} \label{sec:proof_of_halfspace_density_increment_bound}

Our proof makes crucial use of the following tight bound on the binomial coefficient ${n \choose \frac{n-\tau}{2}}$ for any $\tau = O(n^{3/4})$. Importantly, the bound is tight up to constant multiplicative factors, as opposed to constant factors in the exponent.

\begin{fact} \label{fact:special_case_3/4} If $\tau = O(n^{3/4})$, then ${n \choose \frac{n-\tau}{2}} = \Theta\left(\frac{2^n}{\sqrt{n}} \cdot \exp(-\frac{\tau^2}{2n})\right)$. \end{fact}

\Cref{fact:special_case_3/4} is a special case of a much more general approximation, \cref{cor:middle_binomial_bound}, specifically the case of $s=2$. The proof is relatively tedious and so we relegate it to \cref{sec:approx-bin-coeff}. Using \cref{fact:special_case_3/4} we are able to prove the following claim which is important for the proof of \cref{lem:halfspace_density_increment_bound}.

\begin{claim} \label{clm:binomials_diff_n} For all $n$, all $\Omega(\sqrt{n}) \leq \tau \leq O(n^{3/4})$, and all $\ell = O((n/\tau)^2)$, 
\[
2^{-(n+\ell)}{n+\ell \choose \frac{n+\ell-\tau}{2}} \leq O\left( 2^{-n} {n \choose \frac{n-\tau}{2}} \right)\text{.}
\] 
\end{claim}

\begin{proof} Since $\tau = O(n^{3/4})$, by \cref{fact:special_case_3/4}
\begin{align}
    2^{-(n+\ell)}{n+\ell \choose \frac{n+\ell-\tau}{2}} \leq \frac{1}{\Theta(\sqrt{n+\ell})} \exp\left(-\frac{\tau^2}{2(n+\ell)}\right) = \frac{1}{\Theta(\sqrt{n})} \exp\left(-\frac{\tau^2}{2n(1+\ell/n)}\right). \nonumber
\end{align}
Observe that 
\begin{align}
    -\frac{\tau^2}{2n(1+\ell/n)} = -\frac{\tau^2}{2n}\left(1-\frac{\ell}{n+\ell}\right) = -\frac{\tau^2}{2n} + \frac{\tau^2 \ell}{2n(n+\ell)} = -\frac{\tau^2}{2n} + O(1) \nonumber
\end{align}
since $\ell = O((n/\tau)^2)$. Therefore,
\[
2^{-(n+\ell)}{n+\ell \choose \frac{n+\ell-\tau}{2}} \leq \frac{1}{\Theta(\sqrt{n})} \exp\left(-\frac{\tau^2}{2n}\right) \leq O\left( 2^{-n} {n \choose \frac{n-\tau}{2}}\right)
\]
where the second inequality is by \cref{fact:special_case_3/4} since $\tau = O(n^{3/4})$. \end{proof}

We are now set up to prove \cref{lem:halfspace_density_increment_bound}.

\begin{proof}[Proof of \cref{lem:halfspace_density_increment_bound}.] 
First, write
\begin{align} \label{eq:main}
    \Pru{x \sim \{\pm 1\}^{n + \ell}}{\sum_{i=1}^{n+\ell} x_i > \tau} = \Pru{x \sim \{\pm 1\}^{n + \ell}}{\tau < \sum_{i=1}^{n+\ell} x_i \leq 2\tau} + \Pru{x \sim \{\pm 1\}^{n + \ell}}{\sum_{i=1}^{n+\ell} x_i > 2\tau}\text{.}
\end{align}
By Hoeffding's inequality, the second term is
\begin{align} \label{eq:Hoeffding1}
    \Pru{x \sim \{\pm 1\}^{n + \ell}}{\sum_{i=1}^{n+\ell} x_i > 2\tau} &\leq \exp\left(-\frac{2 \cdot (2\tau)^2}{4(n+\ell)}\right) \nonumber \\
    &= \exp\left(-\frac{2\tau^2}{n}\left(1-\frac{1}{(n/\ell) + 1}\right)\right) = O\left(\exp\left(-\frac{2\tau^2}{n}\right)\right)
\end{align}
since $\frac{\tau^2\ell}{n^2} = O(1)$. Using \cref{clm:binomials_diff_n}, the first term is
\begin{align} \label{eq:first_term}
    \Pru{x \sim \{\pm 1\}^{n + \ell}}{\tau < \sum_{i=1}^{n+\ell} x_i \leq 2\tau} &= \frac{1}{2^{n+\ell}}\sum_{\tau' \in (\tau,2\tau] \text{ even}} {n + \ell \choose \frac{n+\ell-\tau'}{2}} \nonumber \\
    &\leq O \left( \frac{1}{2^n} \sum_{\tau' \in (\tau,2\tau] \text{ even}} {n \choose \frac{n-\tau'}{2}} \right)
    = O\left( \Pru{x \sim \{\pm 1\}^{n}}{\tau < \sum_{i=1}^{n} x_i \leq 2\tau} \right)\text{.}
\end{align}
We now just need to show that the first term dominates the second term. By \cref{fact:special_case_3/4}
\begin{align} \label{eq:second_term}
    \Pru{x \sim \{\pm 1\}^{n}}{\tau < \sum_{i=1}^{n} x_i \leq 2\tau} &= \frac{1}{2^n} \sum_{\tau' \in (\tau,2\tau] \text{ even}} {n \choose \frac{n-\tau'}{2}} \nonumber \\
    &\geq \Omega\left(\frac{\tau}{\sqrt{n}} \exp\left(-\frac{2\tau^2}{n}\right)\right) = \Omega\left(\exp\left(-\frac{2\tau^2}{n}\right)\right) \text{.}
\end{align}
Plugging the bounds from \cref{eq:Hoeffding1}, \cref{eq:first_term}, and \cref{eq:second_term} back into \cref{eq:main} yields
\[
     \Pru{x \sim \{\pm 1\}^{n + \ell}}{\sum_{i=1}^{n+\ell} x_i > \tau} \leq O \left( \Pru{x \sim \{\pm 1\}^{n}}{\tau < \sum_{i=1}^{n} x_i \leq n^{3/4}} \right) . \qedhere
\]
\end{proof}

\section{Sample-Based Testing and Learning}

In this section we prove upper and lower bounds for testing and learning
convex sets on $\zpm^n$ with samples.

\subsection{Upper Bound}
\thmlearner*

Our proof of \cref{thm:learner} uses the standard approach of showing that when $S$ is convex, its low-degree Fourier coefficients contain most of the information about $S$. We can then learn $S$ by estimating its low-degree Fourier coefficients. This learning approach was established by Linial, Mansour, and Nisan, and is referred to as the "Low-Degree Algorithm" \cite{LinialMN93}. 
The section is organized as follows:
\begin{enumerate}
    \item \cref{section:fourier-setup}: Setup of the Fourier analysis over the ternary hypercube that will be necessary for the learning result.
    \item  \cref{section:fourier-concentration}: Bounds on the Fourier concentration of convex sets, using the influence bounds from \cref{section:influence}.
    \item \cref{section:low-degree-algorithm}: The low-degree learning algorithm and the proof of \cref{thm:learner} and \cref{cor:UB-2sidedsamples}.
\end{enumerate}

\subsubsection{Fourier Analysis Setup over the Ternary Hypercube}
\label{section:fourier-setup}

This subsection uses mostly standard techniques, following Chapter 8 of \cite{O14} which outlines how to generalize Fourier analysis of Boolean functions to arbitary product spaces.

Let $\pi_{1/3}$ and $\pi_{1/3}^{\otimes n}$ denote the uniform distribution over $\zpm$ and $\zpm^n$, respectively. Let $L^2(\zpm^n, \pi_{1/3}^{\otimes n})$ denote the real inner product space of functions $f \colon \zpm^n \to \RR$ with inner product $\langle f, g \rangle = \Exp_x[f(x)g(x)]$.

\begin{definition} A \emph{Fourier basis} for $L^2(\zpm, \pi_{1/3})$ is an orthonormal basis $\phi_{-1},\phi_0,\phi_{+1} \colon \{0,\pm 1\}^n \to \mathbb{R}$ with $\phi_0 \equiv 1$. \end{definition}

An important message in Chapter 8 of \cite{O14} is that the specific choice of Fourier basis does not matter. For concreteness, we can use the following basis given in Example 8.10 of \cite{O14}.

\begin{definition} \label{def:basis} Define the following Fourier basis for $L^2(\zpm, \pi_{1/3})$: $\phi_0 \equiv 1$, 
\[
    (\phi_{-1}(-1), \phi_{-1}(0), \phi_{-1}(1)) = (-\sqrt{6}/2, 0, \sqrt{6}/2) \text{, and } (\phi_{1}(-1), \phi_1(0), \phi_{1}(1)) = (-\sqrt{2}/2, \sqrt{2}, -\sqrt{2}/2 ) \text{.}
\]
\end{definition}

It can be easily confirmed that the basis in \cref{def:basis} is orthonormal and so is indeed a Fourier basis. Now, given $\alpha \in \zpm^n$, we define $\phi_{\alpha} \in L^2(\zpm^n, \pi_{1/3}^{\otimes n})$ as 
\begin{align}
    \phi_{\alpha}(x) := \prod_{i = 1}^n \phi_{\alpha_i}(x_i)
\end{align}
An immediate corollary of Proposition 8.13 from \cite{O14} is that $(\phi_{\alpha})_{\alpha \in \zpm^n}$ is a Fourier basis for $L^2(\zpm^n, \pi_{1/3}^{\otimes n})$. I.e., $\phi_{(0,0,\ldots,0)} \equiv 1$ and this basis is orthonormal. Definition 8.14 of \cite{O14} now asserts that every function $f \colon \zpm^n \to \RR$ can be written as 
\begin{align}
    f(x) = \sum_{\alpha \in \zpm^n} \widehat{f}(\alpha)\phi_{\alpha}(x) \text{ where } \widehat{f} = \langle f, \phi_{\alpha} \rangle = \Exp_x[f(x)\phi_{\alpha}(x)]
\end{align}
is the Fourier coefficient of $f$ on $\alpha$. 

\subsubsection{Fourier Concentration for Convex Sets}
\label{section:fourier-concentration}

We use the notation $\# \alpha := |\{i \colon \alpha_i \neq 0\}|$. Our goal is now to prove the following fact about the Fourier coefficients of convex sets. Here we abuse notation and use $S \colon \zpm^n \to \{\pm 1\}$ defined as $S(x) = (-1)^{\mathbf{1}(x \notin S)}$ to denote membership in the set $S$.

\begin{lemma} [Fourier Concentration for Convex Sets] \label{lem:concentration} There exists a constant $C > 0$ such that for any convex set $S \subseteq \zpm^n$ and $\eps > 0$,
\[
\sum_{\alpha \colon \# \alpha > \frac{C}{\eps}n^{3/4}\log^{1/4}n} \widehat{S}(\alpha)^2 \leq \eps \text{.}
\]
\end{lemma}

\begin{proof} The idea is to make use of our upper bound on the influence of convex sets from \cref{thm:inf-bound}. We will need the following slightly different definition of the influence given by \cite{O14}. We will refer to this slightly different notion as \emph{Fourier influence} and will denote it by $\Inf^{\mathsf{Fourier}}(f)$ for clarity. We will show that these definitions are equivalent up to a constant factor. Below, for $x \in \zpm^n$, $i \in [n]$, and
$b \in \zpm$, we write $x^{i \gets b}$ for the vector obtained from $x$ by setting $x_i$ to $b$.

\begin{definition} [Def. 8.17 and 8.22, \cite{O14}] \label{def:operators} For $f \in L^2(\zpm^n, \pi_{1/3}^{\otimes n})$ and $i \in [n]$, the \emph{projection of $f$ onto $i$} is 
\[
\text{E}_i f(x) = \Exp_{b \in \{0,\pm 1\}}[f(x^{i \gets b})] \text{.}
\]
The \emph{$i$'th coordinate Laplacian operator} $\text{L}_i$ is defined as $\text{L}_i f := f - \text{E}_i f$. The \emph{Fourier influence of coordinate $i$ on $f$ is } $\Inf_i^{\mathsf{Fourier}}(f) = \langle f, \text{L}_i f \rangle$. The \emph{total Fourier influence of $f$} is $\Inf^{\mathsf{Fourier}}(f) = \sum_{i=1}^n \Inf_i^{\mathsf{Fourier}}(f)$.
\end{definition}

We will need the two following identities.

\begin{proposition} [Prop. 8.16 and Prop. 8.23, \cite{O14}] \label{prop:identities} Every $f \in L^2(\zpm^n, \pi_{1/3}^{\otimes n})$ satisfies the following two identities:
\begin{enumerate}
    \item $\sum_{\alpha \in \zpm^n} \widehat{f}(\alpha)^2 = \Exp[f^2]$
    \item $\Inf^{\mathsf{Fourier}}(f) = \sum_{\alpha \in \zpm^n} \# \alpha \cdot \widehat{f}(\alpha)^2$
\end{enumerate}
\end{proposition}

\begin{lemma} [Fourier Concentration from Influence] \label{lem:degree-bound} Let $\cF \colon \zpm^n \to \{\pm 1\}$ be a class of functions with Fourier influence upper bounded by $\Inf^{\mathsf{Fourier}}(f) \leq B$ for all $f \in \cF$. Then, for any $\eps > 0$, we have
\[
\sum_{\alpha \in \zpm^n \colon \# \alpha > B/\eps} \widehat{f}(\alpha)^2 \leq \eps \text{.}
\]
\end{lemma}
\begin{proof} By item (1) of \cref{prop:identities}, we have $\sum_{\alpha} f(\alpha)^2 = 1$. Thus, $\sum_{\alpha} \# \alpha \cdot \widehat{f}(\alpha)^2$ is the expectation of $\# \alpha$ when $\alpha$ is sampled with probability $\widehat{f}(\alpha)^2$. By item (2) of \cref{prop:identities} we have
\[
\sum_{\alpha \in \zpm^n} \# \alpha \cdot \widehat{f}(\alpha)^2 = \Inf^{\mathsf{Fourier}}(f) \leq B
\]
and now applying Markov's inequality yields the desired inequality. \end{proof}

\begin{fact} [Equivalence of Influence Definitions] \label{fact:inf-def} For every $f \colon \zpm^n \to \{\pm 1\}$, we have 
\[
\frac{3}{8} \cdot \Inf^{\mathsf{Fourier}}(f) \leq \Inf(f) \leq \frac{3}{4} \cdot \Inf^{\mathsf{Fourier}}(f)\text{.}
\]
\end{fact}

\begin{proof} Let $\Delta_i(f)$ denote the number of lines in the ternary hypercube of the form $(x^{i \gets -1}, x^{i \gets 0}, x^{i \gets 1})$ such that $f_{\{x^{i \gets -1}, x^{i \gets 0}, x^{i \gets 1}\}}$ is not constant, and let $\Delta(f) = \sum_i \Delta_i$ be the total number of such lines. Recall the definition of $\Inf(f)$ from \cref{eq:influence} and observe that since every such line contains either $1$ or $2$ influential edges, we have $\Delta(f) \cdot 3^{-n} \leq \Inf(f) \leq 2\Delta(f) \cdot 3^{-n}$. We will show that $\Inf^{\mathsf{Fourier}}(f) = \frac{8}{3} \Delta(f) \cdot 3^{-n}$ and combining these observations completes the proof. We have
\begin{align}
    \Inf_i^{\mathsf{Fourier}}(f) = \langle f, \text{L}_i f\rangle = \Exp_x\left[f(x)(f(x)-\text{E}_i f(x))\right] = \Exp_x\left[1 - f(x)\Exp_{b \in \{0,\pm 1\}}[f(x^{i \gets b})]\right] \text{.}  
\end{align}
Now, for a fixed $x$, consider the line in dimension $i$, containing $x$: $\ell_i(x) := (x^{i \gets -1}, x^{i \gets 0}, x^{i \gets 1})$. Observe that if $f$ is constant on $\ell_i(x)$, then $1 - f(x)\Exp_{b \in \{0,\pm 1\}}[f(x^{i \gets b})] = 0$. If $f$ is non-constant on $\ell_i(x)$, then either it contains two $+1$'s and one $-1$ or vice versa. In both cases we have
\begin{align}
    \Exp_{a \in \{0,\pm 1\}}\left[1 - f(x^{i \gets a})\Exp_{b \in \{0,\pm 1\}}[f(x^{i \gets b})]\right] = 1 - \Exp_{b \in \zpm}[f(x^{i\gets b})]^2 = 8/9 \text{.}
\end{align}
Therefore, 
\begin{align*}
    \Inf_i^{\mathsf{Fourier}}(f) &= \Exp_x\left[1 - f(x)\Exp_{b \in \{0,\pm 1\}}[f(x^{i \gets b})]\right] = \Exp_x\Exp_{a \in \{0,\pm 1\}}\left[1 - f(x^{i \gets a})\Exp_{b \in \{0,\pm 1\}}[f(x^{i \gets b})]\right] \\
        &= \frac{8}{9} \Exp_x \left[\mathbf{1}(f|_{\ell_i(x)} \text{ is not constant})\right] = \frac{8}{3} \Delta_i(f) \cdot 3^{-n} 
\end{align*}
and summing over all $i$ completes the proof. \end{proof}

Combining \cref{thm:inf-bound}, \cref{lem:degree-bound}, and \cref{fact:inf-def} completes the proof of \cref{lem:concentration}. \end{proof}

\subsubsection{Low-Degree Learning Algorithm\texorpdfstring{ and Proof of \cref{thm:learner}}{}}
\label{section:low-degree-algorithm}

Recall that we are using the basis for the space of functions $f\colon \zpm^n \to \mathbb{R}$ given by 
\begin{align}
    \phi_{\alpha}(x) := \prod_{i = 1}^n \phi_{\alpha_i}(x_i) \text{ ~for every~ } \alpha \in \zpm^n
\end{align}
where $\alpha_{-1},\alpha_0,\alpha_1$ are a basis for the space of functions $f \colon \zpm \to \mathbb{R}$ defined as: $\phi_0 \equiv 1$, 
\[
    (\phi_{-1}(-1), \phi_{-1}(0), \phi_{-1}(1)) = (-\sqrt{6}/2, 0, \sqrt{6}/2) \text{, and } (\phi_{1}(-1), \phi_1(0), \phi_{1}(1)) = (-\sqrt{2}/2, \sqrt{2}, -\sqrt{2}/2 ) \text{.}
\]
Recall that $\widehat{f}(\alpha) = \Exp_x[f(x)\phi_{\alpha}(x)]$. Our learning upper bound \cref{thm:learner} follows immediately by combining \cref{lem:concentration} with the following theorem.

\begin{theorem} [Low-Degree Algorithm over $\zpm^n$] \label{thm:low-degree} Let $\cF \colon \zpm^n \to \{\pm 1\}$ be a class of functions such that for $\eps > 0$ and $\tau = \tau(\eps,n)$,
\[
\sum_{\alpha \in \zpm^n \colon \# \alpha > \tau} \widehat{f}(\alpha) \leq \eps \text{.}
\]
Then $\cF$ can be learned with time and sample complexity $\text{poly}(n^{\tau},1/\eps)$. \end{theorem}

\begin{proof} Let $A := \{\alpha \in \zpm^n \colon \# \alpha \leq \tau\}$. Note that $|A| = \sum_{\Delta=0}^{\tau} {n \choose \Delta} \cdot 2^{\Delta} = \text{poly}(n^{\tau})$. We take $s$ samples $x_1,\ldots,x_s \in \zpm^n$ where $s$ will be chosen later. For each $\alpha$, we use the empirical estimate $Z_\alpha := 
\frac{1}{s} \sum_{i=1}^s f(x_i) \phi_{\alpha}(x_i)$ and return the hypothesis
\begin{align} \label{eq:hypothesis}
    h(x) = \sgn\left(\sum_{\alpha \in A} Z_{\alpha} \phi_{\alpha}(x) \right)\text{.}
\end{align}
Consider the event that $|Z_{\alpha} - \widehat{f}(\alpha)| \leq \sqrt{\eps/|A|}$ for all $\alpha$. We first show that this event occurs with high probability, and then show that if it occurs, then $h$ is a good hypothesis.

\begin{claim} \label{clm:probE} Set $s := \frac{3|A|^2}{\eps} = \text{poly}(n^{\tau},1/\eps)$. Then 
\[
\pr_{x_1,\ldots,x_s}\left[|Z_{\alpha} - \widehat{f}(\alpha)| \leq \sqrt{\eps/|A|} \text{ for all } \alpha \in A\right] \text{.}
\]
\end{claim}

\begin{proof} Fix any $\alpha \in A$ and observe that $Z_{\alpha} = \frac{1}{s} \sum_{i=1}^s X_i$ where $X_1,\ldots,X_s$ are independent copies of $X = f(x)\phi_{\alpha}(x)$ for $x \sim \zpm^n$ drawn uniformly at random. Note that in the setting of the Boolean hypercube when one uses the standard basis of parity functions, the random variable $X$ always lies in $\{\pm 1\}$. In the ternary hypercube this is not the case and in fact $|\phi_{\alpha}(x)|$ can be exponentially large. For instance $\phi_{\vec{1}}(\vec{0}) = (\sqrt{2})^n$. However, since $\langle \phi_{\alpha}, \phi_{\alpha} \rangle = 1$ for all $\alpha$, we're able to show that $\Variance(X) \leq 1$ and this allows us to obtain good estimates for $\widehat{f}(\alpha)$. We have
\[
\sigma^2 := \Variance(X) = \Exp_x[(f(x)\phi_{\alpha}(x))^2] - \Exp_x[f(x)\phi_{\alpha}(x)]^2 = 1 - \widehat{f}(\alpha)^2
\]
by definition of the Fourier coefficient $\widehat{f}(\alpha)$ and the fact that our basis is orthonormal and so in particular $\Exp_x[\phi_{\alpha}(x)^2] = \langle \phi_{\alpha}, \phi_{\alpha} \rangle = 1$. Therefore, $\Variance(Z_{\alpha}) = \frac{\sigma^2}{s} = \frac{1-\widehat{f}(\alpha)^2}{s} \leq \frac{1}{s}$. Note also that $\Exp[Z_{\alpha}] = \widehat{f}(\alpha)$. Now, by Chebyshev's inequality, we have
\[
\pr\left[|Z_{\alpha} - \widehat{f}(\alpha)| \geq \frac{k}{\sqrt{s}}\right] \leq \pr\left[|Z_{\alpha} - \widehat{f}(\alpha)| \geq k \sqrt{\Variance(Z_{\alpha})}\right] \leq \frac{1}{k^2}.
\]
Setting $k = \sqrt{3|A|}$ and recalling $s = \frac{3|A|^2}{\eps}$ yields
\[
\pr\left[|Z_{\alpha} - \widehat{f}(\alpha)| \geq \sqrt{\eps/|A|}\right] = \pr\left[|Z_{\alpha} - \widehat{f}(\alpha)| \geq \frac{k}{\sqrt{s}}\right] \leq \frac{1}{3|A|}
\]
and taking a union bound over all $\alpha \in A$ completes the proof of the claim. \end{proof}

\begin{claim} \label{clm:good_h} Using the definition of $h$ in \cref{eq:hypothesis}, if $|Z_{\alpha} - \widehat{f}(\alpha)| \leq \sqrt{\eps/|A|}$ for all $\alpha \in A$, then $\pr_{x \sim \zpm^n}[h(x) \neq f(x)] \leq \eps$. \end{claim}

\begin{proof} First, observe that 
\begin{align} \label{eq:initial}
    \pr_x\left[f(x) \neq h(x)\right] = \frac{1}{4}\Exp_x\left[(f(x)-h(x))^2\right] \text{.}
\end{align}
Now, if $f(x) \neq h(x)$, then $\left(f(x) - \sum_{\alpha \in A} Z_{\alpha} \phi_{\alpha}(x)\right)^2 \geq 1 = \frac{1}{4}(f(x) - h(x))^2$. Clearly if $f(x) = h(x)$, then this inequality also holds. Thus, for any $x \in \zpm^n$, this inequality holds. Combining this observation with \cref{eq:initial} yields
\begin{align} \label{eq:second}
    \pr_x\left[f(x) \neq h(x)\right] &\leq \Exp_x\left[\left(f(x) - \sum_{\alpha \in A} Z_{\alpha}\phi_\alpha(x) \right)^2\right] \text{.}
\end{align}
In the next calculation, for $\alpha \notin A$, let $Z_{\alpha} := 0$. Now, writing $f(x) = \sum_{\alpha} \widehat{f}(\alpha) \phi_{\alpha}(x)$, expanding the squared sum, applying linearity of expectation, and using the fact that $\Exp_x[\phi_{\alpha}(x)\phi_{\alpha'}(x)] = \langle \phi_{\alpha}, \phi_{\alpha'}\rangle = 0$ for any $\alpha \neq \alpha'$, we get
\begin{align} \label{eq:third}
    \Exp_x\left[\left(\sum_{\alpha}\phi_{\alpha}(x)\left(\widehat{f}(\alpha) - Z_{\alpha}\right) \right)^2\right] &= \Exp_x\left[\sum_{\alpha,\alpha'} \phi_{\alpha}(x) \phi_{\alpha'}(x) \left(\widehat{f}(\alpha) - Z_{\alpha}\right)\left(\widehat{f}(\alpha') -Z_{\alpha'}\right)\right] \nonumber \\
    &= \sum_{\alpha,\alpha'} \langle \phi_{\alpha},\phi_{\alpha'} \rangle \left(\widehat{f}(\alpha) - Z_{\alpha}\right)\left(\widehat{f}(\alpha') -Z_{\alpha'}\right) \nonumber \\
    &= \sum_{\alpha} \left(\widehat{f}(\alpha) - Z_{\alpha}\right)^2 \text{.}
\end{align}
Finally, using \cref{eq:second}, \cref{eq:third}, and the fact that $|\widehat{f}(\alpha) - Z_{\alpha}| \leq \sqrt{\eps/|A|}$ for $\alpha \in A$ and $\sum_{\alpha \notin A} \widehat{f}(\alpha)^2 \leq \eps$, yields
\begin{align*}
    \pr_x\left[f(x) \neq h(x)\right] &\leq \sum_{\alpha} \left(\widehat{f}(\alpha) - Z_{\alpha}\right)^2 = \sum_{\alpha \in A} \left(\widehat{f}(\alpha) - Z_{\alpha}\right)^2 + \sum_{\alpha \notin A} \widehat{f}(\alpha)^2 \leq |A| \cdot \frac{\eps}{|A|} + \eps = 2\eps
\end{align*}
and this completes the proof of the claim. \end{proof}

Combining \cref{clm:probE} and \cref{clm:good_h} completes the proof of \cref{thm:low-degree}. \end{proof}

\subsection{Lower Bound}

In this section we prove the following lower bound on the \emph{sample complexity} of convexity testing in the ternary hypercube.

\thmtwosidedsamplestestingLB*

Our proof of \cref{thm:2sidedsamplestestingLB} follows the standard approach of defining a pair of distributions $\Dyes,\Dno$ over subsets of $\zpm^n$ such that the following hold:
\begin{itemize}
    \item $\Dyes$ is supported over convex sets.
    \item Sets drawn from $\Dno$ are typically far from convex: $\pr_{S \sim \Dno}[\eps(S) = \Omega(1)] = \Omega(1)$.
    \item The distributions over labeled examples from $\Dyes$ and $\Dno$ are close in TV-distance.
\end{itemize}

\subsubsection{The Distributions \texorpdfstring{$\Dyes$ and $\Dno$}{}} \label{sec:distributions}

Our construction uses a variant of random Talagrand DNFs adapted to the case of testing convexity in the ternary hypercube, $\zpm^n$. In particular, our construction is inspired by the approach of \cite{BeBl16} and \cite{Chen17} to prove lower bounds for testing monotonicity of functions on the Boolean hypercube, $\{0,1\}^n$.

Let $N = 3^{\sqrt{n}}$ and choose $N$ \emph{terms} $t^{(1)},\ldots,t^{(N)} \in \zpm^n$ i.i.d.\ according the following distribution. For each $i \in [N]$:
\begin{enumerate}
    \item Form a (multi)-set $T_i$ by taking $\sqrt{n}$ independent uniform samples from $[n]$.
    \item For each $a \in T_i$, set $t^{(i)}_a \in \{\pm 1\}$ uniformly at random. For each $a \notin T_i$, set $t^{(i)}_a = 0$.
\end{enumerate}
Let $\pmb{t} = (t^{(1)},\ldots,t^{(N)})$ denote the random sequence of terms. Recall the \emph{outward-oriented poset} (\Cref{def:oop}) over $\zpm^n$. For each $i \in [N]$, let 
\begin{align} \label{eq:buckets}
    U_{i} := \left\{x \in \midlayers(\sqrt{n}) \colon x \succeq t^{(i)} \text{ and } x \not\succeq t^{(j)} \text{ for all } j \in [N] \setminus \{i\}\right\}
\end{align}
denote the set of points in the middle layers of the ternary hypercube which satisfy the $i$'th term, uniquely. Let $U = \cup_{i=1}^N U_i$ denote the set of points which satisfy a unique term. 

Sets drawn from $\Dyes$ are generated as follows. Choose a uniform random assignment $\pmb{\phi} \colon [N] \to \{0,1\}$. For every $x \in \midlayers(\sqrt{n})$ define
\[
    S_{\pmb{t},\pmb{\phi}}(x) = 
    \begin{cases}
        1, & \text{if } \forall i \in [N] \text{, } x \not\succeq t^{(i)} \\
        0, & \text{if } \exists i \neq j \in [N]\text{, } x \succeq t^{(i)} \text{ and } x \succeq t^{(j)} \\
        \pmb{\phi}(i), & \text{if } x \in U_{i}\text{.}
    \end{cases}
\]
Sets drawn from $\Dno$ are generated as follows. Choose a uniform random function $\pmb{r} \colon U \to \{0,1\}$. For each $x \in \midlayers(\sqrt{n})$ define

\[
    S_{\pmb{t},\pmb{r}}(x) = 
    \begin{cases}
        1, & \text{if } \forall i \in [N] \text{, } x \not\succeq t^{(i)} \\
        0, & \text{if } \exists i \neq j \in [N]\text{, } x \succeq t^{(i)} \text{ and } x \succeq t^{(j)}\\
       \pmb{r}(x), & \text{if } x \in U \text{.}
        \end{cases}
\]
For $x \notin \midlayers(\sqrt{n})$: if $x \in \inlayers(\sqrt{n})$, then both the yes and no distributions assign value $1$ and if $x \in \outlayers(\sqrt{n})$, then both the yes and no distributions assign value $0$.

\Cref{thm:2sidedsamplestestingLB} follows immediately by combining the following three lemmas.

\begin{lemma} \label{lem:yes} Every set in the support of $\Dyes$ is convex. \end{lemma}

\begin{proof} Let $S_{\pmb{t},\pmb{\phi}} \subseteq \zpm^n$ be any set drawn from $\Dyes$. We observe that $S_{\pmb{t},\pmb{\phi}}$ is non-increasing with respect to the \emph{outward-oriented poset} (recall \cref{def:oop}). Suppose $y \notin S_{\pmb{t},\pmb{\phi}}$ and let $x \in \upset(y)$ (recall \cref{def:upset}). We have three cases depending on where $y$ lies.
\begin{itemize}
    \item $y \in \outlayers(\sqrt{n})$: in this case $x \in \outlayers(\sqrt{n})$ as well and so $x \notin S_{\pmb{t},\pmb{\phi}}$.
    \item $y \succeq t^{(i)},t^{(j)}$ for two terms $i \neq j \in [N]$: in this case we have $x \succeq y \succeq t^{(i)},t^{(j)}$ and so $x \notin S_{\pmb{t},\pmb{\phi}}$. 
    \item $y \in U_i$ for some $i \in [N]$ and $\pmb{\phi}(i) = 0$: in this case we have $x \succeq y \succeq t^{(i)}$ and so either (i) $x \in U_i$, (ii) there exists $j \neq i \in [N]$ for which $x \succeq t^{(j)}$, or (iii) $x \in \outlayers(\sqrt{n})$. In all cases $x \notin S_{\pmb{t},\pmb{\phi}}$.
\end{itemize}
Since $S_{\pmb{t},\pmb{\phi}}$ is non-increasing we have $\upset(y) \subset \overline{S_{\pmb{t},\pmb{\phi}}}$ and so by \cref{fact:simplex1} any minimal set of points $X$ such that $y \in \conv(X)$ satisfies $X \subset \overline{S_{\pmb{t},\pmb{\phi}}}$. Thus $S_{\pmb{t},\pmb{\phi}}$ is convex by \cref{fact:convex-characterization}. \end{proof}

\begin{lemma} \label{lem:no} For $S_{\pmb{t},\pmb{r}} \sim \Dno$, we have
$\pr_{\pmb{t},\pmb{r}}[\eps(S_{\pmb{t},\pmb{r}}) \geq \Omega(1)] \geq \Omega(1)$. 
\end{lemma}

We prove \cref{lem:no} in \cref{sec:proof-no}.

\begin{lemma} \label{lem:hard} Given a collection of points $\pmb{x} = (x_1,\ldots,x_s) \in (\zpm^n)^s$ and a set $S \subseteq \zpm^n$, let $(\pmb{x},S(\pmb{x})) := ((x_1,S(x_1)),\ldots,(x_s,S(x_s)))$ denote the corresponding collection of labelled examples. Let $\mathcal{E}_{\mathsf{yes}}$ and $\mathcal{E}_{\mathsf{no}}$ denote the distributions over $(\pmb{x},S(\pmb{x}))$ when $\pmb{x}$ consists of $s$ i.i.d.\ uniform samples and $S \sim \Dyes$ and $S \sim \Dno$, respectively. If $s \leq 3^{\sqrt{n}/3}$, then the total variation distance between $\mathcal{E}_{\texttt{yes}}$ and $\mathcal{E}_{\texttt{no}}$ is $o(1)$. \end{lemma}

We prove \cref{lem:hard} in \cref{sec:proof-hard}.

\subsubsection{Sets Drawn from \texorpdfstring{$\Dno$}{Dno} are Far from Convex\texorpdfstring{: Proof of \cref{lem:no}}{}} \label{sec:proof-no}

\begin{proof} Recall the definition of the set $U$ in \cref{eq:buckets}. We prove \cref{lem:no} by showing that with constant probability over the terms $\pmb{t}$ and the random function $\pmb{r} \colon U \to \{0,1\}$, there exists a collection $L$ of $\Omega(3^n)$ disjoint co-linear triples $(x,y,z)$ such that $x,z \in S_{\pmb{t},\pmb{r}}$, $y \notin S_{\pmb{t},\pmb{r}}$, and $y = \frac{1}{2}(x+z)$. The existence of such a set implies that $\eps(S_{\pmb{t},\pmb{r}}) \geq \frac{1}{3}|L| \cdot 3^{-n} = \Omega(1)$ since the membership of at least one point from each of these triples would need to changed in order to make the set convex.

We first show that there is a large collection $T$ of disjoint co-linear triples lying in $\midlayers(\sqrt{n})$. Then, by \cref{clm:highprob_vertex} and fact that each point in $U$ is included in the set $S_{\pmb{t},\pmb{r}}$ with probability $1/2$, we can argue that with constant probability, a constant fraction of the triples in $T$ will be violations of convexity. 

\begin{claim} \label{clm:triples} There exists a set $T$ of $\Omega(3^n)$ disjoint co-linear triples $(x,y,z)$ such that (a) $x,y,z \in \midlayers(\sqrt{n})$, and (b) $y = \frac{1}{2}(x+z)$. \end{claim}
\begin{proof} Given $z \in \zpm^{n-1}$ and $b \in \zpm$, let $(b,z) \in \zpm^n$ denote the point whose first coordinate is $b$ and the rest of the coordinates are given by $z$. Consider the set of disjoint triples 
\begin{align*} \label{eq:triples}  
    T := \left\{((-1,z),(0,z),(+1,z)) \colon z \in \zpm^{n-1} \text{ such that } \norm{z}_1 \in \left[\frac{2n}{3} - \sqrt{n}, \frac{2n}{3} + \sqrt{n} - 1\right]\right\} \text{.}
\end{align*}
Observe that every triple $(x,y,z)$ is contained in $\midlayers(\sqrt{n})$ and clearly $y = \frac{1}{2}(x+z)$. We use the following fact to lower bound $|T|$. This fact follows from an application of Stirling's approximation.
\begin{fact} For any $N$ and $\ell \in [-O(\sqrt{N}),O(\sqrt{N})]$, we have ${N \choose \frac{2N}{3} + \ell} = \Theta\left(\frac{1}{\sqrt{N}} \cdot \frac{3^N}{2^{2N/3}+\ell}\right)$. \end{fact}
By the above fact, 
\[
|T| = \sum_{\ell = -\sqrt{n}}^{\sqrt{n}-1} {n-1 \choose \frac{2n}{3} + \ell} 2^{2n/3+\ell} = \sum_{\ell = -\sqrt{n}}^{\sqrt{n}-1} \Omega\left(\frac{1}{\sqrt{n-1}} \cdot \frac{3^{n-1}}{2^{2(n-1)/3 + \ell}}\right) 2^{2n/3 + \ell} = \Omega(3^n)
\]
and this completes the proof of the claim. \end{proof}

Let $T$ denote the set of $\Omega(3^n)$ disjoint co-linear triples in $\midlayers(\sqrt{n})$ given by \cref{clm:triples}. We will need the following claim which shows that triples in $T$ are contained in $U$ with constant probability. 

\begin{claim} \label{clm:highprob_vertex} For any $(x,y,z) \in T$, we have $\pr_{\pmb{t}}[x,y,z \in U] \geq \frac{1}{1,000,000}$. \end{claim}

\begin{proof} By definition of $T$ we have $x,y,z \in \midlayers(\sqrt{n})$ and $x_1 = +1$, $y_1 = 0$, $z_1 = -1$ and $x_j = y_j = z_j$ for all $j \in [2,n]$. Recall the distribution over the terms $\bt = (t^{(1)},\ldots,t^{(N)})$ defined in \cref{sec:distributions}. Note that $\pr_{t^{(i)}}[t^{(i)} \preceq y] = (\frac{\norm{y}_1}{2n})^{\sqrt{n}}$ since $t^{(i)} \preceq y$ if and only if each of the $\sqrt{n}$ non-zero coordinates $a$ of $t^{(i)}$ (a) is chosen as one of the non-zero coordinates of $y$ which happens with probability $\norm{y}/n$ and (b) $t^{(i)}_a$ is set to $y_a$ which happens with probability $1/2$. Also note that for a term $t^{(i)}$, we have $t^{(i)} \prec y$ implies $t^{(i)} \prec x,z$. Therefore,
\begin{align} \label{eq:term-mass}
    \pr_{\pmb{t}}[x,y,z \in U_i] &= \pr_{t^{(i)}}[t^{(i)} \preceq x,y,z] \cdot \prod_{j \neq i} \pr_{t^{(j)}}[t^{(j)} \not\preceq x,y,z] \nonumber \\
    &= \pr_{t^{(i)}}[t^{(i)} \preceq y] \cdot \prod_{j \neq i} \pr_{t^{(j)}}[t^{(j)} \not\preceq x,z] \nonumber \\
    &= \pr_{t^{(i)}}[t^{(i)} \preceq y] \cdot \prod_{j \neq i} \left(1 - \pr_{t^{(j)}}[(t^{(i)} \preceq x) \vee (t^{(i)} \preceq z)]\right)\text{.}
\end{align}
The first term is lower bounded by
\begin{align} \label{eq:lower-bound}
    \pr_{t^{(i)}}[t^{(i)} \preceq y] = \left(\frac{\norm{y}_1}{2n}\right)^{\sqrt{n}} \geq \left(\frac{\frac{2n}{3}-\sqrt{n}}{2n}\right)^{\sqrt{n}} = \frac{1}{3^{\sqrt{n}}}\left(1 - \frac{3}{2\sqrt{n}}\right)^{\sqrt{n}} \geq \frac{e^{-3/2-o(1)}}{N} \geq \frac{1}{5N} \text{.}
\end{align}
To lower bound the second term, observe that 
\begin{align} \label{eq:upper-bound1}
    \pr_{t^{(i)}}[t^{(i)} \preceq x] = \left(\frac{\norm{x}_1}{2n}\right)^{\sqrt{n}} \leq \left(\frac{\frac{2n}{3}+\sqrt{n}}{2n}\right)^{\sqrt{n}} = \frac{1}{3^{\sqrt{n}}}\left(1 + \frac{3}{2\sqrt{n}}\right)^{\sqrt{n}} \leq \frac{e^{3/2}}{N} \leq \frac{5}{N}
\end{align}
and the same bound holds for the point $z$. Therefore, by a union bound $\pr_{t^{(i)}}[(t^{(i)} \preceq x) \vee (t^{(i)} \preceq z)] \leq 10/N$. Plugging this bound along with \cref{eq:lower-bound} into \cref{eq:term-mass} and summing over all $i \in [N]$ yields
\begin{align*}
    \pr_{\pmb{t}}[x,y,z \in U] &= \sum_{i=1}^N \pr_{t^{(i)}}[t^{(i)} \preceq y] \cdot \prod_{j \neq i} \left(1 - \pr_{t^{(j)}}[(t^{(i)} \preceq x) \vee (t^{(i)} \preceq z)]\right) \geq N \cdot \frac{1}{5N} \cdot \left(1 - \frac{10}{N}\right)^{N}
\end{align*}
which is at least $\frac{1}{1,000,000}$ and this completes the proof. \end{proof}

Now, for a set $S_{\pmb{t},\pmb{r}} \sim \Dno$, let 
\[
T_{\mathsf{viol}} = \{(x,y,z) \colon x,z \in S_{\pmb{t},\pmb{r}} \text{ and } y \notin S_{\pmb{t},\pmb{r}}\}
\]
denote the set of triples in $T$ that are violations of convexity for $S_{\pmb{t},\pmb{r}}$. By definition of $\Dno$ and using \cref{clm:highprob_vertex}, for any fixed $(x,y,z) \in T$, we have 
\begin{align*}
    \pr[(x,y,z) \in T_{\mathsf{viol}}] = \pr_{\pmb{t}}[x,y,z \in U] \cdot \pr_{\pmb{r}}[\pmb{r}(x) = \pmb{r}(z) = 1, \pmb{r}(y) = 0 ~|~ x,y,z \in U] \geq \frac{1}{8,000,000}.
\end{align*} 
Therefore, $\Exp_{\pmb{t},\pmb{r}}[|T\setminus T_{\mathsf{viol}}|] \leq |T|(1-\frac{1}{8,000,000})$ and so by Markov's inequality
\begin{align*}
    \pr_{\pmb{t},\pmb{r}}\left[|T_{\mathsf{viol}}| \leq \frac{|T|}{8,000,000^2}\right] &\leq \pr_{\pmb{t},\pmb{r}}\left[|T \setminus T_{\mathsf{viol}}| \geq |T|\left(1-\frac{1}{8,000,000^2}\right)\right] \\
    &= \pr_{\pmb{t},\pmb{r}}\left[|T \setminus T_{\mathsf{viol}}| \geq |T|\left(1-\frac{1}{8,000,000}\right)\left(1+\frac{1}{8,000,000}\right)\right] \\
    &\leq \pr_{\pmb{t},\pmb{r}}\left[|T \setminus T_{\mathsf{viol}}| \geq \Exp_{\pmb{t},\pmb{r}}[|T\setminus T_{\mathsf{viol}}|]\left(1+\frac{1}{8,000,000}\right)\right] \\
    &\leq \frac{1}{1+\frac{1}{8,000,000}} = 1 - \frac{1}{8,000,001} \text{.}
\end{align*}
Finally, since $|T| = \Omega(3^n)$, this gives us
\begin{align*}
    \pr_{\pmb{t},\pmb{r}}\left[\eps(S_{\pmb{t},\pmb{r}}) \geq \Omega(1)\right] \geq \pr_{\pmb{t},\pmb{r}}\left[|T_{\mathsf{viol}}| \geq \frac{|T|}{8,000,000^2}\right] \geq \frac{1}{8,000,001}
\end{align*}
and this completes the proof of \cref{lem:no}. \end{proof}

\subsubsection{\texorpdfstring{$\Dyes$ and $\Dno$}{Dyes and Dno} are Hard to Distinguish\texorpdfstring{: Proof of \cref{lem:hard}}{}} \label{sec:proof-hard}

\begin{proof} Recall the definition of the set $U_{i}$ in \cref{eq:buckets}. For $a \neq b \in [s]$, let $E_{ab}$ denote the event that $x_a$ and $x_b$ belong to the same $U_{i}$ for some $i \in [N]$. Observe that conditioned on $\overline{\vee_{a,b}E_{ab}}$, the distributions $\mathcal{E}_{\mathsf{yes}}$ and $\mathcal{E}_{\mathsf{no}}$ are identical. 

Let $x,y \in \zpm^n$ denote two independent uniform samples. We have
\begin{align} \label{eq:asd}
    \pr[E_{ab}] = \pr_{x,y,\pmb{t}}\left[\bigvee_{i=1}^n (x \in U_{i} \wedge y \in U_{i})\right] = \sum_{i=1}^n \pr_{x,y,\pmb{t}} \left[x \in U_{i} \wedge y \in U_{i}\right] = \sum_{i=1}^N \pr_{y,\pmb{t}}[y \in U_{i}]^2 
\end{align}
where the second equality holds since the $U_{i}$'s are disjoint and the third equality holds by independence of $x$ and $y$. Now, for a fixed $i \in [N]$, if $y \notin \midlayers(\sqrt{n})$ observe that $\pr_{\pmb{t}}[y \in U_{i}] = 0$ and if $y \in \midlayers(\sqrt{n})$, we have $\pr_{\pmb{t}}[y \in U_i] \leq 5/N$ by \cref{eq:upper-bound1}. Thus, $\pr_{y,\pmb{t}}[y \in U_{i}] \leq 5/N$ and combining this with \cref{eq:asd} yields $\pr[E_{ab}] \leq 25/N$. Finally, by a union bound, we have
\[
d_{\text{TV}}(\mathcal{E}_{\mathsf{yes}},\mathcal{E}_{\mathsf{no}}) \leq \pr_{\pmb{x},\pmb{t}}\left[\bigvee_{a \neq b \in [s]} E_{ab}\right] \leq s^2 \cdot \frac{25}{N} = o(1)
\]
since $N = 3^{\sqrt{n}} = \omega(s^2)$. \end{proof}

\section{One-Sided Error Testing}
\label{sec:UB_nonadaptive_1sided}

\subsection{Non-Adaptive Upper Bound}
We complete the proof of \cref{thm:UB_nonadaptive_1sided} in this section.
The upper bound on the query complexity for testing convexity non-adaptively with one-sided error is achieved by \cref{alg:tester}. (As a reminder, the notions of upward shadow $\upset(\by)$ and middle layers $\midlayers(\ell)$ in the algorithm are introduced in \cref{def:upset} and \cref{eq:inn_mid_out}, respectively.)

\begin{algorithm}
\caption{Convexity tester for sets in $\zpm^n$.} \label{alg:tester}
\textbf{Input:} A set $\bS \subseteq \zpm^n$ and a parameter $\eps \in (0,1)$.

\noindent Set $\ell := \sqrt{2n \ln 8/\eps}$ and repeat $\frac{4}{\eps}$ times:
\begin{enumerate}
    \item Query $\by \in \zpm^n$ uniformly at random.
    \item If $\by \in \overline{S} \cap \midlayers(\ell)$, then query all points in $\upset(\by) \cap \midlayers(\ell)$.
    \item If there exists $\bX \subseteq \bS \cap \upset(\by) \cap \midlayers(\ell)$ such that $\by \in \conv(\bX)$, then reject.
\end{enumerate}
Accept.
\end{algorithm}

The analysis of \cref{alg:tester} relies on the following lemma regarding sets that are far from convex. 

\begin{lemma} \label{lem:farfromconvex} Let $\bS \subseteq \zpm^n$ and $\eps \leq \eps(\bS)$. If $\ell = \sqrt{2n \ln 8/\eps}$, then 
\[
|\conv\big(\bS \cap \midlayers(\ell)\big) \cap \big(\overline{\bS} \cap \midlayers(\ell)\big)| \geq \frac{\eps}{2} \cdot 3^n \text{.}
\]
\end{lemma}
In words, there are at least $\frac{\eps}{2} \cdot 3^n$ points in the middle layers $\midlayers(\ell)$ that are not in $\bS$ but that lie in the convex hull of the portion of $\bS$ in the middle layers. 
\begin{proof} Let $\bT := \conv(\bS \cap \midlayers(\ell)) \cap \zpm^n$. Clearly, $\bT$ is convex and so 
\begin{align*}
    \eps(\bS) \cdot 3^n \leq \Delta(\bS,\bT) = |\bT \cap \overline{\bS}| + |\overline{\bT} \cap \bS| \text{.}
\end{align*}
Now observe that $\bS \cap \midlayers(\ell) \subseteq \bT$ and so $\overline{\bT} \cap \bS \subseteq \overline{\midlayers(\ell)}$. Thus, $|\overline{\bT} \cap \bS| \leq |\overline{\midlayers(\ell)}|$. Next, we have
\[
|\bT \cap \overline{\bS}| = |\bT \cap (\overline{\bS} \cap \midlayers(\ell))| + |\bT \cap (\overline{\bS} \cap \overline{\midlayers(\ell)})| \leq |\bT \cap (\overline{\bS} \cap \midlayers(\ell))| + |\overline{\midlayers(\ell)}| \text{.}
\]
By \cref{fact:chernoff}, $|\overline{\midlayers(\ell)}| \leq 2\exp(-\ln (8/\eps)) \cdot 3^n = \frac{\eps}{4} \cdot 3^n$. Therefore, combining the above yields
\begin{align*}
   |\bT \cap (\overline{\bS} \cap \midlayers(\ell))| \geq \eps(\bS) \cdot 3^n - 2|\overline{\midlayers(\ell)}| \geq \left(\eps(\bS) - \frac{\eps}{2}\right) \cdot 3^n \geq \frac{\eps(\bS)}{2} \cdot 3^n
\end{align*}
where the last step holds since $\eps \leq \eps(\bS)$. \end{proof}

We now prove the correctness of \cref{alg:tester}. The tester always accepts when $\bS$ is convex, since in this
case $\conv(\bS \cap \midlayers(\ell)) \subseteq \bS$. Now suppose $\eps(\bS) \geq \eps$. If $\by \in
\conv(\bS \cap \midlayers(\ell)) \cap (\overline{\bS} \cap \midlayers(\ell))$, then there exists some $\bX \subseteq \bS \cap \midlayers(\ell)$ such that $(\bX,\by)$ is a minimal violating pair. Crucially, \cref{fact:simplex1} guarantees that $\bX \subseteq \upset(\by)$. Thus, if the tester picks such a $\by$ in step (1), then it is guaranteed to reject $\bS$ since step (2) queries all points in $\upset(\by) \cap \midlayers(\ell)$. Therefore, using \cref{lem:farfromconvex}, the probability that the tester rejects $\bS$ after one iteration of steps 1-3 is at least 
\[
\Pru{\by \in \zpm^n}{\by \in \conv(\bS \cap \midlayers(\ell)) \cap (\overline{\bS} \cap \midlayers(\ell))} \geq \eps/2 \text{.}
\]
Thus, the tester rejects $\bS$ with probability at least $1-(1-\eps/2)^{4/\eps} \geq 2/3$. 

We now bound the number of queries. I.e., we need to bound the size of $\upset(\by) \cap \midlayers(\ell)$ when $\by \in \midlayers(\ell)$. Note that each point in this set can be obtained by choosing a set of $2\ell$ coordinates where $\by$ has a zero, and then flipping each of these coordinates to a value in $\{0,\pm 1\}$. Therefore, when $\by \in \midlayers(\ell)$, we have
\begin{align*}
    |\upset(\by) \cap \midlayers(\ell)| \leq \binom{n}{2\ell} \cdot 3^{2\ell} \leq n^{4\ell} = n^{\sqrt{32 n \ln 8/\eps}}
\end{align*}
and so the total number of queries made by the tester is at most $\frac{4}{\eps} \cdot n^{\sqrt{32 n \ln 8/\eps}}$. 
This completes the proof of \cref{thm:UB_nonadaptive_1sided}.

\subsection{Non-Adaptive Lower Bound}
\label{section:non-adaptive-lower-bound}
\label{sec:LB_nonadaptive_1sided}

We complete the proof of \cref{thm:LB_nonadaptive_1sided} establishing the lower bound on the query complexity of non-adaptive convexity testers with one-sided error in this section.
The starting point for the lower bound is the notion of an \emph{anti-slab} in $\zpm^n$.


\begin{definition} [Slab] \label{def:slab} Fix $\tau \geq 1$ and $\bv \in \{0,\pm 1\}^n$. The $(\tau,\bv)$-\emph{slab} is defined as
\[
\slab_{\tau,\bv} = \left\{\bx \in \zpm^n \colon \left|\langle \bv,\bx\rangle\right| \leq \tau \right\} \text{.}
\]
We refer to $\overline{\slab_{\tau,\bv}}$ as the $(\tau,\bv)$-\emph{anti-slab}.
\end{definition}

Note that a slab is an intersection of two  parallel halfspaces and so an anti-slab is a union of two parallel and disjoint halfspaces. Anti-slabs are clearly non-convex, and the following claim establishes an important property of any certificate of non-convexity for the anti-slab. In particular, it shows that if a set of queries contains a witness of non-convexity for the $(\tau,\bv)$-anti-slab, then it must contain two points $\bx \in \overline{\slab_{\tau,\bv}}$ and $\by \in \slab_{\tau,\bv}$ whose projections onto $\bv$ are separated by at least distance $\tau$.

\begin{claim} [The Structure of Violating Pairs for Anti-slabs] \label{claim:only_distant_certificates} Suppose $(\bX,\by)$ is a violating pair for the $(\tau,\bv)$-anti-slab, $\overline{\slab_{\tau,\bv}}$. Then there exists a point $\bx \in \bX$ for which $|\langle \bv, \bx - \by \rangle| > \tau$. \end{claim}

\begin{proof} We have $\by \in \conv(\bX)$ and so $\sum_{\bx \in \bX} \lambda_{\bx} \bx = \by$ where $\sum_{\bx \in \bX} \lambda_{\bx} = 1$. Moreover, we have $\by \in \slab_{\tau,\bv}$ and so 
\begin{align} \label{eq:inner_prod}
    \sum_{\bx \in \bX} \lambda_{\bx} \langle \bv,\bx \rangle = \left\langle \bv, \sum_{\bx \in \bX}\lambda_{\bx} \bx \right\rangle = \langle \bv,\by\rangle \in [-\tau,\tau] \text{.}
\end{align}
We also have $X \subseteq \overline{\slab_{\tau,\bv}}$, which implies $\left|\langle \bv,\bx \rangle\right| > \tau$ for all $\bx \in \bX$. Therefore, by \cref{eq:inner_prod} there clearly has to be some $\bx \in \bX$ where $\langle \bv,\bx \rangle$ is positive and some $\bx' \in \bX$ where $\langle \bv,\bx' \rangle$ is negative, for otherwise the LHS would be outside the interval $[-\tau,\tau]$. In particular, this implies $\langle \bv,\bx \rangle > \tau$ and $\langle \bv,\bx' \rangle < -\tau$ and so
\[
\langle \bv, \bx' \rangle < -\tau \leq \langle \bv,\by \rangle \leq \tau < \langle \bv, \bx \rangle \text{.}
\]
Thus, if $\langle \bv,\by \rangle \leq 0$, then $|\langle \bv, \bx - \by \rangle| > \tau$, and if $\langle \bv,\by \rangle \geq 0$, then $|\langle \bv, \bx' - \by \rangle| > \tau$. \end{proof}

We now introduce our hard family of sets: truncated anti-slabs. (As a reminder, the sets $\inlayers(t)$ and $\outlayers(t)$ are defined in \cref{eq:inn_mid_out}.)

\begin{definition} [Truncated Anti-slab] \label{def:tas} Fix $\tau \geq 1$, $\bv \in \{0,\pm 1\}^n$, and $t \geq 1$. The $t$-truncated $(\tau,\bv)$-anti-slab is defined as follows:
\[
\tas_{\tau,\bv,t} = \left(\overline{\slab_{\tau,\bv}} \cup \inlayers(t)\right) \setminus \outlayers(t) \text{.}
\]
\end{definition}

\begin{figure}[t] \label{fig:TAS}
	\begin{center}
		\includegraphics[trim = 0 0 0 0, clip, scale=0.175]{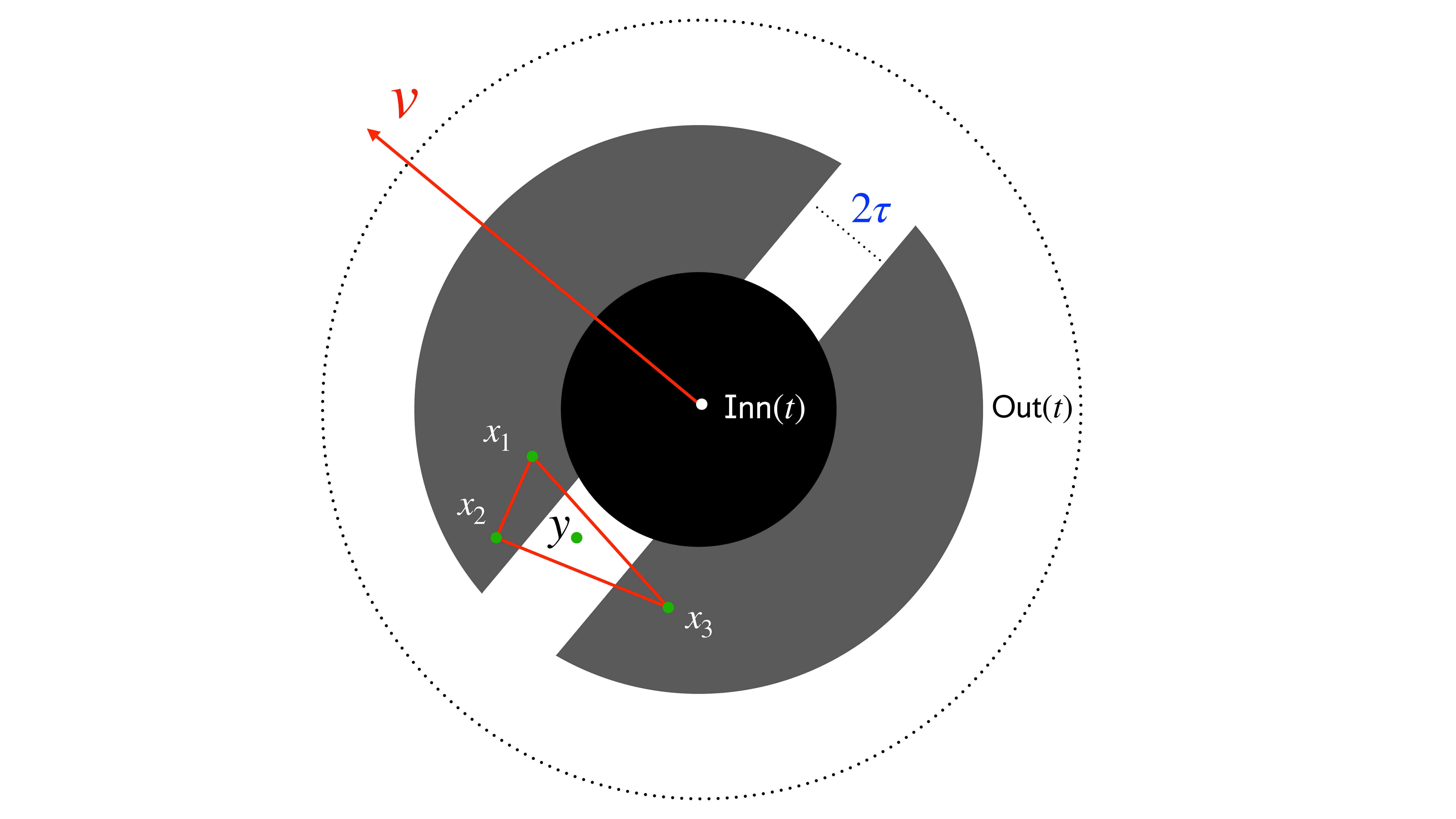}
		\caption{An illustration of the $t$-truncated $(\tau,v)$-anti-slab. The dotted circle represents $\{\pm 1\}^n$ and everything within it is $\zpm^n$. The dark shaded regions are $\tas_{\bv}$. The pair $(\{\bx_1,\bx_2,\bx_3\},\by)$ is a violation for the set.}
	\end{center}
\end{figure}
In particular, we fix $\tau = \sqrt{n}$, $t = 0.7\sqrt{n}$, and consider vectors $\bv \in \zpm^n$ for which $\norm{\bv}_1 = n/2$. Thus, henceforth we will drop the subscripts $\tau,t$ and abbreviate $\tas_{\bv} := \tas_{\sqrt{n}, \bv, 0.7\sqrt{n}}$.

In other words, $\tas_{\bv}$ is the set obtained by taking the $(\sqrt{n},\bv)$-anti-slab, adding in all points with fewer than $\frac{2}{3}n-0.7\sqrt{n}$ non-zero entries, and removing all points with more than $\frac{2}{3}n+0.7\sqrt{n}$ non-zero entries. The intuition for why these sets are hard to test (for non-adaptive testers with one-sided error) is as follows. Suppose a one-sided error tester $T$ has queried a set $\bQ \subset \zpm^n$ and rejects $\tas_{\bv}$. By \cref{cor:witness}, $\bQ$ must contain a minimal violating pair $(\bX,\by)$ for $\tas_{\bv}$. Note that $\bX \subset \tas_{\bv}$, $\by \notin \tas_{\bv}$, and $\by \prec \bx$ for all $\bx \in \bX$ by \cref{fact:simplex1}. By \cref{claim:only_distant_certificates}, there is some $\bx \in \bX$ such that $|\langle \bv, \bx - \by \rangle| > \sqrt{n}$. Additionally, by the truncation, we have $\bx \notin \outlayers(0.7\sqrt{n})$ and $\by \notin \inlayers(0.7\sqrt{n})$. Since $y \prec x$, this implies $\norm{\bx - \by}_1 \leq 1.4\sqrt{n}$. In summary, for $T$ to reject $\tas_{\bv}$ after querying $\bQ$, there must be some $\by \prec \bx \in \bQ$ for which $|\langle \bv, \bx - \by \rangle| > \sqrt{n}$, but also $\norm{\bx - \by}_1 \leq 1.4\sqrt{n}$. 

We consider the family of sets $F = \{\tas_{\bv} \colon \norm{\bv}_1 = n/2\}$. By the above argument the lower bound boils down to the following question: given $y \prec x$ such that $\norm{x-y}_1 \leq 1.4\sqrt{n}$, how many vectors $v \in \zpm^n$ (with $\norm{v}_1 = n/2$) exist for which $|\langle v, x-y \rangle| > \sqrt{n}$? We show that this number is upper bounded by $|F| \cdot \exp(-\Omega(\sqrt{n}))$ and so, by a union bound, the number of sets in $F$ that $T$ can reject after querying $\bQ$ is bounded by $|\bQ|^2 \cdot |F| \cdot \exp(-\Omega(\sqrt{n}))$. Therefore, for $T$ to be a valid non-adaptive tester with one-sided error, we must have $|\bQ|^2 = \exp(\Omega(\sqrt{n}))$ and this gives the result. This argument is formalized in \cref{sec:LB_proof}.

Of course, for the above argument to prove \cref{thm:LB_nonadaptive_1sided}, we need to show that truncated anti-slabs are $\Omega(1)$-far from convex.

\begin{lemma} \label{lem:farfromconvex1} Consider $\bv \in \zpm^n$ where $\norm{\bv}_2^2 = n/2$. We have $\dist(\tas_{\bv}, \mathbf{convex}) = \Omega(1)$. \end{lemma}

The above \cref{lem:farfromconvex1} is a corollary of the following \cref{lem:farfromconvex2}.

\begin{restatable}{lemma}{lemfarfromconvex} \label{lem:farfromconvex2} Consider $\bv \in \zpm^n$ where $\norm{\bv}_1 = n/2$. There exists a set $L \subset (\zpm^n)^3$ of $\Omega(3^n)$ disjoint colinear triples such that for every $(\bx,\by,\bz) \in L$ the following hold.
\begin{enumerate}
    \item $\by = \frac{\bx+\bz}{2}$ and $\by \in \slab_{\sqrt{n},\bv}$, $\bx,\bz \in \overline{\slab_{\sqrt{n},\bv}}$.
    \item $\bx,\by,\bz \in \midlayers(0.7\sqrt{n})$.
\end{enumerate}
\end{restatable}

In \cref{sec:LB_proof} we prove \cref{thm:LB_nonadaptive_1sided} using \cref{claim:only_distant_certificates} and \cref{lem:farfromconvex1}. In \cref{sec:farfromconvex} we prove \cref{lem:farfromconvex2}, which immediately implies \cref{lem:farfromconvex1}.

\subsubsection{Proof of the Lower Bound} \label{sec:LB_proof} Recall the definition of $\inlayers(t)$, $\midlayers(t)$, and $\outlayers(t)$ in \cref{eq:inn_mid_out}. Given $\bv \in \zpm^n$, recall that
\begin{align*}
    \tas_{\bv} &= \left(\overline{\slab_{\sqrt{n},\bv}} \cup \inlayers(0.7\sqrt{n})\right) \setminus \outlayers(0.7\sqrt{n}) 
\end{align*}
is the $0.7\sqrt{n}$-truncated $(\sqrt{n},\bv)$-anti-slab (\cref{def:tas}). Let $\bV$ denote the set of all vectors 
 $\bv \in \zpm^n$ where $\norm{\bv}_2^2 = n/2$. By \cref{lem:farfromconvex1}, we have $\dist(\tas_{\bv},\textbf{convex}) = \Omega(1)$ for all $\bv \in \bV$. Also note that $|\bV| = \binom{n}{n/2} \cdot 2^{n/2} = 2^{3n/2}/\Theta(\sqrt{n})$.

Given $\bx,\by \in \zpm^n$, let $\Delta(\bx,\by) = \{i \in [n] \colon \bx_i \neq \by_i\}$. For $\bv \in \zpm^n$, let $\mathsf{NZ}_{\bv} = \{i \colon \bv_i \neq 0\}$. Let $T$ be a one-sided error, non-adaptive tester for convex sets in $\zpm^n$.

\begin{claim} \label{clm:certificate} Fix $\bv \in \bV$ and suppose that $T$ rejects $\tas_{\bv}$ after querying a set $\bQ \subseteq \zpm^n$. Then there exists $\bx \neq \by \in \bQ$ such that (a) $|\Delta(\bx,\by)| \leq 1.4\sqrt{n}$ and (b) $|\Delta(\bx,\by) \cap \mathsf{NZ}_{\bv}| > \sqrt{n}$. \end{claim}

\begin{proof} By \cref{cor:witness}, $\bQ$ must contain a minimal violating pair $(\bX,\by)$ for $\tas_{\bv}$. By \cref{claim:only_distant_certificates}, there exists $\bx \in \bX$ for which $|\langle \by - \bx, \bv\rangle| > \sqrt{n}$. Observe that $|\Delta(\bx,\by) \cap \mathsf{NZ}_{\bv}| \geq |\langle \by - \bx, \bv\rangle|$ and so (b) holds. 

Now, since $(\bX,\by)$ is a violating pair we have $\bx \in \tas_{\bv}$ and $\by \notin \tas_{\bv}$ and since $(\bX,\by)$ is minimal, \cref{fact:simplex1} implies that $\by \prec \bx$. By construction, we have $\inlayers(0.7\sqrt{n}) \subseteq \tas_{\bv}$ and $\outlayers(0.7\sqrt{n}) \subseteq \overline{\tas_{\bv}}$ and so it must be the case that $\bx,\by \in \midlayers(0.7\sqrt{n})$. In summary, $\bx$ can be obtained by changing at most $1.4\sqrt{n}$ zero values in $\by$ to non-zero values. Thus, (a) holds. \end{proof}

Now, given $\bx,\by \in \zpm^n$, let
\begin{align}
    \bV(\bx,\by) = \Big\{\bv \in \bV \colon |\Delta(\bx,\by) \cap \mathsf{NZ}_{\bv}| > \sqrt{n}\Big\} \subseteq \bV \text{.}
\end{align}

\begin{claim} \label{clm:Vxy_bound} If $|\Delta(\bx,\by)| \leq 1.4\sqrt{n}$, then $|\bV(\bx,\by)| \leq O(|\bV|) \cdot 2^{-0.08\sqrt{n}}$. \end{claim}

\begin{proof} Note that $|\mathsf{NZ}_{\bv}| = n/2$ for all $\bv \in \bV$ and so $|\bV(\bx,\by)|$ can be bounded as follows. In the following calculation, a vector $\bv \in \bV(\bx,\by)$ is chosen by picking $\ell > \sqrt{n}$ coordinates from $\Delta(\bx,\by)$ and $n/2-\ell$ coordinates from $[n] \setminus \Delta(\bx,\by)$ to be non-zero. Then each of these coordinates is fixed to a value in $\{\pm 1\}$.
\begin{align} \label{eq:lbeq}
    |\bV(\bx,\by)| &= \sum_{\ell > \sqrt{n}} \binom{|\Delta(\bx,\by)|}{\ell} \binom{n-|\Delta(\bx,\by)|}{n/2-\ell}  2^{n/2} \nonumber \\
    &\leq 2^{n/2} \binom{n - |\Delta(\bx,\by)|}{n/2 - \sqrt{n}}\sum_{\ell > \sqrt{n}} \binom{|\Delta(\bx,\by)|}{\ell} \nonumber \\
    &= \frac{2^{3n/2}}{\Theta(\sqrt{n})} \cdot 2^{-|\Delta(\bx,\by)|} \sum_{\ell > \sqrt{n}} \binom{|\Delta(\bx,\by)|}{\ell} = O(|\bV|) \cdot 2^{-|\Delta(\bx,\by)|} \sum_{\ell>\sqrt{n}} \binom{|\Delta(\bx,\by)|}{\ell}
\end{align}
To bound the RHS, observe that $2^{-k} \cdot \sum_{\ell > \sqrt{n}} \binom{k}{\ell}$ is precisely the probability that a random subset $S \subseteq [k]$ has size $|S| > \sqrt{n}$, which is a monotone increasing function of $k$. Thus, the RHS of \cref{eq:lbeq} is a monotone increasing function of $|\Delta(\bx,\by)|$ and so is maximized by setting $\Delta(\bx,\by) = 1.4\sqrt{n}$. Thus,
\begin{align*}
    |\bV(\bx,\by)| &\leq O(|\bV|) \cdot 2^{-1.4\sqrt{n}} \sum_{\ell > \sqrt{n}} \binom{1.4\sqrt{n}}{\ell} \\
    &\leq O(|\bV|) \cdot 2^{-1.4\sqrt{n}} \cdot \sqrt{n} \cdot \binom{1.4\sqrt{n}}{\sqrt{n}}
    \leq O(|\bV|) \cdot 2^{-0.08\sqrt{n}}\text{.}
\end{align*}
The last inequality holds by the well known bound ${m \choose k} \leq \left(\frac{e m}{k}\right)^k$ as follows. We have ${1.4\sqrt{n} \choose \sqrt{n}} = {1.4\sqrt{n} \choose 0.4\sqrt{n}} \leq (\frac{e \cdot 1.4}{0.4})^{0.4\sqrt{n}} = 2^{0.4\log_2(1.4e/0.4)\sqrt{n}} < 2^{1.31\sqrt{n}}$. \end{proof}

Now, given a set of queries $\bQ \subseteq \zpm^n$, let
\begin{align}
    \bV(\bQ) = \Big\{ \bv \in \bV \colon \exists \bx \neq \by \in \bQ \text{ such that } |\Delta(\bx,\by) \cap \mathsf{NZ}_{\bv}| > \sqrt{n} \text{ and } |\Delta(\bx,\by)| \leq 1.4\sqrt{n} \Big\} \subseteq \bV \text{.}
\end{align}
By \cref{clm:certificate}, if $T$ rejects $\tas_{\bv}$ after querying the set $\bQ$, then $\bv \in \bV(\bQ)$. Informally, $\bV(\bQ)$ contains all $\bv$ for which $\bQ$ can contain a witness of non-convexity for the set $\tas_{\bv}$. Moreover, by \cref{clm:Vxy_bound} and the union bound, we have
\begin{align} \label{eq:TQ_bound}
    \left|\bV(\bQ)\right| \leq \sum_{\bx,\by \in \bQ} \left|\bV(\bx,\by)\right| \leq |\bQ|^2 \cdot O(|\bV|) \cdot 2^{-0.08\sqrt{n}} \text{.}
\end{align}
Now, let $\bQ$ be the set of $q$ queries sampled according to the distribution defined by the non-adaptive, one-sided error tester $T$. Then, using linearity of expectation and the bound from \cref{eq:TQ_bound} we obtain
\begin{align*}
    &\sum_{\bv \in \bV} \Pru{\bQ}{T \text{ rejects } \tas_{\bv} \text{ after querying } \bQ} \leq
\sum_{\bv\in \bV} \Pru{\bQ}{\bv \in \bV(\bQ)} = \Exu{\bQ}{|\bV(\bQ)|} \leq q^2 \cdot O(|\bV|) \cdot 2^{-0.08\sqrt{n}}
\end{align*}
and therefore, by averaging over $\bV$, there exists $\bv \in \bV$ such that
\begin{align}
    \frac{2}{3} \leq \Pru{\bQ}{T \text{ rejects } \tas_{\bv} \text{ after querying } \bQ} \leq O(1) \cdot q^2 \cdot 2^{-0.08\sqrt{n}}
\end{align}
where the first inequality is due to the fact that $T$ rejects any $\bS_{\bv}$ with probability at least $2/3$. Therefore, it follows that $q \geq \Omega(1) \cdot 2^{0.04\sqrt{n}}$.

\subsubsection{Truncated Anti-slabs are Far from Convex} \label{sec:farfromconvex}

We complete the proof of \cref{lem:farfromconvex2} in this section, restated below for ease of reading.

\lemfarfromconvex*

\begin{proof} Let $J = \{j \in [n] \colon \bv_j \neq 0\}$. Without loss of generality, by a rotation, we may assume that $\bv_j = 1$ for all $j \in J$. 
Note that under this assumption, we have $\langle \bv,\bx \rangle = \sum_{j \in J} \bx_j$ for all $\bx \in \zpm^n$.

To construct our set $L$ of disjoint colinear triples we start by constructing a \emph{matching} of $\Omega(3^n)$ pairs $(\bx,\by)$ such that (a) $\by \in \slab_{\sqrt{n},\bv} \cap \midlayers(0.7\sqrt{n})$, (b) $\bx \in \overline{\slab_{\sqrt{n},\bv}} \cap \midlayers(0.7\sqrt{n})$, and (c) $\by$ can be obtained from $\bx$ by changing a subset of $\bx$'s $+1$ coordinates in $J$ to $0$. A third point $\bz$ is obtained by reflecting $\bx$ across $\by$, i.e., this same set of coordinates is changed to $-1$ to obtain $\bz$. By symmetry we have $\norm{\bz}_1 = \norm{\bx}_1$, $(\bx,\by,\bz)$ are colinear, and the resulting set of triples are disjoint. We also choose the original matching so that we will always have $\bz \in \overline{\slab_{\sqrt{n},\bv}} \cap \midlayers(0.7\sqrt{n})$ and so the resulting triple satisfies item (1) and (2) of the lemma, i.e., it is a violation of convexity for the $0.7\sqrt{n}$-truncated $(\sqrt{n},\bv)$-anti-slab, $\tas_{\sqrt{n},\bv,0.7\sqrt{n}}$.

To construct our matching we use the following simple claim.

\begin{claim} \label{claim:matching} Let $(U,V,E)$ be a bipartite graph and $\Delta > 0$ be such that (a) each vertex $x \in U$ has degree exactly $\Delta$ and (b) each vertex $y \in V$ has degree at least $\Delta$. Then there exists a matching $M \subseteq E$ in $(U,V,E)$ of size $|M| \geq (1-1/e) |V|$. \end{claim}

\begin{proof} We construct a random map $\phi \colon U \to V$ as follows. For each $x \in U$ let $\phi(x)$ be a uniform random neighbor of $x$. Observe that $\phi^{-1}(y) \cap \phi^{-1}(y') = \emptyset$ for all $y \neq y' \in V$. Thus, given $\phi$, we can obtain a matching $M_{\phi}$ as follows: for each $y \in V$, if $\phi^{-1}(y) \neq \emptyset$, then add $(x,y)$ to $M_{\phi}$ for some arbitrary $x \in \phi^{-1}(y)$. To lower bound the size of $M_{\phi}$, observe that
\[
\Exp_{\phi}\left[|M_{\phi}|\right] = |V| - \Exp_{\phi}\left[\sum_{y \in V} \mathbf{1}(\phi^{-1}(y) =
\emptyset)\right] = |V| - \sum_{y \in V} \Pru{\phi}{\phi^{-1}(y) = \emptyset}\text{.}
\]
Now, if $\phi^{-1}(y) = \emptyset$, this means that all $\deg(y) \geq \Delta$ neighbors of $y$ were mapped to some neighbor other than $y$, of which there are exactly $\Delta$ in total. Therefore,
\[
\Pru{\phi}{\phi^{-1}(y) = \emptyset} = \left(1-\frac{1}{\Delta}\right)^{\deg(y)} \leq 1/e
\]
since $\deg(y) \geq \Delta$. Thus, $\Exp_{\phi}\left[|M_{\phi}|\right] \geq |V| \cdot (1 - 1/e)$ and so there exists a matching $M$ satisfying the claim. \end{proof}

Given $\bx \in \zpm^n$ and $b \in \zpm$, let $|\bx|_{b,J} = |\{j \in J \colon \bx_j = b\}|$ and similarly for $\overline{J}$. Let $I = \left[n/6 + 0.6\sqrt{n}, n/6 + 0.8\sqrt{n}\right]$. We define the following sets.
\begin{align} \label{eq:X}
    X = \left\{ \bx \in \{0,\pm 1\}^n \colon \sqrt{n} < \sum_{j \in J} \bx_j < 1.2\sqrt{n} \text{, } |\bx|_{1,J} \geq |\bx|_{0,J} + 1.1\sqrt{n} \text{, and } |\bx|_{0,\overline{J}} \in I\right\} 
\end{align}
\begin{align} \label{eq:Y}
    Y = \left\{ \by \in \{0,\pm 1\}^n \colon -0.1\sqrt{n} < \sum_{j \in J} \by_j < 0.1\sqrt{n} \text{, } |\by|_{1,J} \geq |\by|_{0,J} - 1.1\sqrt{n} \text{, and } |\by|_{0,\overline{J}} \in I\right\}
\end{align}
Observe that $X \subset \overline{\slab_{\sqrt{n},\bv}}$ and $Y \subset \slab_{\sqrt{n},\bv}$. We now partition $X$ and $Y$ as follows. For each $\ell \in \NN$, let
\begin{align} \label{eq:XYell}
    X_{\ell} = \left\{\bx \in X \colon |\bx|_{0,J} = \ell \right\} \text{ and } 
    Y_{\ell} = \left\{\by \in Y \colon |\by|_{0,J} = \ell + 1.1\sqrt{n} \right\}\text{.}
\end{align}
For each such $\ell$ we consider the bipartite graph $(Y_{\ell},X_{\ell},E_{\ell})$ where there is an edge $(\by,\bx) \in E_{\ell}$ if $\bx$ can be obtained from $\by$ by choosing a set of $1.1\sqrt{n}$ coordinates from $J$ where $\by$ has a $0$ and flipping all of these bits to $+1$. Formally, $(\by,\bx) \in E$ iff $\exists A \subseteq J$ of size $|A| = 1.1\sqrt{n}$ such that (a) for all $j \in A$, $\by_j = 0$, $\bx_j = +1$, and (b) for all $j \in [n] \setminus A$, $\by_j = \bx_j$. Observe now that (a) every vertex in $Y_{\ell}$ has degree exactly $\Delta := {\ell + 1.1\sqrt{n} \choose 1.1\sqrt{n}}$, and (b) each vertex $\bx \in X_{\ell}$ has degree
\[
\text{deg}(\bx) = {|\bx|_{1,J} \choose 1.1\sqrt{n}} \geq {|\bx|_{0,J} + 1.1\sqrt{n} \choose 1.1\sqrt{n}} = {\ell + 1.1\sqrt{n} \choose 1.1\sqrt{n}} =  \Delta
\] 
where the inequality is by definition of $X$ and the second to last equality is by definition of $X_{\ell}$. Thus, by \cref{claim:matching}, there exists a matching $M_{\ell}$ in $(Y_{\ell},X_{\ell},E_{\ell})$ of size $|M_{\ell}| \geq \Omega(|X_{\ell}|)$.

Now, we obtain a set of disjoint colinear triples by taking
\[
L = \left\{(\bx,\by,2\by-\bx) \colon (\by,\bx) \in \bigcup_{\ell = \frac{n}{6} - 1.3\sqrt{n}}^{\frac{n}{6} - 1.2\sqrt{n}} M_{\ell}\right\} \text{.}
\]
Note that by construction every $(\bx,\by,\bz) \in L$ is a colinear triple in $\zpm^n$.

\paragraph{Proof of items (1) and (2) of \cref{lem:farfromconvex2}:}

By definition of the sets $X$ and $Y$, we have $\bx \in \overline{\slab_{\sqrt{n},\bv}}$ and $\by \in \slab_{\sqrt{n},\bv}$. Note that $\bz$ is obtained from $\bx$ by flipping a set of $1.1\sqrt{n}$ coordinates in $J$ where $\bx$ is $+1$ to $-1$. Therefore, we have $\bz \in \overline{\slab_{\sqrt{n},\bv}}$ since 
\[
\sum_{j\in J} \bz_j = \sum_{j \in J} \bx_j - 2.2\sqrt{n} < 1.2\sqrt{n} - 2.2\sqrt{n} = -\sqrt{n}
\]
where the inequality used the definition of the set $X$. Thus, item (1) of the lemma is satisfied.

Now, for every $(\bx,\by,\bz) \in L$, we have
\[
\frac{n}{6} - 1.3\sqrt{n} \leq |\bx|_{0,J} = |\bz|_{0,J} = |\by|_{0,J} - 1.1\sqrt{n} \leq \frac{n}{6} - 1.2\sqrt{n}
\]
and so 
\begin{align*}
    |\bx|_{0,J},|\bz|_{0,J},|\by|_{0,J} \in \left[\frac{n}{6} - 1.3\sqrt{n}, \frac{n}{6} - 0.1\sqrt{n}\right] \text{.}
\end{align*}
Now, recalling that $I = [n/6+0.6\sqrt{n},n/6+0.8\sqrt{n}]$ and the the definition of $X$ and $Y$, we have
\begin{align*}
    |\by|_{0,\overline{J}}, |\bz|_{0,\overline{J}}, |\bx|_{0,\overline{J}} \in \left[\frac{n}{6} + 0.6\sqrt{n}, \frac{n}{6} + 0.8\sqrt{n}\right] \text{.}
\end{align*}
Combining the two bounds above we get that the number of $0$-coordinates of $\bx,\by$, and $\bz$ are all in the range $[n/3-0.7\sqrt{n},n/3+0.7\sqrt{n}]$. Therefore, we have $\norm{\bx}_1,\norm{\by}_1,\norm{\bz}_1 \in [2n/3-0.7\sqrt{n},2n/3+0.7\sqrt{n}]$, i.e., item (2) of the lemma is satisfied. 

\paragraph{Proof that $|L| \geq \Omega(3^n)$:} It remains to lower bound the size of $L$. Towards this, recall that 
\begin{align} \label{eq:L_bound}
    |L| = \sum_{r = 1.2\sqrt{n}}^{1.3\sqrt{n}} |M_{n/6-r}| = \sum_{r = 1.2\sqrt{n}}^{1.3\sqrt{n}} \Omega(|X_{n/6-r}|)\text{.}
\end{align}
We use the following claim to simplify our calculation of $|X_{n/6-r}|$.
\begin{claim} \label{claim:subset} If $|\bx|_{0,J} < n/6 - 1.2\sqrt{n}$ and $\sum_{j \in J} \bx_j > \sqrt{n}$, then $|\bx|_{1,J} \geq |\bx|_{0,J} + 1.1\sqrt{n}$. \end{claim}

\begin{proof} Note that $|\bx|_{1,J} - |\bx|_{-1,J} = \sum_{j \in J} \bx_j > \sqrt{n}$ and 
\[
|\bx|_{1,J} + |\bx|_{-1,J} = n/2 - |\bx|_{0,J} > n/3 + 1.2\sqrt{n} > 2|\bx|_{0,J} + 3.6\sqrt{n}\text{.}
\]
Adding these inequalities and dividing by $2$ yields $|\bx|_{1,J} > |\bx|_{0,J} + 1.85\sqrt{n} > |\bx|_{0,J} + 1.1\sqrt{n}$. \end{proof}

In particular, recalling the definition of $X$ in \cref{eq:X}, using \cref{claim:subset}, we get that for $\ell \in [n/6-1.3\sqrt{n},n/6-1.2\sqrt{n}]$, we can write
\[
X_{\ell} = \left\{\bx \in \zpm^n \colon \sqrt{n} < \sum_{j \in J} \bx_j < 1.2\sqrt{n} \text{, } |\bx|_{0,J} = \ell \text{ and } |\bx|_{0,\overline{J}} \in I\right\}\text{.}
\]
I.e., the condition $|\bx|_{1,J} \geq |\bx|_{0,J} + 1.1\sqrt{n}$ in the definition of $X$ is not needed to describe $X_{\ell}$ for the values of $\ell$ that we consider.

\begin{claim} \label{clm:X_bound} $\sum_{r = 1.2\sqrt{n}}^{1.3\sqrt{n}} |X_{n/6-r}| = \Omega(3^n)$. \end{claim}

\begin{proof} For simplicity let us assume that $\sqrt{n}$ is an integer. Note that $\sum_{r = 1.2\sqrt{n}}^{1.3\sqrt{n}} |X_{n/6-r}|$ is equal to
\begin{align} \label{eq:count_X'}
   \left(\sum_{0.6\sqrt{n} \leq q \leq 0.8\sqrt{n}} {\frac{n}{2} \choose \frac{n}{3}-q} 2^{\frac{n}{3}-q}\right) \left(\sum_{1.2\sqrt{n} \leq k \leq 1.3\sqrt{n}} {\frac{n}{2} \choose \frac{n}{3} + k} \sum_{0.5\sqrt{n} < s < 0.6\sqrt{n}} {\frac{n}{3} + k \choose \frac{n}{6}+\frac{k}{2} + s}\right) \text{.}
\end{align}
The first term in the product in \cref{eq:count_X'} comes from the fact that the bits in $\overline{J}$ can be set to anything, as long as the number of zero bits is in the interval $I = [n/6+0.6\sqrt{n},n/6+0.8\sqrt{n}]$. Equivalently, the number of non-zero entries is in the interval $[n/3-0.8\sqrt{n},n/3-0.6\sqrt{n}]$.

Now consider the second term. The first sum is over the number of non-zero coordinates in $J$, which is in the interval $[\frac{n}{3}+1.2\sqrt{n},\frac{n}{3}+1.3\sqrt{n}]$. The second sum is over all ways to set the non-zero coordinates in $J$ so that they're sum is in the interval $(\sqrt{n},1.2\sqrt{n})$. Notice that if the number of non-zero coordinates is $\frac{n}{3} + k$, then the sum of the non-zero coordinates is in the interval $(\sqrt{n},1.2\sqrt{n})$ iff the number of $+1$'s is in the interval $(\frac{n}{6} + \frac{k}{2} + 0.5\sqrt{n}, \frac{n}{6} + \frac{k}{2} + 0.6\sqrt{n})$. This explains the second sum in the term.

To bound the RHS of \cref{eq:count_X'}, we use the following fact, which follows readily from Stirling's formula.

\begin{fact} \label{fact:Stirling}Let $N \in \NN$, $t \in \mathbb{Z}$ be such that $|t| \leq c \sqrt{N}$ for some constant $c > 0$. Then,
\[\text{(a) }{N \choose N/2 + t} = \Theta\left(\frac{1}{\sqrt{N}} \cdot 2^N\right) \text{ and (b) } {N \choose 2N/3 + t} = \Theta\left(\frac{1}{\sqrt{N}}\cdot \frac{3^N}{2^{2N/3+t}}\right)\text{.}\] 
\end{fact}

By part (b) of \cref{fact:Stirling} we can bound the first term of \cref{eq:count_X'} as
\begin{align} \label{eq:term1}
    \sum_{0.6\sqrt{n} \leq q \leq 0.8\sqrt{n}} {\frac{n}{2} \choose \frac{n}{3}-q} 2^{\frac{n}{3}-q} = \sum_{0.6\sqrt{n} \leq q \leq 0.8\sqrt{n}} \Omega\left(\frac{1}{\sqrt{n}} \cdot 3^{n/2}\right) = \Omega(3^{n/2})\text{.}
\end{align}
For the second term in \cref{eq:count_X'}, we have $k,s = \Theta(\sqrt{n})$.  
Thus, by part (a) of \cref{fact:Stirling} we have 
\[
{\frac{n}{3} + k \choose \frac{n}{6} + \frac{k}{2} + s} \geq \Omega\left(\frac{1}{\sqrt{n}} \cdot 2^{\frac{n}{3}+k}\right) ~\Longrightarrow~
    \sum_{0.5\sqrt{n} < s < 0.6\sqrt{n}} {\frac{n}{3} + k \choose \frac{n}{6}+\frac{k}{2} + s} \geq \Omega\left(2^{\frac{n}{3}+k}\right)
\]
and so the second term of \cref{eq:count_X'} is
\begin{align} \label{eq:term2}
    &\sum_{1.2\sqrt{n} \leq k \leq 1.3\sqrt{n}} {\frac{n}{2} \choose \frac{n}{3} + k} \sum_{0.5\sqrt{n} < s < 0.6\sqrt{n}} {\frac{n}{3} + k \choose \frac{n}{6}+\frac{k}{2} + s} &\geq \sum_{1.2\sqrt{n} \leq k \leq 1.3\sqrt{n}} {\frac{n}{2} \choose \frac{n}{3} + k} \cdot \Omega\left(2^{\frac{n}{3}+k}\right) 
\end{align}
which is at least $\Omega(3^{n/2})$ by part (b) of \cref{fact:Stirling} since $k = \Theta(\sqrt{n})$. 

To summarize, the LHS in the claim statement is equal to the quantity in \cref{eq:count_X'}, which is a product of two terms, each of which is at least $\Omega(3^{n/2})$ (\Cref{eq:term1} and \cref{eq:term2}). Therefore, $\sum_{r = 1.2\sqrt{n}}^{1.3\sqrt{n}} |X_{n/6-r}| = \Omega(3^n)$ as claimed. \end{proof}

Plugging the bound from \cref{clm:X_bound} into \cref{eq:L_bound} completes the proof of \cref{lem:farfromconvex2}. \end{proof}

\subsection{Sample-Based Lower Bound}
\label{sec:LB_samples_1sided}

There is an upper bound of $O(n3^n)$ samples required for exactly learning any set $\bS \subseteq \zpm^n$, due to the coupon-collector argument, and therefore there is an upper bound of $3^{O(n)}$ on one-sided error testing of convex sets with samples. For large enough $\eps > 0$, there is a slightly improved bound of $O(3^n \cdot \tfrac{1}{\eps}\log(1/\eps))$ for one-sided sample-based testers for any property of sets on $\{0,\pm 1\}^n$ (even in the distribution-free setting where the distribution over $\zpm^n$ is arbitrary and unknown to the algorithm), due to the general upper bound of $O(\textsc{VC}(\mathcal H) \cdot \tfrac{1}{\eps}\log(1/\eps))$ on one-sided sample-based testing, where $\textsc{VC}(\mathcal H)$ is the VC dimension of the property $\mathcal H$ \cite{BFH21}. We show that the exponent $O(n)$ is optimal for one-sided sample-based testers.

\thmsamplesonesided*

\begin{proof}
It suffices to prove the lower bound, due to the discussion above.
Suppose that $T$ is a one-sided sample-based tester and let $\bQ \subseteq \zpm^n$ denote a random set of $s$ samples made by $T$. If $T$ is given a non-convex set $\bS \subseteq \zpm^n$, then it must reject $\bS$ with probability at least $2/3$. Moreover, by \cref{cor:witness}, for $T$ to reject $\bS$ it must be that $\bQ$ contains a minimal violating pair $(\bX,\by)$ for $\bS$ and, by \cref{fact:simplex1}, $\bX \subseteq \upset(\by)$. Thus, in particular, there must exist two points $\bx,\by \in \bQ$ such that $\bx \in \upset(\by)$. Thus, by the union bound over all pairs in $\bQ$, we have
\begin{align} \label{eq:sample_up_bound1}
    2/3 \leq \Pru{\bQ}{T \text{ rejects } \bS} \leq \Pru{\bQ}{\exists \bx,\by \in \bQ \colon \bx \in \upset(\by)} \leq s^2 \cdot \Pru{\bx,\by \in \zpm^n}{\bx \in \upset(\by)}
\end{align}
To compute this probability, notice that 
\[
\bx \in \upset(\by) \text{ if and only if } \forall i \in [n] \colon (\by_i = 0) \vee (\bx_i = \by_i = 1) \vee (\bx_i = \by_i = -1)
\]
and so 
\begin{align} \label{eq:sample_up_bound2}
\Pru{\bx,\by \sim \{0,\pm 1\}^n}{\bx \in \upset(\by)} = (5/9)^n\text{.}
\end{align}
Thus, combining \cref{eq:sample_up_bound1} and \cref{eq:sample_up_bound2}, we have $s \geq \sqrt{\frac{2}{3}(\frac{9}{5})^{n}} = 3^{\Omega(n)}$. \end{proof}

\bibliographystyle{alpha}
\bibliography{biblio}

\appendix

\section{Approximating Binomial Coefficients Near the Middle} \label{sec:approx-bin-coeff}

In this section we prove the following approximation of the binomial coefficient ${n \choose \frac{n-\tau}{2}}$.

\begin{theorem} \label{thm:binomial_approx1} For all $0 \leq \tau \leq n(1-\Omega(1))$, we have 
\begin{align*}
{n \choose \frac{n-\tau}{2}} = \frac{2^n \cdot \sqrt{\frac{2n}{\pi(n-\tau)(n+\tau)}}}{\exp\left(\tau \cdot \sum_{k=1}^{\infty} \left(\frac{\tau}{n}\right)^{2k-1}\left(\frac{1}{2k-1}-\frac{1}{2k}\right) + \Theta\left(\frac{1}{n}\right)\right)} \text{.}
\end{align*}
\end{theorem}

\Cref{thm:binomial_approx1} implies the following corollary.

\begin{corollary} \label{cor:middle_binomial_bound} For every constant integer $s \ge 1$, when $\tau = O(n^{1-\frac{1}{2s}})$ then 
\[
{n \choose \frac{n-\tau}{2}} = \Theta\left(\frac{2^n}{\sqrt{n}} \cdot \exp\left(-\sum_{k=1}^{s-1}\frac{\tau^{2k}}{n^{2k-1}}\Big(\frac{1}{2k-1}-\frac{1}{2k}\Big)\right)\right) \text{.}
\]
\end{corollary}

\begin{proof} Since $\tau = o(n)$, the square-root term in the numerator of \cref{thm:binomial_approx1} becomes $\Theta(1/\sqrt{n})$. Since $\tau = O(n^{1-1/2s})$ we have 
\begin{align*}
\sum_{k=s}^{\infty} \frac{\tau^{2k}}{n^{2k-1}}\left(\frac{1}{2k-1}-\frac{1}{2k}\right) &\leq \sum_{k=s}^{\infty} n^{1-k/s}\left(\frac{1}{2k-1}-\frac{1}{2k}\right) \\
&= \sum_{k=0}^{\infty} \left(\frac{1}{n^{1/s}}\right)^k\left(\frac{1}{2(k+s)-1}-\frac{1}{2(k+s)}\right) = O(1) \text{.}
\end{align*}
I.e., the infinite summation converges when one ignores the first $s-1$ terms. \end{proof}

\subsection{Proof of \texorpdfstring{\cref{thm:binomial_approx1}}{Theorem~\ref{thm:binomial_approx1}}} \label{sec:binomial_approx1_proof}

We use the following standard identities and approximations in the proof.

\begin{fact} \label{fact:e-approx} If $N > 1$, then
\begin{align*}
    \left(1+\frac{1}{N}\right)^N = \exp\left(1 - \sum_{k=1}^{\infty}\frac{(-1)^{k+1}}{N^k(k+1)}\right) \text{ and } \left(1-\frac{1}{N}\right)^N = \exp\left(-1 - \sum_{k=1}^{\infty}\frac{1}{N^k(k+1)}\right) \text{.}
\end{align*}
\end{fact}

\begin{proof} Consider the Taylor series expansions $\ln (1+x) = \sum_{k=1}^{\infty} \frac{(-1)^{k-1}x^k}{k}$ and $\ln (1-x) = -\sum_{k=1}^{\infty} \frac{x^k}{k}$ at $x \in (0,1)$. Since $N > 1$, we have $1/N \in (0,1)$. Thus,
\begin{align*}
    \left(1+\frac{1}{N}\right)^N = \exp\left(N \ln \left(1+\frac{1}{N}\right)\right) = \exp\left(N \sum_{k=1}^{\infty}\frac{(-1)^{k-1}}{N^k k}\right) = \exp\left(1 - \sum_{k=1}^{\infty}\frac{(-1)^{k+1}}{N^k (k+1)}\right)
\end{align*}
and 
\begin{align*}
    \left(1-\frac{1}{N}\right)^N = \exp\left(N \ln \left(1-\frac{1}{N}\right)\right) = \exp\left(-N \sum_{k=1}^{\infty}\frac{1}{N^k k}\right) = \exp\left(-1 - \sum_{k=1}^{\infty}\frac{1}{N^k (k+1)}\right)\text{.}
\end{align*}
\end{proof}

\begin{fact}[Stirling's Approximation] 
\label{fact:binomial_estimate}
For all $n \geq 1$
\[
\sqrt{2\pi n} \cdot \left(\frac{n}{e}\right)^n e^{\frac{1}{12n+1}} \leq n! \leq \sqrt{2\pi n} \cdot \left(\frac{n}{e}\right)^n e^{\frac{1}{12n}} \text{.}
\]  
Therefore, for all $n 
\geq 1$ and $1 \leq k \leq n-1$, 
\[
{n \choose k} = \sqrt{\frac{n}{2\pi k (n-k)}}\left(\frac{n}{k}\right)^k\left(\frac{n}{n-k}\right)^{n-k} \cdot \exp\left(\frac{1}{12n + \Theta(1)} - \frac{1}{12k + \Theta(1)} - \frac{1}{12(n-k) + \Theta(1)}\right)\text{.}
\]
\end{fact}

We are now ready to complete the proof of the approximation for binomial coefficients.

\begin{proof}[Proof of \cref{thm:binomial_approx1}]
Letting $k = \frac{n-\tau}{2}$ and substituting this value in \cref{fact:binomial_estimate} yields
\begin{align} \label{eq:stirling-approx}
    {n \choose \frac{n-\tau}{2}} &=  \sqrt{\frac{2n}{\pi(n-\tau)(n+\tau)}} \cdot \left(\frac{n}{\frac{n-\tau}{2}}\right)^{\frac{n-\tau}{2}}\left(\frac{n}{\frac{n+\tau}{2}}\right)^{\frac{n+\tau}{2}} \cdot \exp(\Theta(1/n)) \nonumber \\
    &= 2^n \sqrt{\frac{2n}{\pi(n-\tau)(n+\tau)}} \cdot \left(\frac{1}{1-\frac{\tau}{n}}\right)^{\frac{n-\tau}{2}} \left(\frac{1}{1+\frac{\tau}{n}}\right)^{\frac{n+\tau}{2}} \cdot \exp(\Theta(1/n)) \text{.}
\end{align}
We now apply \cref{fact:e-approx} and get the following bounds:
\begin{align} \label{eq:-term}
    \left(1-\frac{\tau}{n}\right)^{\frac{1}{2}(n-\tau)} &= \left(1-\frac{\tau}{n}\right)^{\frac{n}{\tau} \cdot \frac{\tau}{2n}(n-\tau)} \nonumber \\
    &= \exp\left(\left(-1-\sum_{k=1}^{\infty} \left(\frac{\tau}{n}\right)^k \frac{1}{k+1}\right) \cdot \left(\frac{\tau}{2} - \frac{\tau^2}{2n}\right)\right) \nonumber \\
    &= \exp\left(-\frac{\tau}{2}+\frac{\tau^2}{2n}\right)\exp\left(\left(\frac{\tau^2}{2n}-\frac{\tau}{2}\right)\sum_{k=1}^{\infty} \left(\frac{\tau}{n}\right)^k \frac{1}{k+1}\right)
\end{align}
and
\begin{align} \label{eq:+term}
    \left(1+\frac{\tau}{n}\right)^{\frac{1}{2}(n+\tau)} &= \left(1+\frac{\tau}{n}\right)^{\frac{n}{\tau} \cdot \frac{\tau}{2n}(n+\tau)} \nonumber \\
    &= \exp\left(\left(1-\sum_{k=1}^{\infty} \left(\frac{\tau}{n}\right)^k \frac{(-1)^{k+1}}{k+1}\right) \cdot \left(\frac{\tau}{2} + \frac{\tau^2}{2n}\right)\right) \nonumber \\
    &= \exp\left(\frac{\tau}{2}+\frac{\tau^2}{2n}\right)\exp\left(\left(-\frac{\tau^2}{2n}-\frac{\tau}{2}\right)\sum_{k=1}^{\infty} \left(\frac{\tau}{n}\right)^k \frac{(-1)^{k+1}}{k+1}\right) \text{.}
\end{align}
Taking the product of \cref{eq:-term} and \cref{eq:+term} gives 
\begin{align}
    &\left(1-\frac{\tau}{n}\right)^{\frac{1}{2}(n-\tau)}\left(1+\frac{\tau}{n}\right)^{\frac{1}{2}(n+\tau)} \nonumber \\
    &= \exp\left(\frac{\tau^2}{n}\right)\exp\left(\frac{\tau^2}{2n}\sum_{k = 1}^{\infty} \left(\frac{\tau}{n}\right)^k \frac{1 - (-1)^{k+1}}{k+1}\right) \exp\left(-\frac{\tau}{2}\sum_{k = 1}^{\infty} \left(\frac{\tau}{n}\right)^k \frac{1+(-1)^{k+1}}{k+1}\right) \text{.} \nonumber
\end{align}
In the first sum, the $k$th term cancels when $k$ is odd and doubles when $k$ is even. In the second sum, the $k$th term cancels when $k$ is even and doubles when $k$ is odd. Thus,
\begin{align}
    \left(1-\frac{\tau}{n}\right)^{\frac{1}{2}(n-\tau)}\left(1+\frac{\tau}{n}\right)^{\frac{1}{2}(n+\tau)} &= \exp\left(\frac{\tau^2}{n} + \left(\frac{\tau^2}{n}\sum_{k = 2 \text{, even}}^{\infty} \left(\frac{\tau}{n}\right)^k \frac{1}{k+1}\right) - \left(\tau\sum_{k = 1\text{, odd}}^{\infty} \left(\frac{\tau}{n}\right)^k \frac{1}{k+1}\right)\right) \nonumber \\ 
    &= \exp\left(\frac{\tau^2}{n} + \left(\tau\sum_{k = 2 \text{, even}}^{\infty} \left(\frac{\tau}{n}\right)^{k+1} \frac{1}{k+1}\right) - \left(\tau \sum_{k = 2\text{, even}}^{\infty} \left(\frac{\tau}{n}\right)^{k-1} \frac{1}{k}\right)\right) \nonumber \\ 
    &= \exp\left(\frac{\tau^2}{n} + \left(\tau\sum_{k = 1}^{\infty} \left(\frac{\tau}{n}\right)^{2k+1} \frac{1}{2k+1}\right) - \left(\tau \sum_{k = 1}^{\infty} \left(\frac{\tau}{n}\right)^{2k-1} \frac{1}{2k}\right)  \right) \nonumber \\
    &= \exp\left(\left(\tau\sum_{k = 1}^{\infty} \left(\frac{\tau}{n}\right)^{2k-1} \frac{1}{2k-1}\right) - \left(\tau \sum_{k = 1}^{\infty} \left(\frac{\tau}{n}\right)^{2k-1} \frac{1}{2k}\right)  \right) \nonumber \\
    &= \exp\left(\tau \sum_{k=1}^{\infty} \left(\frac{\tau}{n}\right)^{2k-1}\left(\frac{1}{2k-1}-\frac{1}{2k}\right)\right)
\end{align}
and plugging this quantity back into \cref{eq:stirling-approx} completes the proof of \cref{thm:binomial_approx1}. \end{proof}

\section{Testing by Learning} \label{sec:learning->testing}

We require a testing-by-learning reduction that is slightly nonstandard, because the learning algorithm
is not proper, i.e. for a hypothesis class $\cF$ it does not necessarily output a function in $\cF$.
Given a domain $\cX$ and two functions $f,g \colon \cX \to \{\pm 1\}$, let $\dist(f,g) = \pr_{x \sim \cX}[h(x) \neq f(x)]$. The following lemma is specialized to the uniform distribution, which is sufficient for our purposes. The observation is not new; it appears in \cite{CFSS17} and possibly in other places.

\begin{lemma} \label{lem:learning->testing} Let $\cX$ be any finite domain and let $\cF \colon \cX \to \{\pm 1\}$ be a class of functions. Suppose that for every $\epsilon \in (0,1)$ there exists a learning algorithm for $\cF$, under the uniform distribution on $\cX$, using $q(\epsilon)$ queries and $m(\epsilon)$ samples. Then for every $\epsilon \in (0,1)$ there is an $\epsilon$-tester for $\cF$, under the uniform distribution, using $q(\epsilon/4)$ queries and $m(\epsilon/4) + O(1/\epsilon)$ samples.
\end{lemma}
\begin{proof}[Proof sketch]
Since $\cX$ is finite, we may transform any learning algorithm with error $\epsilon/4$ into a proper learning algorithm with error $\epsilon/2$ as follows. On input $f \in \cF$, let $h : \cX \to \{\pm 1\}$ be the output from the learning algorithm, which (with probability at least $2/3$) satisfies $\dist(f,h) \leq \epsilon/4$. Then since $\cX$ is finite, we can find $g \in \cF$ minimizing $\dist(g,h)$ without using queries or samples. This satisfies $\dist(g,h) \leq \dist(f,h)$ by definition, so we have $\dist(g,f) \leq \dist(g,h) + \dist(f,h) \leq  2\dist(f,h) \leq \epsilon/2$, as desired.
\end{proof}

\end{document}